\documentclass[letterpaper,11pt]{article}
\usepackage[margin=1in]{geometry}

\usepackage[switch, modulo]{lineno}

\usepackage[T1]{fontenc}
\usepackage{fnpct}
\usepackage{microtype}
\usepackage{paralist}
\usepackage[ttscale=.85]{libertine}
\usepackage{libertinust1math}
\usepackage{sectsty}
\usepackage[size=tiny]{todonotes}
\usepackage[normalem]{ulem}
\usepackage{hyperref}

\newif\ifdraft
\drafttrue

\usepackage[maxbibnames=99,sortcites]{biblatex}
\usepackage{framed}
\usepackage{tcolorbox}

\usepackage{enumitem}

\usepackage{blindtext}

\newcommand{\E}[1]{\mathbf{E}\left[#1\right]}
\newcommand{\Exp}[1]{\mathbb{E}\left[#1\right]}
\newcommand{\Prob}[1]{\mathbf{P}\left(#1\right)}
\newcommand{\Pro}[1]{\mathbb{P}\left(#1\right)}

\newcommand{\poly}{\operatorname{\text{{\rm poly}}}}

\newcommand{\RCT}{\ensuremath{\mathrm{RCT}}\xspace}

\newcommand{\RCTID}{\ensuremath{\mathrm{RCTID}}\xspace}
\newcommand{\RCTDEG}{\ensuremath{\mathrm{RCTDEG}}\xspace}

\newcommand{\eps}{\varepsilon}
\newcommand{\whp}{w.h.p.\xspace}
\newcommand{\wehp}{w.e.h.p.\xspace}
\newcommand{\aas}{a.a.s.\xspace}

\newcommand{\dd}{\,\mathrm d}

\newcommand{\B}{\mathcal B}
\newcommand{\Hb}{\ensuremath{\mathbb H^2}\xspace}
\newcommand{\bH}{\Hb}
\newcommand{\disk}{\ensuremath{\mathcal D_R}\xspace}
\newcommand{\G}{\ensuremath{\mathcal{G}(n, \alpha, C)}\xspace}

\newcommand{\hrg}{\mathcal{G}(n, \alpha, C)}

\newcommand{\minexp}{\zeta_{min}}
\newcommand{\maxexp}{\zeta_{max}}
\newcommand{\maxUncoul}[1]{\Delta_{#1}}
\newcommand{\colourSpace}[2]{X_{#1}(#2)}

\newcommand{\finallayer}{\ell^{*}}
\newcommand{\startlayer}{\Bar{\ell}}

\newcommand{\bigO}{\mathcal{O}}

\DeclareMathOperator{\polylog}{polylog}
\DeclareMathOperator{\dist}{d_h}

\newcommand{\deguncoul}[2]{\deg_{#1}\left(#2\right)}
\newcommand{\deghigh}[1]{\deg^+{(#1)}}

\newcommand{\indicator}[1]{\mathbf{1}_{\{#1\}}}

\newcommand{\colour}{\psi}
\newcommand{\palette}{\Psi}
\newcommand{\colourpal}[2]{\Psi_{#1}\left(#2\right)}
\newcommand{\colourchosen}[2]{\psi_{#1}\left(#2\right)}

\newcommand{\intensity}{\lambda}
\newcommand{\func}{\rho}

\usepackage[capitalize,noabbrev]{cleveref}
\usepackage{color}
\usepackage{xspace}
\usepackage{tikz}
\usepackage[framemethod=tikz]{mdframed}
\mdfsetup{frametitlealignment=\centering}

\usepackage{uniinput} 

\usepackage{graphicx,tipa}
\usepackage{caption}
\usepackage{subcaption}

\usepackage{amsmath}
\usepackage{amssymb}

\usepackage{amsthm}
\usepackage{thmtools}
\usepackage{thm-restate}

\declaretheorem{theorem}
\newtheorem*{theorem*}{Theorem}
\declaretheorem{lemma, proposition, corollary}[
style=plain,
sibling=theorem
]

\declaretheoremstyle[headfont=\normalfont,bodyfont=\itshape,]{claimstyle} 
\declaretheorem{claim}[
style=claimstyle,
sibling=theorem,
numbered=no
]

\makeatletter 
\DeclareFontFamily{U}{tipa}{}
\DeclareFontShape{U}{tipa}{m}{n}{<->tipa10}{}
\newcommand{\arc@char}{{\usefont{U}{tipa}{m}{n}\symbol{62}}}%
\newcommand{\arc}[1]{\mathpalette\arc@arc{#1}}
\newcommand{\arc@arc}[2]{%
  \sbox0{$\m@th#1#2$}%
  \vbox{
    \hbox{\resizebox{\wd0}{\height}{\arc@char}}
    \nointerlineskip
    \box0
  }%
}
\makeatother

\newcounter{ctr}
\loop
 \stepcounter{ctr}
 \expandafter\edef\csname c\Alph{ctr}\endcsname{\noexpand\mathcal{\Alph{ctr}}}
\ifnum\thectr<26
\repeat

\usepackage{authblk}
\usepackage{fontawesome5}

\newcommand{\authoremail}[1]{%
  \href{mailto:#1}{\tiny\raisebox{4pt}{\faEnvelope[regular]}}%
}

\author[1]{Yannic Maus\thanks{This research was funded in whole or in part by the Austrian Science Fund (FWF) \url{https://doi.org/10.55776/P36280}, \url{https://doi.org/10.55776/I6915}. For open access purposes, the author has applied a CC BY public copyright license to any author-accepted manuscript version arising from this submission.}\authoremail{yannic.maus@tugraz.at}}
\author[2]{{Janosch Ruff\thanks{This research was partially funded by the German Research Foundation (Deutsche Forschungsgemeinschaft, DFG) – project number 390859508.}}\authoremail{Janosch.Ruff@hpi.de}}

\affil[1]{TU Graz, Austria}
\affil[2]{Hasso Plattner Institute, University of Potsdam, Germany}

\title{On Distributed Colouring of Hyperbolic Random Graphs}
\date{}

\begin{document}
\allsectionsfont{\sffamily}
\maketitle
\thispagestyle{empty}

\begin{abstract}
We analyse the performance of simple distributed colouring algorithms under the assumption that the input graph is a hyperbolic random graph (HRG), a generative model capturing key properties of real-world networks such as power-law degree distributions and large clustering coefficients.  Motivated by the shift from worst-case analysis to more realistic network models, we study the number of rounds and size of the colour space required to colour HRGs in the distributed setting.

\smallskip

In the conceivable simplest algorithm each vertex selects a colour uniformly at random and keeps it permanently if no neighbour tries the same colour; otherwise it discards the candidate colour choice and tries a fresh random colour  in the next round. 


\begin{compactitem}
\item We first show that this  randomised algorithm  terminates in exactly two rounds \aas when given a colour space of $\eps\cdot\Delta$, for any constant $\eps>0$, while failing \whp if the colour space is only $\chi\cdot n^{\delta}$ for some small constant $\delta>0$. 
\item We then consider another classic variation of that random colour trial algorithm  that resolves conflicts by prioritising nodes with smaller IDs. While it also colours with colour space $\eps\cdot\Delta$ in two rounds, it \whp fails to do so with $\Delta/\log^{\Omega(1) }\Delta\gg \chi$ colours. 

\item Lastly, inspired by the structure of HRGs, we consider a variant of the random colour trial algorithm that  prioritises high-degree vertices.  It achieves a valid colouring in two rounds \aas using only $\Delta^{1-\delta}$ colours for some constant $\delta>0$. We also show that our bound on the colour space is asymptotically tight (up to polylogarithmic factors) for certain values of the power-law exponent. 
\end{compactitem}
All three randomised algorithms are extremely simple and run in the bandwidth-restricted CONGEST model. 
These positive results demonstrate that constant-time algorithms can outperform the classical $\Omega(\log^{*}n)$ lower bound for the $\Delta+1$-vertex colouring problem in worst-case graphs, established by [Linial; FOCS '87] and [Naor; SIAM Journal Disc. Math. '91].


Our results rely on several new structural insights into HRGs, which may be of independent interest. More generally, our results contribute to the line of research of simple algorithms beating the general lower bounds on HRGs like $\Omega(n)$ for computing the shortest path~[Bläsius, Freiberger, Friedrich, Katzmann, Montenegro-Retana and Thieffry; Transactions of Algorithms ’24]. 


\end{abstract}

\clearpage
\thispagestyle{empty}
\tableofcontents

\newpage
\setcounter{page}{1}
\section{Introduction}

In this paper, we study vertex colouring algorithms in the classic LOCAL model of distributed computing \cite{Linial-92}: A communication network is modelled as an $n$-node graph where each node represents a computational entity that communicates with its neighbours in synchronous rounds. The time complexity of an algorithm in this model is measured by the number of synchronous communication rounds required until all nodes have produced their outputs. In the \emph{distributed vertex colouring problem} each node has to output a colour from some optimally small set of colours that is different from its neighbours' colours. 

The problem has been extensively studied, e.g., \cite{barenboimelkin_book,BE10,BEG17,FHK16,MT20,SW10,johansson99,CLP18,HSS18,HKNT22,GG24}, and, over the decades, randomised algorithms have improved from $\bigO(\log n)$ rounds~\cite{luby86,alon86} over sublogarithmic-time algorithms~\cite{HSS18} to the current state of the art of $\poly\log\log n$ rounds \cite{RG20,BEPS,CLP18,GG24,HKMT21}, though the resulting algorithms remain technically complex---using various involved graph decompositions, the shattering framework \cite{BEPS}, and even the best randomised algorithms rely on deep insights on distributed derandomisation. Despite this massive progress on faster algorithms, a significant gap to the known lower bound remains. In fact, the only lower bound is a $\Omega(\log^* n)$ lower bound for colouring graphs with maximum degree $2$ that stems from Linial's seminal paper for the deterministic setting and was later shown to hold for randomised algorithms as well by Naor~\cite{Linial-92,Naor91}. In fact, there is no lower bound for $\Delta=\omega(1)$.

In this paper, we ask the question whether one can design better algorithms for real-world graphs.
 Following the line of work using hyperbolic geometry to model complex real-world networks \cite{pkbv-info-2010, kpk-h-10, Bogu2010, Serrano2008}, we use the well-established model of hyperbolic random graphs (HRGs) as introduced by Krioukov, Papadopoulos, Kitsak, Vahdat, and Boguñá~\cite{kpk-h-10}. As HRGs have $\bigO(\log n)$ diameter~\cite{DBLP:conf/analco/KiwiM15, fk-dhrg-18, ms-k-19} they can be trivially coloured with the optimal number of colours by having a designated leader node learn the whole graph topology, brute-forcing the solution locally, and re-distributing it to the nodes of the network. However, such an approach is highly unrealistic, abuses the unbounded message size of the LOCAL model, and offers no algorithmic insight. Surprisingly, we show that one can also colour HRGs with  $(1+o(1))\chi$ colours in constant time. Here, $\chi$ is the chromatic number of the graph which is the smallest number of colours that are required in any valid vertex colouring of the graph. This algorithm does not just exploit the structure of HRGs but it also exploits  the unbounded message size of the LOCAL model in a manner similar to the $\bigO(\text{diameter})$ algorithm, providing very little algorithmic insight. HRGs can also be coloured quickly using existing algorithms~\cite{HN21}, though the algorithms are technically complex while still requiring significantly more colours and time than optimal. 

 
 In this paper, we aim to get the best of all worlds, i.e., a extremely simple and efficient algorithms that colour with few colours significantly outperforming all existing results. 

We use the simplest conceivable randomised distributed colouring algorithm as our baseline approach:

\begin{tcolorbox}
\textbf{Random Colour Trial (RCT):}
Each node randomly picks a uniformly random candidate colour from its palette of available colours, i.e., the whole colour palette without the colours used by its already permanently coloured neighbours, and sends this colour to its neighbours while receiving their colours. If no neighbour tries to get coloured with the same candidate colour, the node permanently adopts the colour, otherwise it discards the candidate colour and retries in the next round with a new random candidate colour from an updated palette. \\
This process repeats until all nodes have a valid colour.
\end{tcolorbox}

It is simple and local---when using $\Delta + 1$ colours, where $\Delta$ is the maximum degree of the graph---it is known to converge in $\bigO(\log n)$ rounds with high probability on any graph~\cite{BEPS,johansson99}. As messages only contain colours, it also runs in the CONGEST model, where communication is restricted to $\bigO(\log n)$-bit messages per edge per round.

\subparagraph{Why $\Delta + 1$ colours?} Vertex colouring with $\Delta+1$ colours is the most studied colouring problem in  distributed computing and many other related modern models of computing like streaming~\cite{ACK19}, massively parallel computation~\cite{CDP21}, local computation algorithms~\cite{FHK16}. This stems from a simple observation: using $\Delta+1$ colours guarantees that every node always has an available colour, regardless of how the colouring unfolds, ensuring termination and correctness.
In distributed settings, it is known that already colouring with a single colour less changes the entire nature of the problem. In order to make fast progress distributed colouring algorithms colour large parts of the network need in parallel, but doing so carelessly, may result in an uncoloured subgraph on which we can even existentially not complete the colouring.  As a result there is an $\Omega(\log_{\Delta}\log n)$ lower bound for colouring low-degree graphs~\cite{BFHKLRSU16}, and algorithms are  more involved and slower~\cite{GHKM18,FHM23,BBN25}. More complex and significantly slower algorithms to colour with fewer colours exists for special graph families such as triangle-free graphs~\cite{PS15}, planar graphs~\cite{planar-graphs1}, or graphs with bounded arboricity \cite{BE10}.

\subsection{Our Contribution.}
We show that the simple \RCT algorithm performs exceptionally well on real-world graphs using significantly fewer colours than the best existing distributed algorithms. In particular, it colours a hyperbolic random graph in just $2$ rounds, \aas\footnote{An event $A$ holds \emph{asymptotically almost surely} (\aas) if $\Pro{A} \in 1 - o (1)$.} Moreover, we show that variants of \RCT achieve a similar performance--some using a number of colours close to the chromatic number--all while remaining extremely simple. For most of these algorithms, we establish nearly tight lower bounds, showing they cannot run faster or use fewer colours than described.


 To make sense of our formal results, we first give some minimal background on hyperbolic random graphs (HRGs). For more details and the formal definition see \Cref{sec:prelims} or the excellent exposition in \cite[§3.3]{katz-diss-23}. Hyperbolic random graphs (HRGs) are a generative graph model where vertices are thrown onto a hyperbolic disk and edges connect a pair of vertices if the pair is close in hyperbolic distance to each other. A power-law degree distribution emerges from this process and it is one of the typical properties of a HRG~\cite{gpp-rhg-12, pkbv-info-2010}. More formally, the  probability that a vertex has degree $k$ is given by $\sim k^{-(2\alpha + 1)}$ where the model parameter $\alpha\in (1/2,1)$ influences the probability distribution on the hyperbolic disk and thus indirectly  the power-law exponent. It turns out that many fundamental graph properties of HRGs like the maximum degree~\cite{gpp-rhg-12}, the spectral gap \cite{hrg-spectral}, the second largest component \cite{km-slcrhg-19}, the tree-width \cite{bfk-tw-2016} or the clique number, chromatic number and degeneracy \cite{bmrs-stacs-25} depend on $\alpha$. We call a graph $G$ a \emph{typical hyperbolic random graph} if $G$ fullfils all properties that hold \aas on a hyperbolic random graph. As our first contribution, we show the following theorem for \RCT:

\begin{restatable}[Random colour trial]{theorem}{mainTheorem}\label{thm:main-theorem}\text{ }
\begin{compactenum}
    \item For any constant $\eps>0$, \RCT executed with $\eps\cdot \Delta$ colours on a typical hyperbolic random graph with maximum degree $\Delta$ terminates after exactly 2 rounds \aas 
    \item There exists a constant $\delta >0$, such that \RCT executed with $n^{\delta}\cdot \chi$ colours on a typical hyperbolic random graph with chromatic number $\chi$ never terminates \wehp\footnote{An event $A$ holds \emph{with extreme high probability} (\wehp) if $\Pro{A} \in 1 - n^{-\omega(1)}$.}
\end{compactenum}
\end{restatable}

\smallskip

\Cref{thm:main-theorem} improves upon using the best existing algorithm for the problem in several ways. It is much simpler, it is faster, and it uses fewer colours. Moreover, we show in \Cref{pro:not-two-rounds-whp} that \RCT cannot colour a typical hyperbolic random graph \whp\footnote{An event $A$ holds \emph{with high probability} (\whp) if $\Pro{A} \in 1 - \bigO(1/n)$.} in two rounds.

\smallskip

\textbf{How fast can we colour in general?}
As mentioned, Linial's seminal $\Omega(\log^* n)$ lower bound for deterministic distributed colouring with $\Delta+1$ colours was extended to randomised algorithms by Naor~\cite{Linial-92,Naor91}. These bounds are often cited as applying to general graphs--but in fact, they only hold on graphs with maximum degree two, i.e., collections of paths and rings. For larger $\Delta$, the best published lower bound is that $\Delta+1$ colouring (and even $\bigO(\Delta)$ colouring) cannot be solved in one round with a deterministic algorithm~\cite{szegedy93,Kuhn2006On,M21}. Interestingly, for sufficiently large $\Delta$--at least some $\poly\log n$--there exist algorithms that achieve $\Delta+1$ colouring in $\bigO(\log^* n)$ rounds \cite{HN21} significantly outperforming the $\poly\log\log n$ time complexity for general graphs~\cite{CLP18}. Whether this is tight, or whether constant-round algorithms might suffice even for $\bigO(\Delta)$ colouring, remains one of the most fundamental open problems in distributed computing. While we do not claim that our result of \Cref{thm:main-theorem} extends to general graphs, it is nevertheless striking that on HRGs, the simplest randomised algorithm achieves proper colouring with few colours in exactly two rounds. Either designing an $\bigO(1)$-round $\Delta+1$-colouring algorithm for all graphs with sufficiently large $\Delta$ or ruling out such an algorithm would be considered a major (and maybe unexpected) breakthrough.

\begin{figure}[t]

\centering
\includegraphics[height=0.23\textheight]{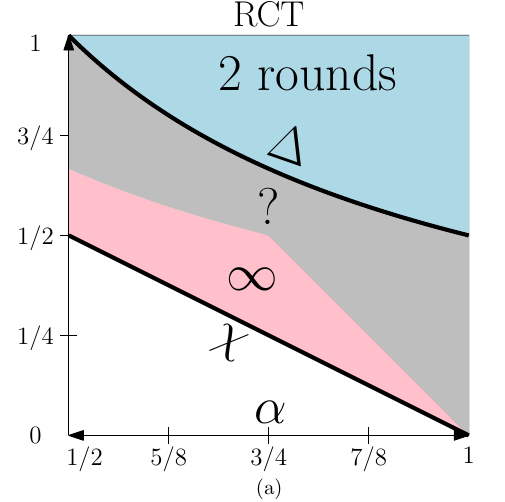}\hfill
\includegraphics[height=0.23\textheight]{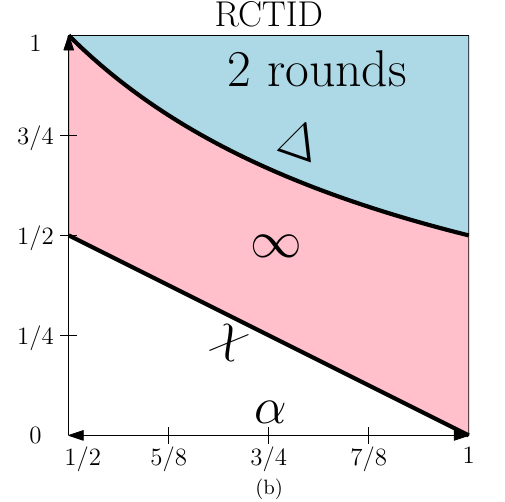}\hfill
\includegraphics[height=0.23\textheight]{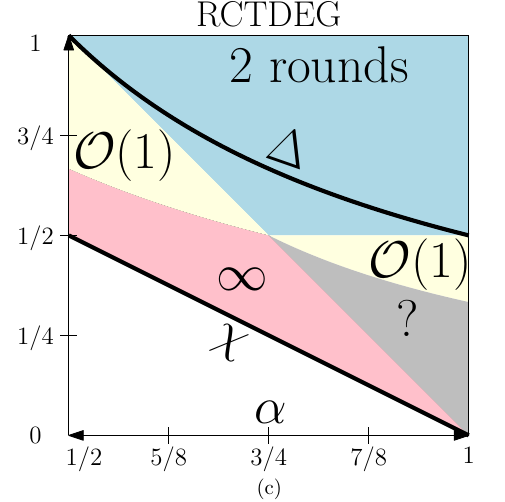}

 \caption{Comparison between the random colour trials \RCT,~\RCTID~~and~\RCTDEG in number of rounds depending on the number of colours. The parameter $\alpha \in (1/2,1)$ controlls the average degree (x-axis) and $f(\alpha)$ (y-axis) represents the exponent of the polynomial $n^{f(\alpha)}$ which is the colour space used by the algorithm. The heavier lines indicate the exponents $f(\alpha) = 1/(2\alpha)$ of the maximum degree $\Delta$ or rather the exponent $f(\alpha) = 1-\alpha$ of the chromatic number $\chi$ for a typical hyperbolic random graph. (a)~\RCT does not terminate in the red area \whp, does terminate in 2 rounds \aas in the blue area, and the grey area is unknown. (b)~\RCTID does not terminate in the red area \whp and terminates in the blue area in 2 rounds \aas (c)~\RCTDEG terminates in 2 rounds \aas in the blue area, terminates in constant rounds in the yellow area \aas, does not terminate in the red area~\whp and the grey area is unknown.}
\label{fig:plot}

\end{figure}

\subparagraph{Limitations to colouring HRGs with random colour trials.}
When considering the minimal number of colours required for \RCT to legally colour a hyperbolic random graph, a natural benchmark is the chromatic number $\chi\ll \Delta$. However, even for small constant choices of $\eps$ in \Cref{thm:main-theorem} part 1, \RCT uses substantially more than $\chi$ colours. In \Cref{thm:main-theorem} part 2  we show that \RCT cannot colour with close to the optimal number of colours $n^{\Theta(1)}$ factor difference remains. See \Cref{fig:plot} for an illustration of the gap between upper and lower bounds and their relation to the chromatic number of all of our main results. 

\smallskip

\subparagraph{Random colour trial with ID priority (\boldmath\RCTID).}  At this point, one might wonder whether a modifications of \RCT could yield even better results for hyperbolic random graphs--specifically, whether it could colour the graph within two rounds with a smaller error probability or use fewer colours. In \RCT all nodes back-off whenever there is a conflict. Following prior work \cite{BEPS}, a natural variant to mitigate this issue that has also shown to be helpful in simplifying proofs is to introduce unique IDs for breaking ties: instead of a vertex discarding its candidate colour whenever there is any neighbour with the same candidate colour, a vertex discards its candidate colour only if it has a neighbour with a smaller ID that tries the same candidate colour. For example the vertex with the smallest ID  always permanently adopts its candidate colour. We call the resulting algorithm \emph{random colour trial with ID priority} (\RCTID). Indeed, we show that the probability of \RCTID not finishing within two rounds on a typical HRG is decaying polynomially in $n$, which is stronger than the guarantees for \RCT shown in \Cref{thm:main-theorem}. Additionally, we show that one can essentially not improve upon this result. The minimal amount of colours required for termination is $\tilde\Omega(\Delta)$\footnote{We write $f(n) \in \Tilde{\Omega}(g(n))$ if there exists a constant $c>0$ such that $f(n) \leq \log^{-c}(n)\cdot g(n)$} (see \Cref{fig:plot}b.). We sum up our findings for \RCTID in the following theorem: 

\begin{restatable}[Random colour trial with ID priority]{theorem}{rctidTheorem}\label{the:rctid}\text{ }
\begin{compactenum}
    \item  For any constant $\eps>0$ there exists a constant $c>0$, such that \RCTID executed with $\eps\cdot \Delta$ colours on a typical HRG with maximum degree $\Delta$ terminates after exactly 2 rounds with probability $1 - n^{-c}$.

\item \RCTID with $\bigO( \Delta\cdot \log^{-3}n)$ colours on a typical HRG with maximum degree $\Delta$ never terminates \wehp
\end{compactenum}
\end{restatable} 


\subparagraph{Colouring with $o(\Delta$) colours via degree priorities (\boldmath\RCTDEG).} The last variant of random colour trials we study on hyperbolic random graphs we call \emph{random colour trial with degree priority} (\RCTDEG). Again, in each round every uncoloured vertex chooses its candidate colour uniform at random from its available colour space. Then a vertex assigns its candidate colour if no neighbour with a larger degree tries the same candidate colour, breaking ties arbitrarily (for example by ID's). One can think of it as a greedy approach where larger degree vertices are given the priority and one aims to colour large degree vertices first. \RCTDEG is motivated by the arguments in our impossibility results for \RCT and \RCTID, in both of which, high degree vertices fail to find a free colour. In a nutshell, we show the following for \RCTDEG:

\begin{theorem}[Simplified version of \cref{the:rctdeg}]
\label{the:rctdegSimplified}
There exist constants $\delta_2 > \delta_1 > 0$ such that on a typical hyperbolic random graph $G$
\begin{itemize}
    \item \RCTDEG with colour space $\Delta^{1 - \delta_1}$ terminates after $2$ rounds \aas;
    \item \RCTDEG with colour space $\Delta^{1 - \delta_2 +o(1)}$ terminates after $\bigO(1)$ rounds \aas;
    \item  \RCTDEG with colour space $\Delta^{1 - \delta_2 - o(1)}$ never terminates \wehp if $\alpha \leq 3/4$. 
\end{itemize}
\end{theorem}

This theorem shows that one can colour real-world graphs exteremly fast with few colours. The main benefit is that the algorithm itself is extremely simple shifting the difficulty to the analysis. All our constant-time runtime upper bounds hold \aas or with probability $1-n^{-c}$ for some constant $c<1$. This is no coincidence. When colouring with $\Delta^{1-\delta}$ colours, then with probability $n^{-c}$, \RCTDEG quickly runs into a partial colouring that cannot be completed to a colouring of the whole graph.
\smallskip

While \RCTDEG uses $o(\Delta)$ colours and small messages, the problem whether there is a low-bandwidth algorithm of comparable simplicity that can colour with close to $\chi$ colours remains open.


\smallskip

\subparagraph{Non-bandwidth efficient Algorithms.} Lastly, we complement our results with a deterministic but in practice unrealistic LOCAL model algorithm to essentially colour with the optimal number of colours.

\begin{restatable}[Deterministic LOCAL colouring]{theorem}{deterministicLOCAL}\label{the:deterministic-colouring}
A threshold hyperbolic random graph can be coloured by a LOCAL deterministic distributed algorithm in constant many rounds and with $(1+o(1))\cdot \chi$ colours \wehp  
\end{restatable}

\smallskip

 Note that the probabilistic guarantee here is with respect to the distribution of the graph and not to the algorithm. In contrast to our random colour trial variations that also work in the more restrictive CONGEST model, the algorithm of \Cref{the:deterministic-colouring} heavily exploits the unbounded message size of the LOCAL model.

\subsection{Background on Hyperbolic random graphs}
\subparagraph{Power-law degree and clustering coefficient.} Many real-world networks like the internet or social networks share common properties like a power-law degree distribution \cite{ faloutsos1999power, vhhk-s-19} and a non-vanishing clustering coefficient \cite{Newman2003, PhysRevE.74.056114, Watts1998}. Krioukov, Papadopoulos, Kitsak, Vahdat, and Boguñá~\cite{kpk-h-10} introduced \emph{hyperbolic random graph} (HRG) where they use hyperbolic geometry for random graphs to incorporate both of these properties \cite{gpp-rhg-12}. In contrast to a random geometric graph (RGG) \cite{p-rgg-03} that lives in the Euclidean space, a hyperbolic random graph is generated in the hyperbolic plane  $\bH$. One restricts the plane to a disk of radius $R$ and then $n$ vertices are sampled into the disk where its radius $R$ scales with $n$. As opposed to the Euclidean plane where space expands polynomially, hyperbolic space expands exponentially, and as a result, most area of the disk is close to the boundary. This leads to most vertices being located close to the boundary of the disk while only a few vertices are close to the centre of the disk. The hyperbolic geometry also implies that vertices at the boundary of the disk have smaller (expected) degrees which gives rise to a heterogenous degree distribution. This is in contrast to the homogeneous degree distribution of RGGs while preserving large clustering coefficients of models based on Euclidean space. 
\subparagraph{Algorithmic properties.} The resulting graph has further properties one finds in many complex networks like a small diameter \cite{fk-dhrg-18, DBLP:conf/analco/KiwiM15, ms-k-19} and a giant component of linear size \cite{bfm-giant-15, fm-giant-18}. From the algorithmic side, many problems like minimum vertex cover \cite{katzmann-exactvc-2023, katzmann-approxvc-2023}, maximum independent set \cite{bfk-tw-2016}, the largest clique \cite{bfk-chrg-18}, the shortest path \cite{bffkmm-icalp-2022} or colouring \cite{bmrs-stacs-25} have been studied in the centralised setting. Moreover, the two core properties of a non-vanishing clustering coefficient and a power law degree distribution inherited by HRGs were also proposed for \emph{propositional satisfiability} by~Bläsius,~Friedrich,~Göbel,~Levy~and~Rothenberger~\cite{TommyAndy} as a possible explanation for the tractability of industrial SAT instances.

\subsection{Background on Distributed Graph Colouring}
Graph colouring is one of the most-studied problems in in the field of distributed computing and we have already mentioned several results. Many early results are summarized in the excellent monograph by Barenboim and Elkin dedicated to distributed graph colouring \cite{barenboimelkin_book}.
Regarding general upper bounds when colouring with $\Delta + 1$ colours, in a breakthrough result, Rozho\v{n} and Ghaffari introduced a $\poly\log n$ deterministic algorithm for the network decomposition problem which due to reductions \cite{SLOCAL17,newHypergraphMatching} applies to a huge class of problems including the classic $\Delta+1$-vertex colouring problem, and also implies a $\poly\log\log n$-round randomised algorithm for $\Delta+1$-vertex colouring \cite{CLP18}.  For the more restrictive CONGEST model~\cite{paleg-00}, Halld{\'{o}}rsson,~Kuhn,~Maus~and Tonoyan later accomplished the same efficiency~\cite{HKMT21}.  Halld{\'{o}}rsson,~Kuhn,~Nolin~and Tonoyan matched these bounds for $\Delta+1$-list colouring~\cite{HKNT22} and $deg+1$-list colouring \cite{HNT22}. The current state of the art for distributed colouring in the LOCAL model stems from  the celebrated work by Ghaffari~and~Grunau~\cite{GG24}. They designed an $\tilde{O}(\log^{5/3}n)$-round algorithm for the maximal independent set problem, that if combined with a reduction by Luby \cite{luby86} and the aforementioned works also yields a  $\tilde{O}(\log^{5/3}\log n)$-round randomised algorithm for $\Delta+1$-vertex colouring. Here $\tilde{O}$ hides factors  that are logarithmic in the argument.

For graphs with a relatively small maximum degree $\Delta$, there are several deterministic algorithms that colour in  $f(\Delta)+ \bigO(\log^*n)$ rounds~\cite{barenboim15,BEG17,FHK16,MT20} while for graphs with maximum degree $\Delta$ larger than $\poly\log n$, there are randomised algorithms which colour within $\bigO(\log^* n)$ rounds~\cite{SW10,HN21}.

\smallskip 

All aforementioned results consider a colour space of size $\Delta + 1$. Eventhough Brook's theorem says that all graphs except for odd cycles and cliques are colourable with at most $\Delta$ colours~\cite{Brooks_1941}, using $\Delta$ instead of $\Delta + 1$ colours makes the problem much more involved and there exists an entire body of work dedicated towards this direction~\cite{PS95,GHKM18,FHM23,HM24,BBN25}

When considering more specialised graph classes, the infamous four~colour~theorem states that any planar graph can be coloured with $4$ colours~\cite{4-color-part-1, 4-color-part-2}. Chechik~and~Mukhtar~\cite{planar-graphs1} presented a deterministic algorithm that colours a planar graph with $6$ colours in $\bigO(\log(n))$ rounds which is runtime-optimal. In~\cite{planar-graphs2}~Aboulker,~Bonamy,~Bousquet,~and~Esperet established that for $4$ colours there is no polylogarithmic time algorithm for general planar graphs while for $5$ colours Postle showed in~\cite{planar-graphs3} that planar graphs are colourable in polylogarithmic time. 

For graphs with arboricity $a$, Barenboim~and~Elkin accomplished an algorithm for the LOCAL model that runs in $\bigO(\log a \log n)$ rounds and produces for any constant $\eps > 0$ a colouring with $\bigO(a^{1+\eps})$ colours~\cite{barenboim10}. Since the chromatic number $\chi$ and the arboricity $a$ of HRGs are bounded by $a \in \Theta(\chi) \in n^{\Theta(1)}$~\whp~\cite{bmrs-stacs-25,bfk-tw-2016}, this implies that HRGs can be coloured in $\bigO(\log^2(n))$ rounds using $\bigO(\chi^{1+\eps})$ colours~\whp over the randomness of $G$.

To the best of our knowledge, the only previous work where distributed colouring has been studied on any random graph model, is the work by Krzywdziński~and~Rybarczyk~\cite{random-graphs-distributed-colouring}. Here the authors considered Erd\H{o}s--Réyni graphs and the LOCAL model without any constraints on the size of the messages for which they obtain a $\bigO(\log\log n)$ rounds algorithm when the graph has average degree $o(\sqrt{n})$.

\section{A Bird's Eye View on our Analysis}
One of the major challenges of analysing randomised algorithms on hyperbolic random graphs arises from the fact that we need to deal with two different sources of randomness: the randomness of the graph and the probabilistic nature of a randomised algorithm. To circumvent the problems that arise from dealing with two different probability spaces at once, we make use of structural graph properties that emerge asymptotically almost surely on HRGs. We call a HRG with these  properties a \emph{typical} HRG.  We identify properties of typical HRGs and condition on those to ensure the desired fast performance of our randomised algorithms. To the best of our knowlegde, we are the first to analyse a randomised algorithm on HRGs. However, Kiwi, Schepers and Sylvester have analysed the behaviour of random walks on HRGs~\cite{Kiwi2024}.

In any case, the main difficulty then lies in identifying the aforementioned properties, i.e.,  properties that (1) hold with a sufficiently large probability over the generation of the random graph, and that (2) are useful for designing a fast algorithm for those typical  graphs, or in our case to show that the given \RCT-based algorithms perform well on HRGs. Examples of properties that emerge with high probability in HRGs and are algorithmically exploited in other works are for example the \emph{dominance rule} for vertex cover \cite{katzmann-exactvc-2023, katzmann-approxvc-2023}, the \emph{degeneracy} for colouring \cite{bmrs-stacs-25} or a mixture of \emph{locality} and
\emph{heterogeneity} for the study of shortest path \cite{bffkmm-icalp-2022}. Bl\"{a}sius and Fischbeck \cite{bf-evacaga-22} provide evidence that these  structural properties may also appear in real-world networks, possibly explaining why simple algorithms outperform their worst-case bounds.

\subsection{Upper bounds for \RCT and the Layer Discretization}
In the classic runtime proof of \RCT one shows that each vertex gets coloured with a constant probability in each round of the algorithm, regardless of what has happened in the past, yielding a logarithmic runtime; one can also show that this is tight for worst-case input instances, see \Cref{app:RCTLowerBoundWorstCase}. In order to show that \RCT completes within few rounds on HRGs, we make use of their typical structure. We use a  discretisation of the disk into $\lfloor R \rfloor \in \bigO(\log(n))$ layers $\mathcal{A}_0, \mathcal{A}_1, \ldots, \mathcal{A}_{\lfloor R \rfloor - 1}$ similar to \cite{fk-dhrg-18}; the main advantage is that vertices in the same layer $\mathcal{A}_\ell$ have similar degrees of roughly $e^{\ell/2}$.  See \Cref{fig:neighbourhood-and-layers} for an illustration of the layers and the neighbourhood of a vertex following the geometry of the hyperbolic disk. 
\begin{figure}[t]
    \centering \includegraphics[height=0.3\textheight]{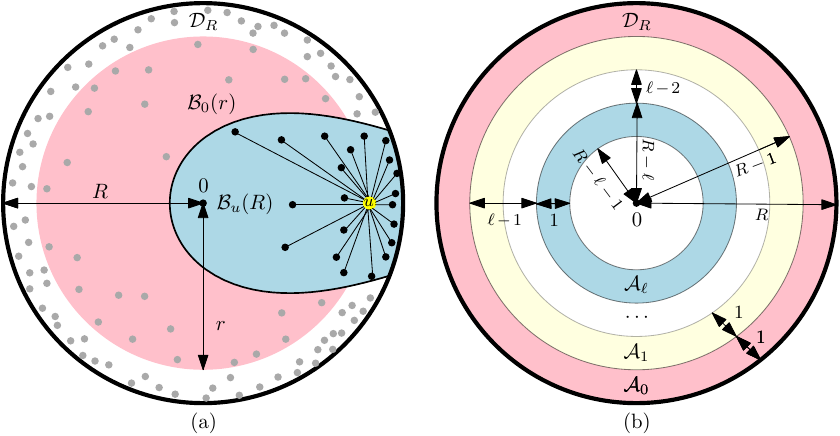}
    \caption{(a) Illustration of the neighbourhood $N(u)$ of a vertex $u$  given by $V \cap \B_u(R)$ (blue region) following the geometry of the hyperbolic disk. The red area is the ball $\B_0(r)$ centred around the origin for some $r$. (b)~Sketch of  layer $\mathcal{A}_0$ (red area), layer $\mathcal{A}_1$ (yellow area) and some layer $\mathcal{A}_\ell$ (blue area) for some $0< \ell<R$.}
    \label{fig:neighbourhood-and-layers}
\end{figure}
In the sequel, we use different arguments and measures of progress for different layers. The probability of a small degree vertex to be coloured is much larger than constant, in fact, the probability to be unsuccessful with colouring in the first round is upper bounded by $\deg(v)/\mathsf{|ColourSpace|}$, as $v$'s neighbours can block at most $\deg(v)$ colours with their candidate colours. Ignoring dependencies, this shows that in the layers with constant degrees only $\bigO(n/\mathsf{|ColourSpace|})=\bigO(n/\Delta)=o(\sqrt{n})$  vertices remain uncoloured after the first iteration, and heuristically, only $\bigO(\sqrt{n}/\Delta)\ll 1$ nodes after the second iteration. 

The main challenge is to make sufficient progress in  layers with large degree nodes. Imagine that we want to prove \Cref{thm:main-theorem} with $\eps=1/100$, that is, we aim to colour the graph with $\Delta/100$ colours. For a node $v$ in a layer with degrees $\gg \Delta/100$ this is troublesome, as we cannot complete the colouring for $v$ (not even existentially) if $v$'s neighbours use the whole colour space resulting in no available colours left for $v$. First, we show that the randomness of \RCT implies that the  individual palette of available colours of each node remains of size $\Omega(\Delta)$ after the first round. While large degree vertices initially may have up to $\Delta$ neighbours, we show that after the first round their uncoloured degree drops to $o(\Delta)$; the latter statement heavily exploits the structure of a typical HRG as most neighbours of a large degree node have a much better than constant success probability. We show that in subsequent rounds, the palette remains large, but due to the smaller uncoloured degree the probability to remain uncoloured decreases drastically. 

\subsection{Upper bound for \RCTID} In order to show a smaller error probability than the one we show for \RCT, we exploit that whether a vertex discards its candidate colour or not, does not necessarily depend on all its neighbours. Following the lines drawn by Barenboim,~Elkin,~Pettie~and~Schneider~\cite{BEPS}, 
we use that the ID's give us a natural order in which we reveal the randomness of the candidate colours (only for the purpose of analysis), such that when we reveal the candidate colour of some vertex $v$ , the probability that $v$ discards its colour only depends on the already revealed randomness. This allows us to apply an Azuma-Hoeffding type concentration bound. To obtain the desired overall result, we need to carefully apply this idea to the nodes in different layers separately exploiting the respective structural properties to obtain strong concentration for the drop in the number of uncoloured vertices per layers and down the line the uncoloured degree of vertices after the first round. This leads to the  smaller $1/\poly n$ error guarantees of \RCTID as compared to the \aas guarantee of \RCT. We remark that a seemingly easier approach for \RCT employed by prior work is actually incorrect~\cite{HKMT21}. As many neighbours have individual small probabilities to remain uncoloured, one may be tempted to show that the uncoloured degree of a vertex $v$ is stochastically dominated by a process of independent binomial random variables indicating whether $v$'s neighbours are coloured or not. If this was true, one could apply a simple Chernoff bound as done in~\cite{HKMT21} to show that the uncoloured degree of vertices drops significantly with high probability.  See \Cref{sec:stochastic-dominance} for a detailed discussion and an explanation of why this analysis in~\cite{HKMT21} is incorrect and more care is needed.

\subsection{Upper bound for \RCTDEG using \texorpdfstring{$o(\Delta)$}{o(Delta)} colours}
\begin{figure}[t]
    \centering \includegraphics[height=0.3\textheight]{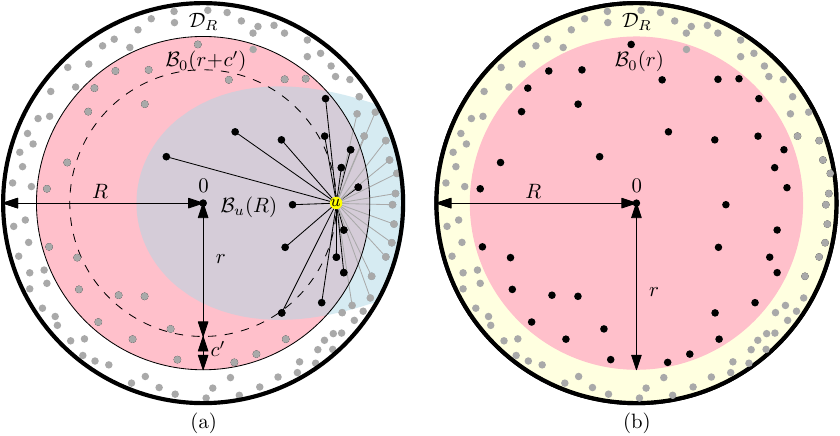}
    \caption{ (a)Illustration of a larger degree neighbourhood $N^+(u)$. It holds \whp that no vertex outside ball $\B_0(r+c')$ (red area) has a larger degree than $\deg(u) = |V \cap  \B_u(R)|$~(vertices in the blue area). Given this event, we have $N^+(u) \subseteq V \cap \B_0(r+c')$ and we can upper bound the larger degree neighbourhood of $u$ by $\deghigh{u} \leq |V \cap \B_0(r+c') \cap \B_u(R)|$~(intersection of blue and red area). (b) Sketch of "high degree vertices". Every vertex in $\B_0(R) \setminus \B_0(r)$ (yellow area) has a degree much smaller than the colour space and all vertices in $\B_0(r)$ (the red area) are coloured in the first round \aas}
    \label{fig:larger-degree-neighbourhood}
\end{figure}
For the analysis of the \RCTDEG algorithm we coin the term \emph{larger degree neighbourhood} in~\Cref{sec:larger-degree-neighbourhood}. The larger degree neighbourhood of a vertex $u$ with degree $\deg(u)$ is given by the set $N^+(u):= \{v \in N(u) : \deg(v) \geq \deg(u)\}$ (see \Cref{fig:larger-degree-neighbourhood}a for an illustration, intersection of the red and blue area). Since \RCTDEG gives priority to vertices with larger degree, this entails that the probability that $u$ is coloured only depends on the randomness of uncoloured vertices in $N^+(u)$ and the remaining colour palette of $u$. Our analysis takes advantage of this fact by showing that, given the colour space is large enough, all vertices with degree above some threshold
$\mathsf{LargeDegree} \ll \mathsf{|ColourSpace|}$ get coloured within the first round \aas (see red area in~\Cref{fig:larger-degree-neighbourhood}b for a sketch). We do so by showing that the entire set of $\mathsf{LargeDegree}$ vertices are rainbow coloured. We show that they all try different candidate colours among each other \aas and due to the priority over small degree vertices they can also adopt these colours permanently resulting in the desired rainbow colouring.  While the probabilistic guarantees for small degree vertices as compared to \RCTID mostly stay intact, such a rainbow colouring cannot be achieved for them, simply as there are way more small degree vertices than the size of the colour space. Thus the analysis requires a careful balance between the threshold of being a large degree vertex and the number of large degree vertices.


 The remaining uncoloured vertices in the second round~(see \Cref{fig:larger-degree-neighbourhood}b, yellow area) then have a degree that is smaller than the colour space given to the algorithm. This implies that, given the first round was successful at colouring all high degree vertices, the algorithm will produce a valid colouring eventually; we show that such a colouring is found quickly as well. For this we make use of the fact that the uncoloured vertices with a relatively large degree, have most of their larger degree neighbours already coloured.

\subparagraph{Why does this result in colouring with $o(\Delta)$ colours?} A key advantage of \RCTDEG is the successful colouring of the $\approx \sqrt{|\mathsf{ColourSpace}|}$ highest degree vertices in the first round: in contrast to \RCT, the rainbow candidate colours among high degree vertices directly ensure that they are also coloured when using \RCTDEG. Now, exploiting the power-law structure of a typical HRG, with  a colour space of size $o(\Delta)$, after the first round all non-coloured vertices have a degree that is smaller than the colour space. Thus, it is also ensured that a legal colouring for the yet to be coloured sub-graph exist. Indeed, we show that such a colouring for the remaining uncoloured sub-graph is found quickly as well.

\subparagraph{Behaviour of \RCTDEG for different power-law exponents.} In order to fully grasp our results for \RCTDEG, recall that the degree distribution of a hyperbolic random graph follows a power law, where the power law exponent $2\alpha + 1$ is controlled by the model parameter $\alpha \in (1/2,1)$. We refer to \Cref{the:rctdeg} in \Cref{sec:rctdeg} for formal statements. Due to the theorem's technicality we chose to explain the results via the illustration in \Cref{fig:plot}c instead. The figure parameterises the size of the colour space and the behaviour of \RCTDEG for different values of the parameter $\alpha$. If the size of the colour space falls into the blue area, the algorithm terminates within 2 rounds only, if it falls into the yellow area it terminates in constant time, and if it falls into the red area it does not terminate at all. Most strikingly, for $\alpha \leq 3/4$ our bounds on the necessary colour space for \RCTDEG to successfully colour a typical hyperbolic random graph in constant rounds are tight. The behaviour of the algorithm for a colour space size residing in the grey area marked with a question mark remains open.


\subsection{Limitations of \RCT, \RCTID and \RCTDEG.}
As any partial colouring can always be extended to a legal colouring at low-degree vertices regardless of how the rest of the graph has been coloured, any non-termination result should investigate the behaviour of large-degree vertices. For example, if some palette becomes empty during the algorithm's execution, it won't terminate. 
Palettes change over time and it is notoriously difficult to keep track of the probability distribution of vertices' palettes over multiple rounds of the algorithm. Thus, the key ingredient in proving our nearly matching lower bounds on the performance of these graph colouring algorithms is a structural result on HRGs. This result allows us to show that if the colour space is too small, a non-completable partial colouring is likely to emerge already after the first round.  This structural result may be of independent interest:  We show that most vertices in most of the layers $\mathcal{A}_0, \mathcal{A}_1, \ldots, \mathcal{A}_{\lfloor R \rfloor - 1}$ have many leaves, i.e., degree $1$ vertices, as neighbours. To be slightly more formal, for a vertex $u$ let $L(u)\subseteq N(u)$ be the set of its neighbours that are leaves and define  $U=\{u\in V \mid |L(u)|=\Theta(\deg(u))\}$  as the set of vertices for which a constant fraction of their neighbours are leaves. We prove the following lemma.

\begin{figure}[t]
    \centering \includegraphics[height=0.3\textheight]{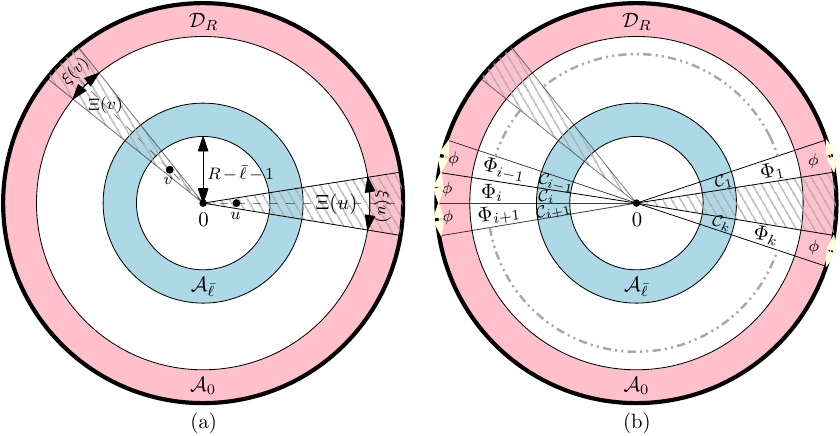}
    \caption{Sketch of the construction for \Cref{lem:leaf}: (a) Illustration of sectors we remove from $\disk$. (b) Excerpt of the subdisk $\disk\setminus \Xi$ divided into $k$ sectors where every sector $\Phi_i$ has angle $\phi$.}
    \label{fig:leaves}
\end{figure}

\smallskip

\noindent\textbf{Leaves Lemma} (informal).
\emph{ For most layers a constant fraction of its vertices are in $U$. 
}
\smallskip

See \Cref{lem:leaf} for the formal statement. 
To prove this we utilise a careful geometric construction such that the event for certain vertices to be a leaf is independent of most parts of the disk. To this end we 
dissect most of the disk into $k$ sectors $\Phi_1, \ldots, \Phi_k$ each of a fixed angle $\phi$ (see~\Cref{fig:leaves} for the construction). Then, for a fixed layer $\mathcal{A}_{\Bar{\ell}}$, we obtain for all $i \in [k]$: 
\begin{enumerate}[label=(\alph*)]
    \item\label{item:easy} in sector $\Phi_i$ we find with constant probability a vertex $w$ in layer $\mathcal{A}_{\Bar{\ell}}$ (area $\mathcal{A}_{\Bar{\ell}} \cap \Phi_i = \mathcal{C}_i$ in~\Cref{fig:leaves}b);
    \item\label{item:hard} $w$ has with constant probability a constant fraction of its neighbours as leaves located in $\mathcal{A}_0 \cap \Phi_i$.
\end{enumerate}
Proving \Cref{item:easy} is straightforward as long as angle $\phi$ defining the intersection of the sector with the area of the layer is large enough. The larger challenge lies in proving \Cref{item:hard} that we address in \Cref{lem:leaves}. In its proof we use that any vertex in layer $\mathcal{A}_0$ has a constant probability to be a leaf. Thus the region of the disk whose randomness we need to reveal (in order to show that many leaves adjacent to $w\in V \cap \mathcal{C}_i$) is relatively small; it intersects at most the two neighbouring sectors besides $\Phi_i$ itself (other than the removed part of the disk sketched in \Cref{fig:leaves}a). Thus, there are enough sectors where the events are independent. Essentially, this allows us to flip coins for the graph construction independently for a constant fraction of sectors, each with constant probability providing a fresh vertex $w$ satisfying (a) and (b). Thus, for each fixed layer, we can use Chernoff bound to show that the number of vertices satisfying the condition of the leaves lemma is well concentrated. Lastly, a union bound over the layers wraps up the proof.

\smallskip

The leaf structure is useful since events in the context of randomised colouring of vertices with degree $1$ are mostly independent of other events. In particular revealing the candidate colour of the only neighbour of the leaf, the probability that the leaf tries the same candidate colour as its neighbour is independent of any other event. In~\Cref{sec:colour-space}, we use the leaves lemma  to show that after a single round of \RCT with a too small colour space, that there are so called \emph{colour locked} vertices. These are vertices that have only one possible candidate colour to choose from as all other colours are used by its premantently coloured neighbours. We show that after the first round many vertices of a relatively large clique are colour locked on a uniform at random colour. When the number of colours available to \RCT is not too much larger than the number of vertices in this clique of colour locked vertices, the "birthday paradox" enforces that \whp there exists a pair in the clique that is colour locked on the same colour. Then not only the algorithm does not terminate anymore but even existentially the computed partial colouring cannot be completed.

 In~\Cref{sec:rctid-part2}, we use the leaves lemma for  our non-termination results for \RCTID. To show that \RCTID never terminates when the number of colours is less than $\Delta/\log^{\Omega(1) }n$, consider a large-degree vertex $u$ for which many of its small-degree neighbours, including many leaves, have a smaller ID.  In this case, one can show that the entire colour space is consumed by $u$'s leaf-neighbours in the first round. This implies that $u$ has no possible candidate colour for round two and beyond and the algorithm will never terminate.

 In \Cref{sec:colour-space-rctdeg},  we use a similar argument as for \RCT on large-degree vertices being colour-locked for \RCTDEG.  In principal, this is also the reason why a \emph{with high probability guarantee} for \RCTDEG is not possible if we use a colour space that is $\Delta^{1-\delta}$: using that few colours, with probability $n^{-c}$ for a constant $c < 1$ there exists a pair of high degree vertices which are neighbours picking the same colour while their respective colour spaces are consumed by their leaves. This birthday paradox can only be omitted with probability $1 - n^{-c}$ which implies that with probability at least $n^{-c}$ the algorithm never finishes.

\section{Conclusion}\label{sec:conclusion}
 In this paper we investigated how the performance of simple randomised distributed colouring algorithms on hyperbolic random graphs depend on the colour space. We showed that the most classic and simple algorithms undergo an abrupt transition from using $2$ rounds when using $\eps\cdot\Delta$ colours, over  constant rounds for fewer colours to never terminating when the colour space is further restricted. Interestingly, our positive results demonstrate that constant-time algorithms can break the fundamental $\Omega(\log^* n)$ worst-case runtime lower bound for the $\Delta+1$-vertex colouring problem proven by Linial and Naor for deterministic and randomised algorithms respectively \cite{Linial-92,Naor91}. 

All our results are based on several structural insights on the structure of hyperbolic random graphs that may be of independent interest. For upper bounds on the colour space required, all three randomised algorithms we analysed require significantly less colours than the maximum degree $\Delta$ and for any $\eps >0$, it is sufficient to use $\eps \cdot \Delta$. We showed that when we use the degrees of vertices as a heuristic for which vertices should be prioritised, the algorithm is highly likely to be successful even if the colour space is $\Delta^{1-\delta}$ for some $\delta >0$. Moreover, our upper bounds for this algorithm match our established lower bounds by a multiplicative polylogarithmic factor for certain regimes of the power law degree distribution inherited by hyperbolic random graphs.


\subparagraph{Future directions.} There are several intriguing directions to go from here. 

\begin{compactitem}
\item 
First, there remains the question whether \RCTDEG is truly superior to \RCT. In the area where \RCTDEG colours with $o(\Delta)$ colours, the performance of \RCT remains open. Surprisingly, \RCTID cannot close this gap and does not terminate with $\bigO(\Delta/ \polylog n)$ colours (see also \Cref{fig:plot}).

\item Second, for $\alpha>3/4$, it is unclear if our lower bound on the colour space for \RCTDEG is tight for all regimes of $\alpha$ or if our upper bound can be improved upon (see grey area in~\Cref{fig:plot}c).

\item Third, it would be interesting to further investigate the existence of efficient deterministic  CONGEST colouring algorithms for hyperbolic random graphs. 

\item Fourth, it remains open whether another simple CONGEST algorithm can colour with $(1+o(1))\chi$ colours; as shown \RCTDEG  uses few colours but still significantly more than the optimal number of colours. Alternatively, it remains open to design such algorithms whose runtime bounds hold with high probability. 

\item Fifth, can other typical problems in the field like maximal independent sets be solved faster on hyperbolic random graphs than on worst-case instances \cite{BBHORS21}?
\end{compactitem}

\smallskip

In general, we  hope to spark further research on distributed algorithms and more generally sublinear algorithms on real-world graphs.

\section*{Outline of the rest of the paper} The remainder of this article is structured as follows. In \Cref{sec:prelims} we formally introduce the model of hyperbolic random graphs and give important definitions and notation. In \Cref{sec:structural} we show structural results for HRGs that we make use of throughout the paper. We then analyse the behaviour of \RCT on hyperbolic random graphs in \Cref{sec:rct}. The results for \RCTID can be found in \Cref{sec:rctid}. \Cref{sec:rctdeg} contains the analysis of \RCTDEG. Our findings for the deterministic approach are given in \Cref{sec:deterministic}. 

\section{Hyperbolic Random Graphs}\label{sec:prelims}

\subparagraph{Hyperbolic plane.}
The central mathematical object in this work are \emph{hyperbolic random graphs} as introduced in their \emph{native representation} by Papadopoulos~et~al.~\cite{pkbv-info-2010}. We operate on the hyperbolic plane $\bH = [0, \infty) \times[0,2\pi)$ where a point $x \in \bH$ is equipped with a radial coordinate $r(x)$ and an angular coordinate $\varphi(x)$. We restrict our attention to the curvature -1 where the \emph{hyperbolic distance} $\dist(\cdot,\cdot)$ for two points $x,y \in \bH$ is  
\begin{align}\label{eq:hyperbolic_distance}
    \cosh(\dist(x,y)) = \cosh(r(x))\cosh(r(y)) - \sinh(r(y))\sinh(r(x))\cos(\varphi(x) - \varphi(y)).
\end{align}
We equip \bH with the topology induced by $\dist$.

Throughout this work, we use a subspace of the hyperbolic plane \bH; namely a disk of radius $R$ for which we write $\disk = [0,R]\times[0,2\pi)$, where we call the point $(0,0)$ the \emph{origin}. Since we restrict our space to $\disk$, the set of points with distance $r$ to $x\in\disk$ is given by the set $\B_x(r) = \{y\in\disk : \dist(x,y)\leq r\}\subseteq \disk$.

We now introduce a probability measure $\mu$ on \disk, which is parametrised by the model parameter $\alpha\in(1/2,1)$. For measurable $\mathcal S\subseteq \disk$, define
\begin{align*}
	\mu(\mathcal S) = \int_S \func(x) \dd x, 
    \qquad \func(x) = \frac{\alpha\sinh(\alpha x)}{2\pi(\cosh(\alpha R) - 1)}, 
\end{align*}
where $\rho$ is the density of $\mu$ with respect to the Lebesgue measure on \disk. 

\subparagraph{Hyperbolic random graphs.} We use the \emph{Poissonised version} of \emph{(threshold) hyperbolic random graphs} similar to Kiwi and Mitsche in~\cite{km-slcrhg-19} (see also \cite[§3.3.4]{katz-diss-23} for the model as defined here). A hyperbolic random graph (\emph{HRG}) is a graph $G = (V,E)$ with vertex set $V$ and edge set $E$ that is obtained by the following inhomogeneous Poisson point process: let $n$ be an integer which is the expected value for the Poisson random variable $|V| = N$, i.e., the number of vertices. Then set the radius of our disk $R := \log(n) + C$ for some $C \in \Theta(1)$ and let the intensity function at polar coordinates $(r, \varphi)$ for $0\leq r < R$ be $\intensity(r, \varphi):=n\rho(x)$. The set of vertices is a random variable $V = \{(r_1,\varphi_1), (r_2,\varphi_2),.., (r_N,\varphi_N)\}$ for which holds that for area $\mathcal{S} \subseteq \disk$ the number of expected vertices is given by $n\mu(\mathcal{S}) = \E{|V \cap \mathcal{S}|}$. 

To obtain our edge set $E$ in the threshold version, we connect a pair of vertices $u,v\in V$ if and only if their distance $\dist(u,v)$ is at most $R$ and write $G \sim \G$ for the resulting graph. In the temperature version of HRGs, an edge $\{u,v\} \in E(G)$ exists  with probability $p(u,v) = (1 + \exp\left(\frac{1}{2T} \left(\dist(u, v) - R\right)\right))^{-1}$ determined by both distance and a ``temperature'' parameter $T$ (see e.g. \cite[§3.1]{Krohmer2016}).

\subparagraph{Typical hyperbolic random graphs.} We call a threshold hyperbolic random graph that has all properties that occur \aas a \emph{typical hyperbolic random graph}. When we analyse randomised algorithms on typical hyperbolic random graphs, we assume that the input instance has all properties that hold \aas for a HRG. The probabilistic guarantees we then give for the algorithm are with respect to typical instances by conditioning on the event that the input graph $G$ is a typical hyperbolic random graph. We stress that typically, one strives for an \aas~probabilistic guarantee for hyperbolic random graphs like in the context of random walks~\cite{Kiwi2024}, the diameter~\cite{fk-dhrg-18, ms-k-19} or average distances~\cite{DBLP:journals/im/AbdullahFB17}.  Actually, the proofs for the diameter and random walks rely on the event that no vertex lies too close to the origin; however, this event does not hold \whp Also the upper bound on the maximum degree $\Delta$ in \Cref{the:max-degree} relevant for our results does not hold \whp and only \aas; see the discussion in~\Cref{sec:degree}.

In the following we collect some results by Gugelmann, Panagiotou and Peter~\cite{gpp-rhg-12} we make use of to prove that a typical hyperbolic random graph has the desired properties. For this it is convenient to characterise the connection of vertices in terms of their \emph{angular distance}. Since vertices are connected if and only if their distance is at most $R$, we define
\begin{align}\label{eq:angle-func}
    \theta_R(r_1,r_2) = \arccos\left(\frac{\cosh(r_1)\cosh(r_2) - \cosh(R)}
    {\sinh(r_1\sinh(r_2)}\right),
\end{align}
which denotes by \Cref{eq:hyperbolic_distance} the angle distance for two points with radii $r_1$ and $r_2$ such that their hyperbolic distance is exactly $R$. Thus, two vertices with an angle distance of at most $\theta_R$ are connected. For $\theta_R(\cdot,\cdot)$ the following is well known (\cite[Lemma 3.1]{gpp-rhg-12}).
\begin{lemma}\label{lem:max-angle}
Let $x,y \in \disk$ and $0 \leq r(x), r(y) \leq R$ and $r(x)+r(y) \geq R$. Then it holds
$$
\theta_R(r(x), r(y)) = \Theta(1)e^{\frac{R-r(x)-r(y)}{2}}.
$$
\end{lemma}

For a ball of radius $r$ and the origin as its centre point (as depicted in \Cref{fig:neighbourhood-and-layers}a), one get via $\int_{0}^r f(\rho)d\rho$ the following.

\begin{lemma}\label{lem:measure-inner-disk}
    For any $0 \leq r \leq R$ we have
$\mu\left(\mathcal{B}_0(r)\right) = (1 - o(1))e^{-\alpha(R- r)}~.$
\end{lemma}

Moreover, we write $N(u) := \{v \in V : \{u,v\} \in E(G)\}$ for the \emph{neighbourhood} of $u$ and $\deg(u) := |N(u)|$ for the \emph{degree} of $u$ which is a Poisson random variable. The following corresponds to the expected degree of a vertex $u$ with radius $r$ (see also \Cref{fig:neighbourhood-and-layers}a)\footnote{The original bounds are slightly different but ours directly follow from using \cite[Corollary 5]{bks-hudg-23} for $\theta_R(r,x)$; we provide the proof in \Cref{sec:degree}.}.

\begin{restatable}[\cite{gpp-rhg-12}, Lemma 3.2 Vertex Degree]{lemma}{degree}\label{lem:vertex-degree}
    Let $G \sim \hrg$ and let $u \in V$ be a vertex with radius $r\geq 1$. Then, the expected degree of $u$ is 
    $$
    \E{\deg(u)} = n\cdot\mu\left(\B_u(R) \cap \mathcal{D}_R\right) \begin{cases}
        \leq (1+o(1))n\cdot \frac{\alpha e^{-r/2}}{\alpha-1/2}\\
        \geq  (1-o(1))n\cdot\frac{\alpha e^{-r/2}}{\pi(\alpha-1/2)}.
    \end{cases} 
    $$
\end{restatable}

Finally, let $\Delta$ denote the maximum degree of a hyperbolic random graph $G$, i.e., $\Delta := \max_{v \in V}\deg(v)$. A hyperbolic random graph has the following maximum degree.

\begin{restatable}[\cite{gpp-rhg-12}, Theorem 2.4 Maximum Degree]{lemma}{maxDegree}\label{the:max-degree}
    Let $G\sim \hrg$ be a hyperbolic random graph. Then $\Delta < n^{\frac{1}{2\alpha}}\cdot\log(n)$ \aas\footnote{The authors claim \whp However, looking into proofs the probabilistic guarantee is only \aas according to our definitions. For completeness, we provide a proof for our here stated form in~\Cref{sec:degree}.} and $\Delta > n^{\frac{1}{2\alpha}}\cdot\log^{-2}(n)$ \wehp\

\end{restatable}

\subparagraph{Distributed Graph Colouring \& Notation.} 
Let $\palette = [k]$ be the set of $k$ colours. The integer $t \in \mathbb{N}{\geq 0}$ denotes the \emph{round} of the distributed algorithm. For a vertex $u \in V$, we write $\colourchosen{t}{u} = \colour$ if it is assigned colour $\colour \in \palette$ after round $t$, and $\colourchosen{t}{u} = \emptyset$ if it remains \emph{uncoloured}. Initially, at round $t = 0$, all vertices are uncoloured, i.e., $\colourchosen{0}{v} = \emptyset$ for all $v \in V$. The set of \emph{available} colours for vertex $u$ in round $t$ is defined as $\colourpal{t}{u} := \palette \setminus \cup_{v \in N(u)} \colourchosen{t-1}{v}$. When analysing the randomised distributed algorithm on a fixed graph $G$, we use $\Pro{\cdot}$ and $\Exp{\cdot}$ for probabilities and expectations; in contrast, we use $\Prob{\cdot}$ and $\E{\cdot}$ when referring to randomness over the graph distribution $G \sim \hrg$.

\section{Typical HRGs: Layers, Leaves \& Larger Degree Neighbourhood}\label{sec:structural}

In this section we introduce structural properties of hyperbolic random graphs that hold at least \aas and thus, also hold for typical hyperbolic random graphs. In \Cref{sec:layers} we introduce the discretisation of the hyperbolic disk into layers with the for us relevant properties in \Cref{lem:layer-properties}. In \Cref{sec:leaves} we present the leaves lemma that we use throughout the paper for the lower bounds on the colour space. At last, we establish to notion of larger degree neighbourhood in \Cref{sec:larger-degree-neighbourhood} which we use to analyse \RCTDEG. 
 
\subsection{Discretisation into Layers}\label{sec:layers}
 Through out the paper, we use the following discretisation of the disk $\mathcal{D}_R$: a \emph{layer} $\mathcal{A}_{\ell} \subset \mathcal{D}_R$ is the sub-space with distance $[\ell, \ell +1)$ to the boundary (see \Cref{fig:neighbourhood-and-layers}b for a sketch). More formally, for $\ell \in [\lfloor R\rfloor]$ the layer of \emph{level} $\ell$ is defined by $\mathcal{A}_\ell: = \B_0(R-\ell)\setminus\B_0(R-\ell-1)$. A similar layering has been introduced by  Friedrich and Krohmer~\cite{fk-dhrg-18} in order to bound the diameter of HRGs. We use a different property here, namely that vertices that share the same layer have roughly the same degree. Using \Cref{lem:measure-inner-disk} the measure of a layer is given by
\begin{align}\label{eq:layer-measure}
    \mu(\mathcal{A}_\ell) = \Theta(1)e^{-\alpha\ell}.
\end{align}
We also refer to the set of vertices $V \cap \mathcal{A}_\ell =: V_\ell$ as the vertices of layer $\ell$. The expected degree of a vertex $v \in V_\ell$ is given via \Cref{lem:vertex-degree} by
\begin{align}\label{eq:layer-degree}
    \E{\deg(v)} = \Theta(1)e^{\ell/2}.
\end{align}
Finally, we write $\finallayer := \max{(\{\ell : V_\ell \neq \emptyset\})}$, that is, the largest level for a layer that still contains at least one vertex. The following lemma is a simple application of mostly \Cref{eq:layer-measure} and \Cref{eq:layer-degree}. It sets us up with concentration bounds we make use of throughout our arguments. It provides us with (1) how many vertices are located in a layer, (2) what the degree of vertex in a layer is, and (3) what the largest level of a layer is.

\begin{lemma}[Layer lemma]\label{lem:layer-properties} Let $G\sim\hrg$ be a threshold hyperbolic random graph and let $\ell \in [\lfloor R\rfloor]$. Then it holds for $V_\ell := V \cap \mathcal{A}_\ell$, $v \in V_\ell$ and the largest layer level $\finallayer$ that
\begin{enumerate}
    \item\label{item:vertices-in-layer} \wehp $|V_\ell| \in \begin{cases}
			\Theta\left(ne^{-\alpha\ell}\right), & \text{if $\ell \leq \lceil \alpha^{-1}(\log(n) - 2\log(\log(n))) \rceil$ }\\
            \bigO(\log^2(n)), & \text{otherwise,}
		 \end{cases}$ 
    \item\label{item:degree-in-layer} and $\deg(v) \in {\bigO}\left(e^{ \ell/2} + \log(n)\right)$ \wehp
    \item\label{item:max-layer} Moreover, $\finallayer \leq \lceil\alpha^{-1}\cdot(\log(n) + \log(\log(n)) \rceil$ \aas
\end{enumerate}
\end{lemma}

\begin{proof}
     Fix any layer $\mathcal{A}_\ell$ and consider the random variable $|V_\ell|$, i.e., the number of vertices that are located in this layer. Notice that the number of vertices located in area $\mathcal{A}_\ell$ is a Poisson random variable. This allows us to apply Chernoff bound. Using \Cref{eq:layer-measure} we have for $\ell \leq \lceil \alpha^{-1}(\log(n) - 2\log(\log(n))) \rceil$ that $\E{|V_\ell|} = \Theta(1)ne^{-\alpha\ell} \in \Omega(\log^2(n))$ since $\E{|V_\ell|}$ is decreasing in $\ell$. Thus, using a Chernoff bound, we obtain 
    \begin{align}\label{eq:vertices-in-layer}
        \Prob{|V_\ell| \not\in \Theta\left(n e^{-\alpha \ell}\right)} \in n^{-\omega(1)}.
    \end{align}
    
For the case $\ell \geq \lceil \alpha^{-1}(\log(n) - 2\log(\log(n))) \rceil$, we use again that $\E{|V_\ell|}$ is decreasing in $\ell$ and thus, $\E{|V_\ell|} \in \bigO(\log^2(n))$ which proves \Cref{item:vertices-in-layer} using another Chernoff bound. 

Next, consider any vertex $v \in V_\ell$. Recall that the neighbourhood of a vertex $v$ lies within the ball $\B_v(R) \cap \disk$. Thus, the size of the neighbourhood is again a Poisson random variable and another Chernoff bound gives us for the degree of a vertex $v$ via \Cref{eq:layer-degree}
\begin{align}\label{eq:degree-in-layer}
    \Prob{\deg(v) \not\in {\bigO}\left(e^{ \ell/2}+ \log(n)\right)} \in \bigO(1/n^{c}),
\end{align}
where $c$ is an arbitrarily chosen constant. This concludes the argument for \Cref{item:degree-in-layer}.

Finally, set $r_0:= R - (\lceil\alpha^{-1}\cdot(\log(n) + \log(\log(n)) \rceil)$. Note that in order to ensure \Cref{item:max-layer}, it suffices to show that $V \cap \B_0(r_0) = \emptyset$. To this end we use \Cref{lem:measure-inner-disk} and obtain

\begin{align*}
    \E{|V \cap B_0(r_0)|} = n\cdot \mu(\B_0(r)) = n\cdot(1-o(1))e^{-\alpha(R-r)} = \Theta(1)(1/\log(n)) \in o(1). 
\end{align*}

Markov's inequality~\Cref{lem:markov} with $a=\sqrt{\log(n))}\E{|V \cap B_0(r_0)|}$ finishes the proof.
\end{proof}

\subsection{Leaves Lemma}\label{sec:leaves}

In this section we introduce a property of HRGs we will make use of in particular for the lower bounds on the number of rounds (\Cref{pro:not-two-rounds-whp}) and for the probability that \RCT does not terminate with a too small colour space (\Cref{pro:dist-coul-never}). We will also us it later in~\Cref{sec:rctid} when we analyse \RCTID. We show a lot of vertices with many neighbours of degree one to which we refer as \emph{leaves}. This is in particular useful in the context of randomised colouring since leaves have only one neighbour and thus, there are no dependencies for the candidate colour an uncoloured leaf tries in each round. The following will be useful to show our desired properties of leaves.

\begin{lemma}[Layer-leaves]\label{lem:leaves}
Let $G\sim\hrg$ be a threshold hyperbolic random graph and let $u \in V(G)$ be a vertex in $G$ with radius $r \in R - \omega(1)$. Then with constant probability $\Omega(\deg(u))$ neighbours of $u$ are leaves in layer $\mathcal{A}_0$.   
\end{lemma}

\begin{proof}
    Recall that $\mathcal{A}_0 = \B_0(R)\setminus\B_0(R-1)$. For vertex $u$ let us write $ N_{0}(u) := N(u) \cap \mathcal{A}_0$, i.e., the subset of the neighbourhood of $u$ located in layer $\mathcal{A}_0$. Then, by \Cref{eq:angle-func}, all the vertices with angle distance $\theta_R(r,R)$ to $u$ are neighbours of $u$ since $\theta_R(r,x)$ decreases monotonically in $x$. Then the expected size of $|N_0(u)|$ is 
        $$\E{|N_0(u)|} \geq n\cdot\mu(\B_0(R)\setminus\B_0(R-1))\cdot 2\theta_R(r,R).$$

Using \Cref{lem:max-angle} and \Cref{eq:layer-measure} for the angle and the measure of the layer we get
    
\begin{equation}\label{eq:leaves}
\E{|N_0(u)|} \geq \Theta(1)e^{(R-r)/2} = \Theta(1)\E{\deg(u)} \in \omega(1),
\end{equation}
where for the equation we used \Cref{lem:vertex-degree} and in the last step we used that $r(u) \in \omega(1)$. Next, consider any vertex $v \in N_0(u)$ and let $Z_v$ be the indicator random variable that is $1$ if $\deg(v) = 1$ and $0$ otherwise and we write $Z = \sum_{v \in N_0(u)}Z_v$ for the number of leaves of $u$ that have radius at least $R-1$. Since the degree of $v$ follows a Poisson distribution and the expected degree of a vertex with radius at most $R-1$ is constant using \Cref{lem:vertex-degree}, we obtain that $\Prob{Z_v = 1} \in \Omega(1)$.

Moreover, applying a Chernoff bound in conjunction with \Cref{eq:leaves} we obtain $|N_0(u)| \in \Omega(\deg(v))$ with constant probability. Thus, by linearity of expectation and law of total expectation we get $\E{Z} \in \Omega(\deg(u)) \in \omega(1)$. Let us now write $Y = |N_0(u)| -Z$ for the number of "none-leaves" for which we get by linearity of expectation that there exists a constant $0 < \delta < 1$ such that $\E{Y} - \E{|N_0(u)|} -\E{Z} \leq \delta\E{|N_0(u)|}$. Finally let $\varepsilon < 1 -\delta$ while being non-vanishing by which we get by an application of Markov's inequality for the random variable $Z$ that
\begin{align*}
    \Prob{Z \in \Omega(\deg(u))} \geq 1 - \Prob{Y \geq \underbrace{(\delta+\varepsilon)}_{<1}|N_0(u)|} \geq 1 - \delta/(\varepsilon + \delta),
\end{align*}
which is non-vanishing by our choice of $\varepsilon > 0$. This proves our claim.
\end{proof}

We now present the \emph{leaves lemma}. Recall that $V_\ell$ is the set of vertices located in layer $\mathcal{A}_\ell$, i.e., $V_\ell = V \cap \mathcal{A}_\ell$ and we write $L(u) := \{v \in N(u) : \deg(v) = 1\}$ for the leaves of $u$. The proof is deferred to \Cref{sec:leavesproof}

\begin{restatable}[Leaves lemma]{lemma}{leavesLemma}\label{lem:leaf}
      Let $G\sim\hrg$ be a threshold hyperbolic random graph and let $c>0$ be constant large enough. Then, for all integers $\startlayer$ such that $\lceil c\log(\log(n))\rceil \leq \startlayer \leq \lfloor \frac{\log(n) - \log(\log(n)) - c }{\alpha} \rfloor$, \aas there exists a set $U_{\startlayer} \subseteq V_{\startlayer}$ where
    \begin{enumerate}
        \item for any $u \in U_{\startlayer}$, a constant fraction of $u$'s neighbours are leaves $|L(u)| \in \Theta(\deg(u))$ and
        \item the size of the set $U_{\startlayer}$ is $|U_{\startlayer} | \in \Theta(|V_{\startlayer}|)$.
    \end{enumerate}
\end{restatable}

\begin{proof}[Proof sketch.]
    We fix a layer $\startlayer$ and by \Cref{lem:leaves} the lemma statement holds in expectation, i.e., in expectation a constant fraction of the vertices in layer $\mathcal{A}_{\startlayer}$ have a constant fraction of its neighbours as leaves that are located in $\mathcal{A}_0$. To obtain the desired concentration, we consider a subset of vertices $U_{\startlayer} \subseteq {V}_{\startlayer}$ that fulfils the following: for any pair of vertices $u,v \in U_{\startlayer}$, let $\B_u$ and $\B_v$ in the disk $\disk$ be the areas in the disk $\disk$ that we need to reveal for the events that $u$ has $\Theta(\deg(u))$ leaves in $\mathcal{A}_0$ and $v$ has $\Theta(\deg(v))$ leaves in $\mathcal{A}_0$. We ensure that these two areas are disjoint, i.e., $\B_u \cap \B_v = \emptyset$. We accomplish this by having an angle distance for $u,v \in U_{\startlayer}$ large enough such that for any point $x \in \mathcal{A}_0 \cap \B_u(R)$ we have $B_x(R) \cap (\B_v(R)\setminus \B_0(R-\startlayer-1)) = \emptyset$. Moreover, for potential leaves we only consider a sub area $\mathcal{A}'_0 \subseteq \mathcal{A}_0$ such that for any point $a \in \mathcal{A}'_0$ and any vertex $w \in V \cap \B_0(R-\startlayer-1)$, the distance is $\dist(a,w)>R$ and consequently, the probability that a vertex in the area $\mathcal{A}'_0$ is a leaf-neighbour of $u \in U_{\startlayer}$ is independent of the set of vertices $V \cap \B_0(R-\startlayer-1)$ since this set of vertices has no neighbours in $\mathcal{A}'_0$. Thus, for any pair of vertices $u,v \in U_{\startlayer}$, the area on the disk $\disk$ that we reveal for the events that $u$ and $v$ have a constant fraction of their neighbours as leaves are disjoint. Then, for each vertex $u \in U_{\startlayer}$ it is an independent coin flip if $u$ has $\Theta(\deg(u))$ many leaves and the concentration follows from Chernoff bound with $|U_{\startlayer}| \in \omega(\log(n))$ coin flips and a union bound over all layers yields the desired result. 
\end{proof}

\subsection{Larger Degree Neighbourhood}\label{sec:larger-degree-neighbourhood}
For the analysis of \RCTDEG we are interested in the neighbourhood of $u$ that has a larger degree than $u$. The following lemma tells us up to which radius $r'$ there might be vertices with a larger degree than than vertex $u$ with radius $r(u) =r$. For an illustration we refer to \Cref{fig:larger-degree-neighbourhood}a.

\begin{lemma}[Larger degree radius]\label{lem:larger-neighbour-radius}
    Let $G\sim\hrg$ be a threshold hyperbolic random graph and let $c' = 7$. Moreover, let $u \in V$ with radius $r \leq R - \frac{2\log(\log(n))}{(1-\alpha)}$ and degree $\deg(u)$. Then, \wehp all vertices with radius $r' \geq r + c'$ have a smaller degree than $\deg(u)$.
\end{lemma}

\begin{proof}
    Consider the expected degree of $u$, which given via~\Cref{lem:vertex-degree} is at least
$$
\E{\deg(u)} \geq (1-o(1))n\cdot \frac{\alpha e^{-r/2}}{\pi(\alpha - 1/2)} \in \Omega(\log^2(n)),
$$

where the second step followed from the assumption of our case that $r \geq R - 2/(1-\alpha)\log(\log(n))$ and $R = 2\log(n) +C$. From a Chernoff bound it then follows that $\deg(u) \geq \E{\deg(u)}/2 \geq (1-o(1))n\cdot\frac{\alpha e^{-(R-\ell)/2}}{2\pi(\alpha-1/2)} = : k$ \wehp using that $\E{\deg(u)} \in \Omega(\log^2(n))$. 

In contrast, consider any vertex $v$ with radius $r' \geq r + c'$. By \Cref{lem:vertex-degree} we get for the expected degree of vertex $v$ 
$$
\E{\deg(v)} \leq (1+o(1))n\cdot\frac{\alpha e^{-(r + c')/2}}{\alpha-1/2} \leq e^{-3/2}(1+o(1))n\cdot\frac{\alpha e^{-(R-\ell)/2}}{2\pi(\alpha-1/2)} ,
$$

Another Chernoff bound then yields that $\deg(v) \leq e^{3/2}\cdot \E{\deg(v)} < k$ \wehp Though the degrees $\deg(u)$ and $\deg(v)$ might be non-independent, a union bound over the complementary events for both degrees gives $\deg(u) > \deg(v)$ \wehp A union bound over all vertices with radius $r'\geq r + c$ finishes the proof.
\end{proof}

Let $N^+(u):= \{v \in N(u) | \deg(v) \geq \deg(u)\}$ be the neighbourhood of $u$ with larger degree and $\deghigh{u} = |N^+(u)|$. We bound $\deghigh{u}$ parameterised by layer level $\ell$ (see also \Cref{fig:larger-degree-neighbourhood}b)

\begin{lemma}[Larger degree neighbourhood]\label{lem:neighbours-larger-degree}
Let $G\sim\hrg$ be a threshold hyperbolic random graph and let $u \in V_\ell$. Then, \wehp, 
$$\deghigh{u}\in\begin{cases}
{\bigO}(e^{\ell/2} + \log(n)) \text{ for $0 \leq \ell \leq \lfloor 2/(1-\alpha)\log(\log(n))\rfloor$},\\
\bigO(e^{\ell(1-\alpha)}) \text{ for $\lceil 2/(1-\alpha)\log(\log(n))\rceil \leq \ell \leq \lceil R/2 \rceil$ },\\
\bigO(ne^{-\alpha\ell}) \text{ for $ \lceil R/2 \rceil \leq \ell \leq \lceil \alpha^{-1}(\log(n) - 2\log(\log(n))) \rceil$}.
\end{cases} $$ 
\end{lemma}
\begin{proof}
\textbf{Case 1}~[$0 \leq \ell \leq \lfloor  2/(1-\alpha)\log(\log(n))\rfloor$]: This follows directly from the upper bound on the degree by \Cref{lem:layer-properties}.

\smallskip

\textbf{Case 2}~[$\lceil  2/(1-\alpha)\log(\log(n))\rceil \leq \ell \leq \lceil R/2 \rceil$]:

Fix a vertex $u \in V_\ell$ and recall that \wehp all vertices with radius $r' \geq r + 7$ have a smaller degree than $\deg(u)$ due to \Cref{lem:larger-neighbour-radius}. Let $A$ be the event that no neighbour of $u$ with radius larger than $r'$ has a larger degree than $\deg(u)$. Since $u \in V_\ell$ has maximal radius $R-\ell$, we obtain by using law of total expectation and \cite[Lemma 3.3]{Krohmer2016} that

\begin{align*}
    \Exp{\deghigh{u}} &=  \Exp{\deghigh{u}|A}\Pro{A} +  \Exp{\deghigh{u}|A^C}\Pro{A^C}\\
    &\leq \Theta(1) e^{\ell(1-\alpha)} +  o(1) \in \omega(\log(n)).
\end{align*}
since $\Pro{A^C} \in n^{-\omega(1)}$ and we used that by our case that $\ell\geq 2/(1-\alpha)\log(\log(n))$ in the last step. A Chernoff bound concludes the case.

\smallskip

\textbf{Case 3}~[$\lceil R/2 \rceil \leq \ell \leq  \lceil \alpha^{-1}(\log(n) - 2\log(\log(n))) \rceil$]: Fix a vertex $u$ with radius of at most $r \leq R/2$ and by \Cref{lem:larger-neighbour-radius} all vertices with radius $r' \geq r + 7$ have a smaller degree than $u \in V \cap \B_0(R/2)$ \wehp Via event $A$ similar to case 2, we obtain by \Cref{lem:measure-inner-disk} and law of total expectation that 

\begin{align*}
    \Exp{\deghigh{u}} &=  \Exp{\deghigh{u}|A}\Pro{A} +  \Exp{\deghigh{u}|A^C}\Pro{A^C}\\
    &\leq n\cdot\mu(\B_0(R/2 + c') +  o(1)\\
    &\leq \Theta(1)ne^{-\alpha\ell}.
\end{align*}
Then we use $\ell \geq \lceil \lceil \alpha^{-1}(\log(n) - 2\log(\log(n))) \rceil$, and a Chernoff yields the desired result.

\end{proof}

\section{Analysis of Random Colour Trial (Proof of Theorem~\ref{thm:main-theorem})}\label{sec:rct}

In this section we analyse \RCT. In~\Cref{sec:round-one} we analyse the first round and we show that the colour space of vertices remains relatively large~(\Cref{lem:slack-rct}), the progress we make on uncoloured vertices per layer~(\Cref{lem:round-one-aas}) and a drop on the maxmum uncoloured degree in \Cref{cor:max-uncoloured-degree}. \Cref{sec:round-two} then contains the analysis of the second round where we show that \RCT finishes \aas after two rounds in~\Cref{pro:rct-second-round}. Our lower bounds on the colour space for \RCT are proven in \Cref{sec:colour-space}. We conclude in \Cref{sec:together} by putting everything together to show \Cref{thm:main-theorem}.

\subsection{First round: Maintaining Palette size while reducing Uncoloured Degree}\label{sec:round-one}

In this section we analyse the first round of \RCT. In order to show that the algorithm terminates, we need to ensure that there are enough colours available to every vertex. To make this more formal, let $X_t(u)$ be the random variable that counts the number of colours that are not assigned to any neighbour of $u$, before round $t$ starts and we define 
\begin{align}\label{eq:slack-counter}
    \colourSpace{t}{u} := |\colourpal{t}{u}| = \sum_{\colour \in \palette}\prod_{v \in N(u)} \left(1 - \indicator{\colourchosen{t-1}{v} = \colour}\right).
\end{align} 
That is $\colourSpace{t}{u}$ tells us how large the colour space of $u$ at round $t$ is.

\begin{lemma}[Slack round 1 \RCT]\label{lem:slack-rct}
    Let $G\sim \hrg$ be a typical hyperbolic random graph and consider round $t=2$ of \RCT with colour space $|\palette| \in \omega(\log(n))$. Then, for any vertex $u \in V$ with degree $\deg{(u)} \in \bigO({|\palette|})$, the colour space of $u$ is at the second round is $\colourSpace{2}{u}\in \Omega(|\palette|)$~\wehp
\end{lemma}
\begin{proof}
First consider a vertex $u \in V$ that has degree $\deg(u) \in o(|\palette|)$. Then the statement immediately follows deterministically as $\colourSpace{t}{u} \geq |\palette| -  \deg(u)$ since at most every neighbour of $u$ can have one different colour assigned and it follows for any $t$
\begin{align*}
    \colourSpace{t}{u} \geq |\palette| - \deg(u) \in \Omega(|\palette|).
\end{align*}

Thus, to finish the proof, it is left to show that our statement holds for any vertex $u$ where the degree is $\deg(u) \in \Theta(|\palette|)$. To this end, fix a vertex $u$ with degree $\deg(u) \in \Theta(|\palette|)$ and consider the first round of our colouring algorithm \RCT. Since every vertex picks a colour uniform at random from $[|\palette|]$, we get for any $v \in N(u)$ that the probability that it tries a specific candidate colour $\colour$ is $p_v: = \Pro{v \text{ tries } \colour \text{ as its candidate in round 1.}} =|\palette|^{-1}$. Next, we count the number of colours that were not tried by any neighbour of $u$. To this end, fix a colour $\colour$ and let $Z_{\colour}$ be the number of times that $\colour$ was tried by a vertex $v \in N(u)$. By a binomial distribution we obtain
\begin{align}\label{eq:uncoloured-indicator}
   \Pro{Z_{\colour} = 0} = (1-p_v)^{\deg(u)} = (1-|\palette|^{-1})^{\Theta(1)|\palette|} \in \Omega(1).
\end{align}
This shows that with constant probability $p > 0$, a colour $\colour \in \palette$ is not picked by any neighbour of $u$. Modelling the amount of colours picked by neighbours of $u$ by a balls into bins experiments, where $|\palette|$ many bins represent the colours and $\deg(u)$ many balls represent the neighbours of $u$, let $Z'_\colour$ the indicator random variable that is $1$ if there is at least one ball in bin $\colour$. Moreover, we "Poissonise" the experiment, by throwing $N \sim \text{Po}(\Exp{\deg(u)})$ balls into $|\palette|$ bins. Let $Z' = \sum_{\colour \in \palette}Z'_{\colour}$ be the number of bins with at least one ball. Then the expected number of balls in bin $\colour$ is $\Exp{Z'_\colour} < 1 - p$ by \Cref{eq:uncoloured-indicator} and we obtain by linearity of expectation $\Exp{Z'} \leq (1-p)|\palette|$. Let $Z$ be the number of bins without a ball and by linearity of expectation we have $\Exp{Z} = \Exp{|\palette| - Z'} \geq p\cdot|\palette|$. Since $Z$ follows a Poisson distribution and $|\palette| \in \omega(\log(n))$ we obtain via Chernoff bound that $Z \in \Theta(|\palette|)$ \wehp Then for the balls into bins experiment under binomial distribution with exactly $\deg(u)$ balls it follows that there are also \wehp $Z\in\Theta(|\palette|)$ empty bins, i.e., non-used colours by any neighbour of $u$ which follows from "de-Poissonising" the experiment~\cite[Corollary 5.9]{mu-pc-05}, which says that if an event holds \wehp in the Poisson case, then it also holds \wehp in the binomial case. A union bound over all vertices with degree $\Theta(|\palette|)$ finishes the proof.
\end{proof}

Recall that $V_\ell = V \cap \mathcal{A}_\ell$ are the vertices contained in a layer~(see \Cref{sec:layers}). In the next lemma we give a bound on the random variable $Z_\ell = |\{v \in V_\ell |{\colourchosen{1}{v} =\emptyset}\}|$, with which we count how many vertices are uncoloured in any layer $\mathcal{A}_\ell$ after round 1 has terminated.

\begin{lemma}[\RCT uncoloured layer vertices]\label{lem:round-one-aas}
    Let $G\sim\hrg$ by a typical hyperbolic random graph and, let $Z_\ell = \sum_{v \in V_\ell} \indicator{\colourchosen{1}{v} =\emptyset}$. Then, for \RCT with colour space $|\palette| \in \Omega(\Delta)$, for any $c>0$ it holds $Z_\ell \in \Tilde{\bigO}(e^{-\ell(\alpha - 1/2)} \cdot n^{\frac{2\alpha - 1}{2\alpha}} + 1)$ with probability $1 - \bigO(\log^{-c}(n))$.
\end{lemma}

\begin{proof}
Fix a vertex $u \in V_\ell$ and let $Z_u$ be the indicator random variable that is $1$ if $u$ is uncoloured after the first round of \RCT. Reveal the candidate colour of $u$ first such that before revealing the candidate colours of $u$'s neighbour we have $\Pro{Z_u} = (1 - |\palette|^{-1})^{\deg(u)}$ where $|\palette|$ is the colour space of \RCT. By the assumption of the lemma that $|\palette| \in \Omega(\Delta)$ and that $G$ is a typical HRG, (so by $\Cref{lem:layer-properties}$ $\deg(u) \in \Tilde{\bigO}(e^{\ell/2})$ since $u \in V_\ell$), it follows that $\Pro{Z_u} \in \Tilde{\bigO}(e^{\ell/2}/\Delta)$. Then the expectation of the random variable $Z_\ell$, which is the number of uncoloured vertices of the set $U_\ell$ after round 1 of \RCT, is by linearity of expectation $\Exp{Z_\ell} \in  \Tilde{\bigO}( |V_\ell| \cdot e^{\ell/2}/\Delta) $.

Since $G$ is atypical HRG, we have $|V_\ell| \in \Tilde{\bigO}(ne^{-\alpha\ell} + 1)$ by \Cref{lem:layer-properties} and $\Delta \geq \log^{-2}(n)n^{\frac{1}{2\alpha}}$ by \Cref{the:max-degree} and it follows $\Exp{Z_\ell} \in \Tilde{\bigO}(e^{-\ell(\alpha - 1/2)} \cdot n^{\frac{2\alpha - 1}{2\alpha}} + 1)$. An application of Markov's inequality~\Cref{lem:markov} with $a=\Exp{Z_\ell}\cdot\log^c(n) + \log(n)$ where $c$ is an arbitrarly large constant then yields the desired concentration.

\end{proof}

We bound the total number of all uncoloured vertices after the first round of \RCT. For this we write $Z:= \sum_{v \in V} \indicator{\colourchosen{1}{v} = \emptyset}$ for the random variable that is the number of uncoloured vertices after the first round. 

\begin{lemma}[Uncoloured vertices round 1]\label{lem:uncloured-layers-aas}
    Let $G\sim\hrg$ be a typical hyperbolic random graph and let $c$ be any constant. Then the number of uncoloured vertices after the first round of \RCT with colour space $\Omega(\Delta)$ is at most $Z \in \Tilde{\bigO}\left(n^{\frac{2\alpha - 1}{2\alpha}}\right)$ with probability $1 - \bigO(\log^{-c}(n))$.     
\end{lemma}

\begin{proof}
 
    Fix a layer with index $\ell$ and let $Z_\ell$ be the number of uncoloured vertices in layer $\mathcal{A}_\ell$. Then using \Cref{lem:uncloured-layers-aas} with $c' = c + 1$ we obtain $Z_\ell \in \Tilde{\bigO}(e^{-\ell(\alpha - 1/2)} \cdot n^{\frac{2\alpha - 1}{2\alpha}} + 1)$ with probability $1 - \bigO(\log^{-(c+1)}(n))$.

    Notice that, since $\alpha > 1/2$, it follows $Z_\ell \in \Tilde{\bigO}\left(n^{\frac{2\alpha - 1}{2\alpha}}\right)$ with probability $1 - \bigO(\log^{-(c+1)}(n))$. Then we use a union bound over all $\bigO(\log(n))$ layers and it holds simultaneously for all layers that $Z_\ell \in \Tilde{\bigO}\left(n^{\frac{2\alpha - 1}{2\alpha}}\right)$ with probability $1 - \bigO(\log^{-c}(n))$.
    
The result then follows since there are $\bigO(\log(n))$ many layers and thus, with probability $1 - \bigO(\log^{-c}(n))$, we have $$Z \in \sum_{\ell \in [\bigO(\log(n))]} Z_\ell \in \Tilde{\bigO}\left(n^{\frac{2\alpha - 1}{2\alpha}}\right).$$ 
\end{proof}

We write $\maxUncoul{t}:=\max_{u \in V}\left(\sum_{v \in N(u)}\indicator{\colourchosen{t-1}{v} = \emptyset}\right)$, i.e., $\maxUncoul{t}$ is the \emph{maximum uncoloured degree} before round $t$ starts. For round $t = 2$ of \RCT, we upper bound $\maxUncoul{t}$. Notice that $\maxUncoul{t=2}$ is upper bounded by the number of all uncoloured vertices after the first round which is bounded by~\Cref{lem:uncloured-layers-aas}. By this we get the following corollary.

\begin{corollary}[Maximum uncoloured degree after \RCT round 1]\label{cor:max-uncoloured-degree}
     Let $G\sim\hrg$ be a typical hyperbolic random graph and $\maxUncoul{t=2}$ the maximum uncoloured degree after the first round of \RCT with colour space $\Omega(\Delta)$. Then, for any constant $c$, 
 $\maxUncoul{t=2} \in \Tilde{\bigO}\left(n^{\frac{2\alpha - 1}{2\alpha}}\right)$ with probability $1 - \bigO(\log^{-c}(n))$.
\end{corollary}

\subsection{Theorem~\ref{thm:main-theorem}, part 1: Second Round }\label{sec:round-two}
In this section we analyse the second round of \RCT. We use our bounds for the uncoloured degree of any vertex after round 1 and the number of uncoloured vertices per layer after round 1 as ingredients to prove that after the second round the expected number of vertices in any layer is vanishing which then implies that \RCT terminates after two rounds \aas~(\Cref{pro:rct-second-round}). Moreover, we show that \RCT does not terminate earlier than after two rounds \wehp~(\Cref{pro:exactly-two-rounds}) and that a \whp guarantee after two rounds is not possible (\Cref{pro:not-two-rounds-whp}).

We now bound the number uncoloured vertices after the second round, which is the random variable $Z := \sum_{v \in V}\indicator{\colourchosen{2}{v} = \emptyset}$. 

\begin{proposition}[\RCT two rounds \aas]\label{pro:rct-second-round}
 Let $G\sim\hrg$ by a typical hyperbolic random graph and let $Z$ be the number of uncoloured vertices after round 2 of \RCT  with colour space $\Omega(\Delta)$ has ended. Then, $Z=0$ \aas     
\end{proposition}
\begin{proof}
Fix a layer $\mathcal{A}_\ell$ with index $\ell$ and let $U_\ell \subseteq V_\ell$ be the set of uncoloured vertices in layer $\mathcal{A}_\ell$ after round one has ended. Now let $A_1$ be the event that $|U_\ell|\in \Tilde{\bigO}(e^{-\ell(\alpha - 1/2)} \cdot n^{\frac{2\alpha - 1}{2\alpha}} + 1)$ which by \Cref{lem:round-one-aas} holds with probability $\Pro{A_1} = 1 - \bigO(\log^{-2}(n))$. Moreover, let $A_2$ be the event that the maximum uncoloured degree is bounded by $\maxUncoul{t=2} \in \Tilde{\bigO}\left(n^{\frac{2\alpha - 1}{2\alpha}}\right)$ which by \Cref{cor:max-uncoloured-degree} has probability $\Pro{A_2} = 1 - \bigO(\log^{-2}(n))$. Finally, let $A_3$ be the event that the minimal colour space of a vertex is $X := \min_{u\in U_\ell}\colourSpace{t=2}{u} \in \Omega(\Delta)$ which by \Cref{lem:slack-rct} occurs with probability $\Pro{A_3} =1 - \bigO(\log^{-2}(n))$ since $\Delta \in \omega(\log(n))$ for a typical HRG by~\Cref{the:max-degree}. Let us write $A: = \{A_1 \cap A_2 \cap A_3\}$ and notice that by union bound over the complementaries $\Pro{A} =1 - \bigO(\log^{-2}(n))$. 

Next, let $Z_\ell:=\sum_{u \in U_\ell}\indicator{\colourchosen{2}{u} = \emptyset}$ be the number of uncoloured vertices in layer $\mathcal{A}_\ell$ after round 2 has ended and for $u \in V_\ell$, we upper bound $\Exp{\indicator{\colourchosen{2}{u} = \emptyset} |A}$. For this, consider a vertex $u \in U_\ell$ with uncoloured degree in round 2, for which we write $|\{v \in N(u) : \colourchosen{1}{v} = \emptyset\}| =:\deguncoul{2}{u}$. Moreover, let $X$ be the minimal colour space of any neighbour of $u$. Then, reveal $u$'s candidate colour $\colour$ and the probability that any uncoloured neighbour of $u$ has the same candidate colour $\colour$ is $\Exp{\indicator{\colourchosen{2}{u} = \emptyset}} \leq (1 - 1/X)^{\deguncoul{2}{u}} \leq \frac{\deg_2(u)}{X}$. Recall that event $A$ implies that the maximum uncoloured degree is $\maxUncoul{t=2} \in \Tilde{\bigO}(n^{\frac{2\alpha - 1}{2\alpha}})$ and that the minimal colour space for every vertex is $X \in \Omega(\Delta)$. Then we get by conditioning on event $A$ that 

\begin{align}\label{eq:expectation-indicator-conditioned}
    \Exp{\indicator{\colourchosen{2}{u} = \emptyset} |A} \leq \Theta(1)\min({\deg(u) , \maxUncoul{t=2}}) \cdot \Delta^{-1} \in \Tilde{\bigO}(\min(e^{\ell/2}, n^{\frac{2\alpha - 1}{2\alpha}})\cdot n^{-\frac{1}{2\alpha}}),
\end{align} 

where the first inequality follows since the uncoloured degree of a vertex $u$ is upper bounded by both, $u$'s degree and the maximum uncoloured degree of all vertices, and the later follows since $G$ is a typical HRG and so $\Delta \geq \log^{-2}(n)\cdot n^{\frac{1}{2\alpha}}$ by \Cref{the:max-degree} and $\deg(u) \in \Tilde{\bigO}(e^{\ell/2})$ by \Cref{lem:layer-properties} since $u \in V_\ell$.

We now bound $\Exp{Z_\ell |A}$, the expected number of vertices after round 2, conditining on event $A$. Recall that if event $A$ occurs, then the number of uncoloured vertices in layer $\mathcal{A}_\ell$ after round 1 is $|U_\ell|\in \Tilde{\bigO}(e^{-\ell(\alpha - 1/2)} \cdot n^{\frac{2\alpha - 1}{2\alpha}} + 1)$. Thus, by conditioning on $A$ we obtain via linearity of expectation
\begin{align}\label{eq:expectation-all-conditioned}
    \Exp{Z_\ell| A} &\in \Tilde{\bigO}(e^{-\ell(\alpha - 1/2)} \cdot n^{\frac{2\alpha - 1}{2\alpha}} + 1) \Exp{\indicator{\colourchosen{2}{u} = \emptyset}|A} \in \Tilde{\bigO}(e^{-\ell(\alpha - 1/2)} \cdot n^{-(1/\alpha - 1)} \cdot \min(e^{\ell/2}, n^{\frac{2\alpha - 1}{2\alpha}})),
\end{align}
where we used \Cref{eq:expectation-indicator-conditioned} in the second step.

We now proceed by a case distinction for $\ell$ to upper bound $Z_\ell$ for any layer. Our desired statement then follows afterwards via a union bound over all $\mathcal{O}(\log(n))$ layers.

\smallskip

\textbf{Case 1}~[$\ell \leq \lfloor(2 - 1/\alpha)\log(n)\rfloor$]:
By the case that $\ell \leq \lfloor(2 - 1/\alpha)\log(n)\rfloor$ we have 
$\Tilde{\bigO}(\min(e^{\ell/2}, n^{\frac{2\alpha - 1}{2\alpha}})) \in \Tilde{\bigO}(e^{\ell/2})$ and subsequently by \Cref{eq:expectation-all-conditioned} it follows for a constant $\delta >0$ that

\begin{align}\label{eq:case_1-for-aas}
\Exp{Z_\ell| A}  \in \Tilde{\bigO}(e^{\ell(1-\alpha)} \cdot n^{-(1/\alpha - 1)}) \in \Tilde{\bigO}(n^{-2(1/\alpha + \alpha -2)}) \in \bigO(n^{-\delta}),   
\end{align}

where we first used that $\alpha < 1$ in conjunction with $\ell \leq \lfloor(2 - 1/\alpha)\log(n)\rfloor$ by our case and then again that $\alpha < 1$ in the last step.

\smallskip

\textbf{Case }~2[$\ell \geq \lceil(2 - 1/\alpha)\log(n)\rceil$]: By the case that $\ell \leq \lceil(2 - 1/\alpha)\log(n)\rceil$ we have 
$\Tilde{\bigO}(\min(e^{\ell/2}, n^{\frac{2\alpha - 1}{2\alpha}})) \in \Tilde{\bigO}(n^{\frac{2\alpha - 1}{2\alpha}})$ and subsequently by \Cref{eq:expectation-all-conditioned} it follows for a constant $\delta >0$ that

\begin{align}\label{eq:case-2-for-aas}
\Exp{Z_\ell| A}  \in \Tilde{\bigO}(e^{-\ell(\alpha - 1/2)} \cdot n^{2 - \frac{3}{2\alpha}}) \in \Tilde{\bigO}(n^{-2(1/\alpha + \alpha -2)}) \in \bigO(n^{-\delta}) ,    
\end{align}

where we first used that $\alpha > 1/2$ in conjunction with $\ell \geq \lceil(2 - 1/\alpha)\log(n)\rceil$ by our case and then in the last step that $\alpha < 1$.

\smallskip

To finish the proof, fix any layer $\mathcal{A}_\ell$ and let $B_\ell$ be the event that in layer $A_\ell$ there is no uncoloured vertex after round 2. Using Markov's inequality~\Cref{lem:markov} it follows $$\Pro{B_\ell |A} =1 - \Pro{Z_\ell \geq 1 | A} \geq 1- \Exp{Z_\ell| A} \in 1 - \bigO(n^{-\delta}) ,$$ 

by \Cref{eq:case_1-for-aas} and \Cref{eq:case-2-for-aas}. Finally, recall that $\Pro{A} = 1 - \bigO(\log^2(n))$ and we obtain $\Pro{B_\ell \cap A}=\Pro{B_\ell |A}\Pro{A} = 1 - \bigO(\log^{-2}(n))$. Considering the complementary event and taking a union bound over $\bigO(\log(n))$ many layers then reveals that no layer has an uncoloured vertex after the second round with probability $1 - \bigO(1/\log(n)))$ what finishes the proof.
\end{proof}

We complement the result of \RCT terminating after 2 rounds \aas using $\Omega(\Delta)$ colours by showing that \RCT with $\bigO(\Delta)$ colours is unlikely to finish earlier than after two rounds.

\begin{proposition}[\RCT no first round termination]\label{pro:exactly-two-rounds}
     Let $G\sim\hrg$ be a typical hyperbolic random graph and run \RCT colour space at most $|\palette| \in \bigO(\Delta)$ on $G$. Then, \wehp \RCT does not terminate after the first round.
\end{proposition}

\begin{proof}
    Consider the giant component $H$ of $G$ and let $n' = |V(H)|$. Then $n' \in \Theta(n)$ \wehp (\cite[Theorem 4]{bfkrz-esa-2023}. Fix any vertex $u \in V(H)$ and run a depth first search (DFS) on $u$ after each vertex picked its candidate colour in the first round of \RCT which is u.a.r. from $\bigO(\Delta)$ colours. Then the DFS exposes $n'-1 \in \Theta(n)$ many edges one-by-one and whenever it reveals the candidate colour of a new vertex $v$, the probability that it tried the same colour in the first round as the neighbour whose edge we used to traverse the graph is $p_v \in \Omega(1/\Delta)$ as we use $\bigO(\Delta)$ many colours. Using linearity of expectation for the random variable $Z$, with which we count number of neighbours with the same colour we obtain $\Exp{Z} \geq \sum_{v \in V(H)\setminus\{u\}} p_v = (n' - 1)\cdot p_v\in \Omega(n/\Delta)$. Using that the largest degree $\Delta$ for a typical hyperbolic random graph is at most $n^{\frac{1}{2\alpha} +o(1)}$ by \Cref{the:max-degree}, we obtain $\Exp{Z} \in n^{\frac{2\alpha - 1}{2\alpha} -o(1)} \in n^{\Omega(1)}$ since $\alpha> 1/2$. Finally, we observe that via our DFS set up, the probability $p_v$ is independent for each vertex $v \in V(H)\setminus\{u\}$ such that we can and we will apply a Chernoff bound by which we get the desired concentration for $Z$. The claim follows since neighbours that have the same candidate colour implies that the algorithm goes to another round.
\end{proof}

The subsequent statement implies that our analysis of \RCT finishing after two rounds \aas is tight in the sense that a \whp guarantee is not possible (since $\alpha > 1/2$). 

\begin{proposition}[\RCT no two rounds \whp guarantee]\label{pro:not-two-rounds-whp}
    Let $G\sim \hrg$ be a typical hyperbolic random graph. Then, in the second round of \RCT, there exists a pair $\{u,v\} \in E(G)$ where $v$ is a leaf that tries the same candidate colour as $u$ with probability at least $n^{-(1/\alpha - 1) - o(1)}$.
\end{proposition}

\begin{figure}[t]
    \centering \includegraphics[height=0.3\textheight]{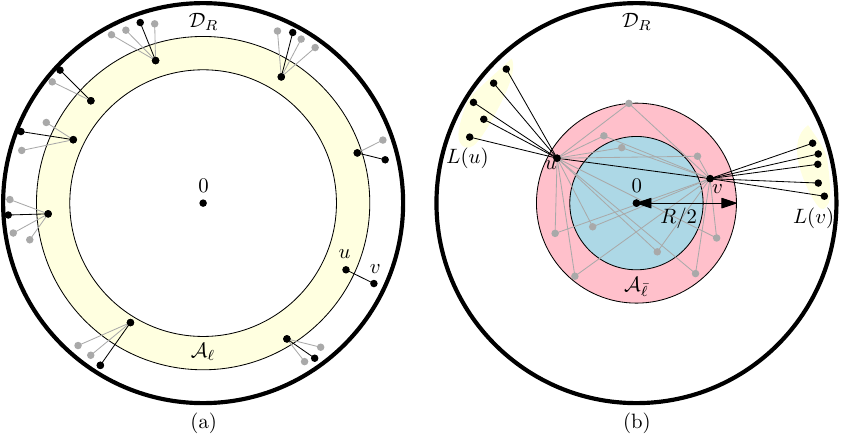}
    \caption{(a) Illustration of vertices in $\mathcal{A}_\ell$, that have the same candidate colour as one of its leaves. (b) Vertices $u$ and $v$ of radius at most $R/2$ that same candidate colour in the first round and many leaves~(\Cref{lem:birthday-paradox}).}
    \label{fig:colour-locked}
\end{figure}

\begin{proof}

Let $c$ be constant such that for $\ell := \ell(c) = c\log(\log(n))$ \Cref{lem:leaf} is applicable. Then, \aas a constant fraction of the vertices of $V_\ell$ have at least one leaf and the property carries over to our typical HRG $G$. Let us denote this set by $U_\ell \subseteq V_\ell$. It then holds
\begin{align}\label{eq:wvhp-many}
    |U_\ell| \geq \Theta(1)|V_\ell| \geq \Theta(1) ne^{-\alpha\ell} = \Theta(1) n\cdot\log^{-\alpha \cdot c}(n),
\end{align}

where we used that ${|V_\ell}| \geq \Theta(1) ne^{-\alpha\ell}$ for a typical HRG by $\Cref{lem:layer-properties}$. Now, let $U' \subseteq U_\ell$ be the set of vertices located in layer $\mathcal{A}_\ell$ where $u \in U_\ell$ has the same candidate colour as one of the leaves in its neighbourhood (see also \Cref{fig:colour-locked}a). We aim to lower bound $|U'|$. To this end, fix a vertex $u \in U_\ell$ and reveal the candidate colour $\colour$ of $u$. Since a vertex tries a candidate colour uniform at random from $\palette$, the probability that a leaf $v \in N(u)$ tries the same candidate colour $\colour$ as $u$ is $|\palette|^{-1} \in \Omega(\Delta^{-1})$. Since ${|U_\ell| \geq \Theta(1) n\cdot\log^{-\alpha \cdot c}(n)}$ for our typical HRG $G$, it follows by linearity of expectation that

\begin{align*}
    \Exp{|U'|} \geq \Theta(1) n\Delta^{-1} \cdot\log^{-\alpha \cdot c}(n) \in n^{1-\frac{1}{2\alpha} - o(1)}\log^{-\alpha \cdot c}(n),
\end{align*}

where we used that $G$ is a typical HRG and thus, $\Delta \in n^{\frac{1}{2\alpha} \pm o(1)}$ by \Cref{the:max-degree}. Since leaves try their candidate colour independently uniform at random, $|U'|$ is a binomial random variable. Thus, we can and we will apply a Chernoff bound and 

\begin{align}\label{eq:size-of-pairs}
    |U'| \in n^{1-\frac{1}{2\alpha} - o(1)}\log^{-\alpha \cdot c}(n) =: n' \text{ } \wehp
\end{align}

We next consider the second round of the \RCT and the probability that there exists a vertex $u \in U'$ that has the same candidate colour as its leaf $v$. Fix a vertex $u \in U'$ and reveal its candidate colour in the second round. Then the probability that leaf $v$ tries the same candidate is $p_v = |\palette|^{-1} \in \Omega(\Delta^{-1})$. Since we have by \Cref{eq:size-of-pairs}~\wehp at least $n'$ many pairs $\{u,v\}$ with $u \in U'$ and $v$ being a leaf that was not coloured in the first round, we can repeat this experiment $n'$ times. Let $Y$ be the random variable that counts the number of pairs $\{u,v\}$ that try the same candidate colour. Since $Y$ is binomial distributed we obtain by law of total probability
\begin{align*}
    \Pro{Y \geq 1} &\geq \Pro{Y = 1 | \text{ } |U'| \in n'}\Pro{|U'| \in n'}\\
    &\geq (1-o(1))) n'\cdot p_v(1-p_v)^{n'} \\
    &\geq \Theta(1) n'\cdot\Delta^{-1}\\
    &\in \Delta^{-1}\cdot n^{1-\frac{1}{2\alpha} - o(1)}\log^{-\alpha \cdot c}(n).
\end{align*}

Using again \Cref{the:max-degree} by which $\Delta \in n^{\frac{1}{2\alpha} \pm o(1)}$ holds for a typical HRG, it follows  that $$\Pro{Y \geq 1} \geq \Pro{Y \geq 1 | \Delta \in n^{\frac{1}{2\alpha} + o(1)}} \in (1-o(1))\Theta(1)n^{-(1/\alpha - 1) - o(1)},$$ and the claim follows.
\end{proof}

\subsection{ Theorem~\ref{thm:main-theorem}, part 2: Limitations of Random Colour Trial}\label{sec:colour-space}
In the previous part we looked at \RCT with $\Theta(\Delta)$ many colours. Ideally, one would like to colour with as few colours as possible close to the chromatic number. In this section we show that, when the colour space is upper bounded by a number that is substantially larger than the chromatic number, \RCT does never terminate on HRGs \whp

To this end we introduce the concept of \emph{colour locked}. We say that a vertex is colour locked after round $t$, if it is not coloured after $t$ rounds, while all but one colour is used by it's neighbours. In other words, the vertex can only select one possible colour in all rounds $t' > t$. We show that a constant fraction of a clique is colour locked after one round if the set of colours is not \emph{too larger}. 

To express our bound, we set $\minexp := \min((2\alpha + 1/2)^{-1}, 2(1-\alpha))$ and we show that in the first round of \RCT on HRGs there exists a pair of vertices that are neighbours and try the same colour in the first round. Moreover, both vertices with the same candidate colour have "many" leaves given that the colour space of the algorithm is at most $\eps/\log(n) \cdot n^\minexp $. Recall that $L(u)$ is the set of neighbours of $u$ that are leaves, i.e., vertices of degree $1$.

\begin{lemma}[Birthday Paradox]\label{lem:birthday-paradox}\label{lem:colour-locked} Let $G\sim \hrg$ be a typical hyperbolic random graph, let $\minexp = \min((2\alpha + 1/2)^{-1}, 2(1-\alpha))$ and let $\eps \in (0,1)$ be constant small enough. Then in the first round of \RCT with colour space at most $|\palette| \leq \epsilon/\log(n) \cdot n^{\minexp}$, there exists a pair of neighbours $\{u,v\}\in E(G)$ such that $u$ and $v$ try the same candidate colour and $|L(u)|, |L(v)| \in \Omega(n^{\minexp})$.
\end{lemma}

\begin{proof}
       Set $\startlayer = \min\left(\frac{2\log(n)}{2\alpha + 1/2},\lceil R/2 \rceil + 1\right)$ and consider the set of vertices in layer $\mathcal{A}_{\startlayer}$, i.e., $V_{\startlayer}$. Notice that  $\mathcal{A}_{\startlayer} \subset \B_0(R/2)$ and thus, the set of vertices $V_\ell$ forms a clique. We show that there exists a pair of vertices $u,v \in V_\ell$, that try the same candidate colour and $|L(u)|, |L(v)| \in \Omega(n^{\minexp})$ which then proves our desired statement (see \Cref{fig:colour-locked}b for a sketch)

       First, we show the desired property that there exists a pair $u,v \in V_\ell$ that try the same candidate colour. To this end consider $U_{\startlayer} \subseteq V_{\startlayer}$ where a constant fraction of their neighbours are leaves. By \Cref{lem:leaf} we have $|U_{\startlayer}| \in \Omega(|V_{\startlayer}|)$ \aas for a hyperbolic random graph and since $G$ is a typical HRG this holds for our input. In turn this implies by \Cref{lem:layer-properties} that

       \begin{align}\label{eq:birthday-equation}
           |U_{\startlayer}| \geq \Theta(1)|V_{\startlayer}| \geq \Theta(1) n^{-\alpha\startlayer}) \in \Omega\left(n^{\min\left(1-\frac{2\alpha}{2\alpha + 1/2}, 1-\alpha\right)}\right),
       \end{align} 
       
       for a typical HRG and by our choice of $\startlayer$. 

       Moreover, in the first round any vertex tries a candidate uniform at random from the entire colour palette $\palette$. For the event $B: = \{\text{$\exists u,v \in U_{\startlayer}: u$ and $v$ try the same candidate colour.}\}$ we then obtain by the birthday paradox (see e.g. \cite[§5.1]{mu-pc-05}),  that
    \begin{align*}
        \Pro{B} &= 1 - \prod_{j=1}^{|U_{\startlayer}|-1}(1- j\cdot |\palette|^{-1})\\
        &\geq 1- \exp{\left(-\frac{\Theta(1)\log(n)\cdot |U_{\startlayer}|^2\cdot n^{-\min((2\alpha + 1/2)^{-1}, 2(1-\alpha))}}{\eps \cdot n^{\min(1/(2\alpha + 1/2), 2(1-\alpha))}}\right)} \tag{$|\palette| \leq \epsilon/\log(n) \cdot n^{\minexp}$} \\
        &\geq 1 - 1/n^{c(\eps)}, \tag{$|U_\ell| \in  \Omega\left(n^{\min\left(1-\frac{2\alpha}{2\alpha + 1/2}, 1-\alpha\right)}\right)$}
    \end{align*}
       for any constant $c(\eps) >0$ given that $\eps>$ is small enough where we used $\minexp = \min((2\alpha + 1/2)^{-1}, 2(1-\alpha))$ and then \Cref{eq:birthday-equation} in the last step. This shows that indeed there exists two vertices $u,v \in U_\ell$ that try the same candidate colour \wehp

       Finally, we show that $u$ and $v$ have $|L(u)|, |L(v)| \in \Omega(n^{\minexp})$ many leaves. To this end fix a vertex $u \in U_{\startlayer} \subseteq V_{\startlayer}$ and recall that $U_{\startlayer}$ is the set of vertices where a vertex $v \in U_{\startlayer}$ has a constant fraction of its vertices as leaves. Then we obtain via \Cref{lem:vertex-degree} for $u \in U_{\startlayer}$ that 
       
       $$\E{\deg(u)} \geq \Theta(1)e^{R-(R-\startlayer)/2}\geq \Theta(1) n^{\min(1/(2\alpha + 1/2), 1/2)} \in \Omega(n^{\minexp}),$$ 
       
       via our choice of $\startlayer = \min\left(\frac{2\log(n)}{2\alpha + 1/2},\lceil R/2 \rceil + 1\right)$ and we used that $\min(1/(2\alpha + 1/2), 1/2) \geq \min(1/(2\alpha + 1/2), 2(1-\alpha)) =: \minexp$ in the last step. Since the degree of a vertex is a Poisson random variable it holds via Chernoff bound $\deg(u) \in \Omega(n^{\minexp})$ \wehp and thus, also for our typical HRG. Moreover, via union bound this holds for all vertices of the set $U_{\startlayer}$. The claim then follows since for any vertex $u \in U_{\startlayer}$ we have $|L(u)| \in \Omega(\deg(u))$. 
\end{proof}

We now show that there exists a pair of vertices that are neighbours and that are colour locked on the same colour after the first round which implies that \RCT never terminates.

\begin{proposition}[Colour locked]\label{pro:dist-coul-never}
    Let $G\sim \hrg$ be a typical hyperbolic random graph, let $\minexp = \min((2\alpha + 1/2)^{-1}, 2(1-\alpha))$ and let $\eps \in (0,1)$ be constant small enough. Then \RCT with at most $|\palette| \leq \eps/\log(n) \cdot n^{\minexp}$ many colours on $G$ never terminates \wehp 
\end{proposition}

       \begin{proof}
Consider the pair of neighbours $\{u,v\} \in E(G)$ that exists \wehp by \Cref{lem:birthday-paradox} where we have for both vertices $|L(u)|, |L(v)| \in \Omega(n^{\minexp})$ many leaves. We show for $u$ that there is at least one leaf $w \in L(u)$ that tries any colour $\colour \in \palette$ in the first round \wehp Using a union bound this also holds for $v$. Then the only colour that is discarded by any of the leaves, is the colour that both $u$ and $v$ try as their respective candidate colour. Thus, $u$ and $v$ are colour locked on the same colour after the first round and the algorithm never terminates as both $u$ and $v$ will try the same colour in each round.

Now, fix $u$ which by \Cref{lem:birthday-paradox} has $|L(u)| \in \Omega(n^{\minexp})$ many leaves. Next, fix a colour $\colour \in \palette$ and define $Z_{\colour}$ as the random variable that counts the number of leaves of $u$, that try candidate colour $\colour$ in the first round. For this we use $Z_w$ for the indicator random variable that is $1$ if $w \in L(u)$ picked the fixed candidate colour in the first round. Since in the first round a vertex tries a candidate colour u.a.r. from $\palette$, we get  by linearity of expectation and $u$ having $|L(u)| \in \Theta(n^{\minexp})$ leaves that 
       \begin{align}\label{eq:leaves-consume-space}
           \Exp{Z_{\colour}} = \Exp{\sum_{w \in L(u)} Z_w}\geq \Theta(1)|L(u)|\cdot|\palette|^{-1} \geq \Theta(1)\eps^{-1}\log(n),
       \end{align}
      
       where we used that $|\palette| \leq \eps/\log(n) \cdot n^{\minexp}$. To wrap things up, we use a Chernoff bound for $Z_{\colour}$ and then a union bound over all colours. Indeed, we obtain that each candidate colour is tried at least once by a leaf with probability $\Pro{Z_\psi \geq 1} \geq 1 - n^{-\Theta(1)/\eps}$ using a Chernoff bound. Hence, under the assumption that $\eps > 0$ is small enough, each candidate colour has been tried at least once by a leaf of $u$ with probability $1 - n^{-c}$ for any constant $c$. Since leaves only have an edge to $u$, this implies that none of the colours but the colour pick of $u$ is discarded after the first round finishes \wehp 

       Via union bound this holds also true for $v$ and thus, \wehp $u$ and $v$ are colour locked on the same colour after the first round and \RCT never finishes proving the desired result.
\end{proof}

\subsection{Putting everything together: Proof of Theorem~\ref{thm:main-theorem}}\label{sec:together}

The combination of our propositions gives us our theorem for \RCT. 
\mainTheorem*

\begin{proof}

Colour space $\eps\cdot\Delta$:

First round: After the first round, \RCT does not finish \wehp by \Cref{pro:exactly-two-rounds}.

Second round: \RCT terminates after the second round \aas using \Cref{pro:rct-second-round} but not \whp~by~\Cref{pro:not-two-rounds-whp} since $\alpha > 1/2$.

Colour space $n^{\delta} \cdot \chi$:

The chromatic number of a hyperbolic random is $\chi \in \Theta\left(n^{1-\alpha}\right)$ \wehp~\cite[Corollary 10]{bmrs-stacs-25}. From this and \Cref{pro:dist-coul-never} our theorem follows since $\minexp > 1 - \alpha$ using that $\alpha < 1$. 
\end{proof}

\section{Analysis of Random Colour Trial with ID Priority (Proof of Theorem~\ref{the:rctid})}\label{sec:rctid}
In this section we enhance our random colour trial algorithm \RCT by using the ID's of vertices. We refer to the algorithm as  random colour trial with ID priority (\RCTID). Similar to \RCT, in each round of \RCTID every uncoloured vertex tries a uniform at random colour from its colour palette. The only difference is that $u$ assigns its candidate colour $\colour$ if and only if no neighbour with a smaller ID tries the same candidate colour $\colour$.

In \Cref{sec:concentration-magic} we make use of the ID's in the sense that it allows us to apply an Azuma-Hoeffding type inequality~(see \Cref{lem:BEPS}) resulting in strong concentrations for the number of uncoloured vertices per layer and the uncoloured degree of a vertex in a layer after the first round. Our result that \RCTID finishes after two rounds with probability $1-\bigO(n^{-c})$ can be found in~\Cref{sec:rctid-part1} where we show this in~\Cref{pro:rctid-upper-bound}. Finally, in~\Cref{sec:rctid-part2} we show that using a polylogarithmic fraction of $\Delta$ colours for \RCTID the algorithm never terminates \wehp~(\Cref{pro:no-termination-rctid}). 

\subsection{Concentration Bounds using ID's}\label{sec:concentration-magic}
In the following let $\deguncoul{t}{u} = |\{v \in N(u) | \colourchosen{t-1}{v} = \emptyset\}|$ be the \emph{uncoloured degree} of vertex $u$ at round $t$ and we refer to $\maxUncoul{\ell} = \max_{u \in U_\ell}\deguncoul{t}{u}$ as the \emph{maximum uncoloured degree} of layer $\mathcal{A}_\ell$. Moreover, recall that $X_t(u) = |\colourpal{t}{u}|$ is the random variable that gives us the colour space available to vertex $u$ at round $t$~(see also \Cref{eq:slack-counter}).

In the following lemma we show a strong concentration for the upper bound of a subset of uncoloured vertices in a layer.

\begin{lemma}\label{lem:magic-ID-concentration}
    Let $t \in \mathbb{N}_{\geq 1}$ be the round of the \RCTID and for $\ell \in [\finallayer]$, let $U_\ell \subseteq V_\ell$  be a subset of uncoloured vertices at round $t$ of size $|U|=n' \geq 1$. Moreover, let $X_\ell = \min_{u \in U_\ell}X_{t}(u)$ and let $\Delta_\ell = \max_{u \in U_\ell}\deguncoul{t}{u}$. Then, for the random variables $Z_t = \sum_{u \in U} \indicator{\colourchosen{t}{u} =\emptyset}$ we have $Z_t \in \Tilde{\bigO}(\sqrt{n'}(1 +  \sqrt{n'}\cdot\Delta_\ell/X_\ell)$ \wehp
\end{lemma}

\begin{proof}
     Reveal the candidate colour of all uncoloured vertices $V\setminus U_\ell$. Then, we order the vertices of $U_\ell$ by their ID's in ascending order. For $u \in U_\ell$ with ID $i$, let $Z_i$ be the indicator random variable that is $1$ if $u$ stays uncoloured after round $t$ is finished. Since a vertex with ID $i$ assigns it candidate colour if and only if no neighbour with ID smaller than $i$ has the same candidate colour, $Z_i$ is uniquely determined by its own candidate colour and the candidate colours of all uncoloured neighbours of $u$ with an ID smaller than $i$. Let $\colour_1, \colour_2,\ldots,\colour_m$ be the candidate colours of all uncoloured neighbours of $u$ which, by assumption of the lemma is at most $m \leq \Delta_\ell$. Then $$\Pro{Z_i = 1 | \colour_1, \colour_2,\ldots, \colour_{\Delta_\ell}} \leq \frac{\Delta_\ell}{X_t(u)} \leq \frac{\Delta_\ell}{X_\ell},$$

using that the colour space of any vertex $u \in U_\ell$ is at least $X_\ell$ by the hypothesis of our lemma. Let $Z_t$ be the number of uncoloured vertices of $U_\ell$ after round $t$ has ended. By linearity of expectation and by $|U_\ell| = n'$ it follows $\Exp{Z_t} \leq n' \cdot\Delta_\ell/X_\ell$. 

To obtain the desired concentration we use an Azuma-Hoeffding bound~\Cref{lem:BEPS} and set $a = \sqrt{n'}\cdot\log(n)$ to obtain

$$
\Pro{Z_t > \Exp{Z_t} + a} = \Pro{Z_t > n' \cdot\Delta_\ell/X_\ell +  \sqrt{n'}\cdot\log(n)}  \leq \exp{\left(-\frac{a^2}{n'}\right)} = \exp{\left(-\frac{\log^2\cdot n'}{n'}\right)},
$$
using that $0 \leq Z_i \leq 1$.  
\end{proof}

The following gives us an upper bound on the number of uncoloured vertices in a layer after the first round of \RCTID has ended.

\begin{lemma}[Round~1 \RCTID uncoloured layer vertices]\label{lem:uncloured-layer-ID}
    Let $G\sim\hrg$ be a typical hyperbolic random graph and let $\ell \in [\finallayer]$ be any level of a layer that contains at least one vertex. Then, after the first round of \RCTID with $|\palette| \in \Omega(\Delta)$ many colours, the number of vertices not coloured in layer $\mathcal{A}_\ell$ is \wehp at most   $Z_\ell \in
			\Tilde{\bigO}(\sqrt{ne^{-\alpha\ell}} + n^{1 -\frac{1}{2\alpha}}e^{-\ell(\alpha -1/2)} +1)$.
\end{lemma}

\begin{proof}
 We use \Cref{lem:layer-properties} and \Cref{lem:magic-ID-concentration} to derive our desired bound: we define the random variable $Z_\ell = \sum_{v \in V_\ell} \indicator{\colourchosen{1}{v} ={\emptyset}}$ to count the number of vertices not coloured after the first round in layer $\mathcal{A}_\ell$. Using \Cref{lem:magic-ID-concentration} and setting $U_\ell := V_\ell$ we have $Z_\ell \in \Tilde{\bigO}(\sqrt{|V_\ell|}+ |V_\ell|\cdot\Delta_\ell/X_\ell)$ \wehp In the following, we bound the three parameters $|V_\ell|$, $\Delta_\ell$ and $X_\ell$.
 
 Since $G$ is a typical HRG, we obtain by \Cref{lem:layer-properties} that $|V_\ell| \in \Tilde{\bigO}(1 + ne^{-\alpha \ell})$.
 
 Recall that $\Delta_\ell$ is the maximum uncoloured degree of a vertex in $U_\ell$. This is upper bounded by the maximum degree of a vertex in $U_\ell$ which for a typical HRG is by \Cref{lem:layer-properties} $\Delta_\ell \leq \max_{u \in U_\ell}\deg(u) \in \Tilde{\bigO}(e^{\ell/2})$.
 
 Finally, since we are in the first round all vertices have a colour space of size $|\palette| \in \Omega(\Delta)$ and thus, $X_\ell \in \Omega(\Delta)$. Using that $G$ is a typical HRG it follows by \Cref{the:max-degree} that $\Delta \geq n^{\frac{1}{2\alpha}}\log^{-2}(n)$, and thus $\Delta_\ell/X_\ell \in \Tilde{\bigO}(e^{\ell/2} \cdot n^{-\frac{1}{2\alpha}})$. Putting this together with $|V_\ell|$ we have obtain \wehp by~\Cref{lem:magic-ID-concentration} that

 \begin{align*}
     Z_\ell \in \begin{cases}
			\Tilde{\bigO}(\sqrt{ne^{-\alpha\ell}}(1+{n^{-(1/\alpha - 1)/2}e^{{\ell}(1-\alpha)/2}}))  , & \text{if $\ell \leq \lceil \alpha^{-1}(\log(n) - 2\log(\log(n))) \rceil$ }\\
            \bigO(\log^2(n)), & \text{otherwise,}
		 \end{cases} 
 \end{align*}

 where we used the upper bound of \Cref{lem:layer-properties} for the second case, which upper bounds the number of uncoloured vertices by all vertices in a layer.
\end{proof}

Recall that $\deguncoul{t}{u}: = |\{v \in N(u) : \colourchosen{t-1}{v} = \emptyset\}|$ is \emph{uncoloured degree} of $u$ at round $t$. We now bound the uncoloured degree of a vertex before round 2 of \RCTID starts. We parameterise the uncoloured degree of a vertex $u$ by its radial coordinate. 

\begin{figure}[t]
    \centering \includegraphics[height=0.3\textheight]{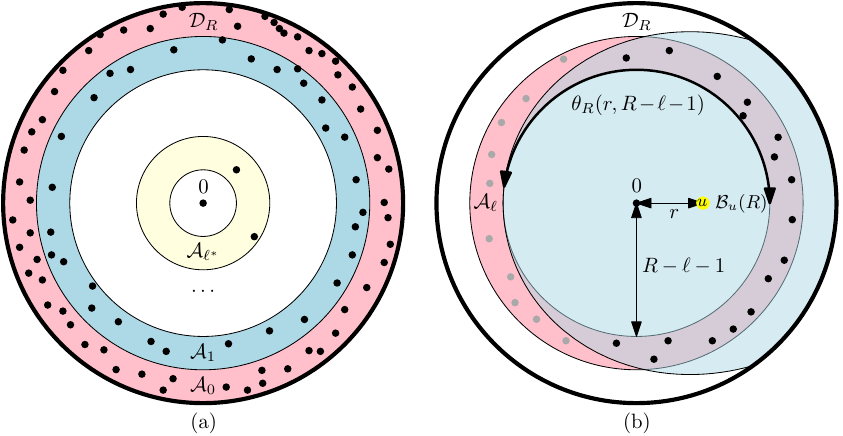}
    \caption{(a)~In \Cref{lem:uncloured-layer-ID} we bound the number of vertices of a layer that are not coloured after round 1. (b)~Sketch of a layer $\mathcal{A}_\ell$ (red area) intersecting the neighbourhood of $u$ (blue area) in \Cref{lem:uncoloured-degree-ID}.}
    \label{fig:uncoloured}
\end{figure}

\begin{lemma}[Uncoloured degree round 1]\label{lem:uncoloured-degree-ID}
Let $G\sim \hrg$ be a typical hyperbolic random graph and consider round $t=2$ of \RCTID with colour space $|\palette| \in \Omega(\Delta)$. Then, for a vertex $u \in V$ with radius $r(u) = r$, it holds $\deguncoul{2}{u} \in \Tilde{\bigO}(\sqrt{n}\cdot e^{-r/4} + n^{\frac{1}{2\alpha}} \cdot e^{- r/2} + 1)$ \wehp
\end{lemma}

\begin{proof}
    We use the discretisation of the disk $\disk$ into layers and we write $N_\ell(u) := V_\ell \cap N(u)$ for the vertices in layer $\mathcal{A}_\ell$ that have an edge to $u$. We aim to apply \Cref{lem:magic-ID-concentration} for each layer to bound the random variables $Y_\ell = |N_\ell(u)|$ which are the number of neighbours of $u$ located in layer $\mathcal{A}_\ell$, that are not coloured after the first round. The proof is then finished by summing over all layers to obtain $\deguncoul{2}{u}$. 
    
    Fix a layer with index $\ell$ and we then have 
    
    $$\E{|N_\ell(u)|} \leq \E{|V_\ell|}\cdot2\theta_R(r, R-\ell-1) = n\cdot\mu(\mathcal{A}_\ell)\cdot2\theta_R(r, R-\ell-1),$$ using that $\theta_R(\cdot,\cdot)$ is monotonically decreasing in both its arguments and \Cref{eq:layer-measure} (see also \Cref{fig:uncoloured}b). Using \Cref{lem:max-angle} and \Cref{eq:layer-measure} we then obtain
    $$
    \E{|N_\ell(u)|}\leq \Theta(1)ne^{-\alpha\ell} \cdot e^{(R- (R-\ell - 1) -r)/2} \leq \Theta(1)n e^{-r/2 - \ell(\alpha -1/2)}.
    $$

    Since $|N_\ell(u)|$ follows a Poisson distribution we can and we will apply a Chernoff bound and obtain $n' = |N_\ell(u)| \in \Tilde{\bigO}(1 + n e^{-r/2 - \ell(\alpha -1/2)})$ \wehp Next, recall that for any $v \in N_\ell(u)$, it holds  $\deg(v) \in \Tilde{\bigO}\left(e^{ \ell/2}\right)$ \wehp using \Cref{lem:layer-properties} which is an upper bound on the largest uncoloured degree $\Delta_\ell \leq \max_{v \in U_\ell}\deg(v)$ for the set $U_\ell$ for a typical HRG. Moreover, since we are in the first round, the minimal colour space is given by $X_\ell = |\palette| \in \Omega(\Delta)$.

    Now, we obtain by \Cref{lem:magic-ID-concentration} that $Y_\ell \in \Tilde{\bigO}(\sqrt{n'}(1 +  \sqrt{n'}\cdot\Delta_\ell/X_\ell))$ \wehp Plugging in our values we obtain $Y_\ell \in \Tilde{\bigO}(\sqrt{n} \cdot e^{-\ell/2(\alpha - 1/2) - r/4} + n^{1-\frac{1}{2\alpha}}\cdot e^{\ell(1-\alpha) - r/2} +1)$ \wehp using \Cref{the:max-degree} for a lower bound on $\Delta$ for a typical HRG. Since we have at most $\mathcal{O}(\log(n))$ many layers, via union bound $Y_\ell \in \Tilde{\bigO}(\sqrt{n} \cdot e^{-\ell/2(\alpha - 1/2) - r/4} + n^{1-\frac{1}{2\alpha}}\cdot e^{\ell(1-\alpha) - r/2} +1)$ holds simultaneously for all layers \wehp Thus, it follows by $\deguncoul{2}{u} \leq \sum_{\ell = 0}^{\finallayer}Y_\ell$ that \wehp

    \begin{align*}
    \deguncoul{2}{u} \leq &\sum_{\ell = 0}^{\finallayer}Y_\ell \in \sum_{\ell = 0}^{\finallayer} \Tilde{\bigO}(\sqrt{n} \cdot e^{-\ell/2(\alpha - 1/2) - r/4} + n^{1-\frac{1}{2\alpha}}\cdot e^{\ell(1-\alpha) - r/2} +1)    \\
    &\in \Tilde{\bigO}\left(\sum_{\ell = 0}^{\finallayer}(\sqrt{n} \cdot e^{-\ell/2(\alpha - 1/2) - r/4}) + \sum_{\ell = 0}^{\finallayer}(n^{1-\frac{1}{2\alpha}}\cdot e^{\ell(1-\alpha) - r/2}) +1\right)\\
    &\in \Tilde{\bigO}(\sqrt{n}\cdot e^{-r/4} + n^{1 - \frac{1}{2\alpha}} \cdot e^{\finallayer(1-\alpha) - r/2} + 1),
    \end{align*}

    where the last line follows since $\alpha \in (1/2,1)$. The result follows by using $\finallayer \leq \lceil\alpha^{-1}{\log(n) + \log(\log(n))}\rceil$ for a typical HRG by \Cref{lem:layer-properties}. 

    \end{proof}

\subsection{Theorem~\ref{the:rctid}, part 1: Colouring with Maximum Degree Colours}\label{sec:rctid-part1}

In this section we show that \RCTID with a colour space of size $\eps\cdot\Delta$ finishes after 2 rounds with an error guarantee that is vanishing polynomially in $n$. In the following lemma we show this for a spcific layer.

\begin{lemma}[Second round \RCTID]\label{lem:RCTI-second-round}
     Let $G\sim \hrg$ be a typical hyperbolic random graph and run \RCTID with $\Omega(\Delta)$ many colours on $G$. Let $\ell \in [\finallayer]$ and let $Z_\ell := |\{v \in V_\ell : \colourchosen{2}{v} =\emptyset\}|$ be the random variable that counts number of vertices in layer $\mathcal{A}_\ell$ that are not coloured after round two. Then, there exists a constant $c(\alpha) >0$ such that $\Exp{Z_\ell} \in \bigO\left(n^{-c(\alpha)}\right)$. 
\end{lemma}

\begin{proof}

Let $Z_\ell = \sum_{v\in V_\ell}\indicator{\colourchosen{2}{v} = \emptyset}$ as defined in our statement and let $U_\ell \subseteq V_\ell$ be the subset of uncoloured vertices of $V_\ell$ before the second round of \RCTID starts. By \Cref{lem:uncloured-layer-ID} we have  $|U_\ell| \in \Tilde{\bigO}(\sqrt{ne^{-\alpha\ell}}(1+{n^{(1/\alpha - 1)/2}e^{{\ell}(1-\alpha)/2}})+1)$ \wehp Moreover, let $\Delta_\ell := \max_{v \in U_\ell}\deguncoul{2}{v}$ be the maximum uncoloured degree of the set of vertices $U_\ell$ before the second round of \RCTID starts which by \Cref{lem:uncoloured-degree-ID} is $\Delta_\ell \in \Tilde{\bigO}(\sqrt{n}\cdot e^{-(R-\ell)/4} + n^{\frac{1}{2\alpha}} \cdot e^{- (R-\ell)/2} + 1)$ \wehp using that the radius of a vertex in layer $\mathcal{A}_\ell$ is at least $r\geq R-\ell-1$. Then let $X_\ell$ be the minimal colour space of the set of vertices $U_\ell$ before the second round starts which is $X_\ell \in \Omega(\Delta)$ \wehp by \Cref{lem:slack-rct}.

Since all events happen \wehp, using law of total expectation and then linearity of expectation yields 

\begin{align}\label{eq:vanishing-expectation}
\begin{split}
\Exp{Z_\ell} &\in \Tilde{\bigO}\underbrace{(\underbrace{\sqrt{ne^{-\alpha\ell}}}_{a} + \underbrace{n^{1 -\frac{1}{2\alpha}}e^{-\ell(\alpha -1/2)})}_{b} +1}_{|U_\ell|})\cdot \Tilde{\bigO}(\underbrace{\underbrace{n^{-(1/\alpha - 1)/2}\cdot e^{-(R-\ell)/4}}_{c} + \underbrace{e^{- (R-\ell)/2}}_{d}}_{\Delta_\ell/X_\ell}) + n^{-\omega(1)},    
\end{split}
\end{align}

where the $n^{-\omega(1)}$  term corresponds to the low probability event that our bounds for $|U_\ell|, \Delta_\ell, X_\ell$ do not hold and we used that $\Delta \in n^{\frac{1}{2\alpha}-o(1)}$ for a typical HRG~(\Cref{the:max-degree}).

We proceed by showing that all resulting sums in~\Cref{eq:vanishing-expectation} is upper bounded by a polynomial vanishing term as stated in our lemma:

\smallskip

\textbf{Case 1}~[$a\cdot c$]: Recall that $R = 2\log(n) +C$ and we get
$$\Exp{Z_\ell} \in \Tilde{\bigO}(\sqrt{ne^{-\alpha\ell}}\cdot n^{-(1/\alpha - 1)/2}\cdot e^{-(R-\ell)/4} ) \in \Tilde{\bigO}(n^{-(1/\alpha - 1)/2} e^{-\ell(\alpha - 1/2)/2}) \in  \Tilde{\bigO}(n^{-(1/\alpha - 1)/2}),$$
since $\alpha > 1/2$. It follows that $\Exp{Z_\ell} \in \bigO(n^{-c(\alpha)})$ since $\alpha < 1$.

\smallskip

\textbf{Case 2}~[$a\cdot d$]: Recall that $\finallayer \leq \lceil\alpha^{-1}\cdot(\log(n) + \log(\log(n)) \rceil$ is the maximum layer index for a typical HRG~(\Cref{lem:layer-properties}). Then by $R = 2\log(n) +C$ we obtain
$$\Exp{Z_\ell} \in \Tilde{\bigO}(\sqrt{ne^{-\alpha\ell}}\cdot e^{-(R-\ell)/2}) \in \Tilde{\bigO}(n^{-1/2} e^{\ell(1-\alpha)/2}) \in \Tilde{\bigO}(n^{-(1 - \frac{1}{2\alpha})}),$$
since $\alpha < 1$, and thus the above term is upper bounded when $\ell = \finallayer \leq \lceil\alpha^{-1}\cdot(\log(n) + \log(\log(n)) \rceil$. We conclude for the case that $\Exp{Z_\ell} \in \bigO(n^{-c(\alpha)})$  as $\alpha > 1/2$.

\smallskip
\textbf{Case 3}~[$b\cdot c$]: We obtain by using $R = 2\log(n) +C$ that
$$\Exp{Z_\ell} \in \Tilde{\bigO}(n^{1 -\frac{1}{2\alpha}}e^{-\ell(\alpha -1/2)}\cdot n^{-(1/\alpha - 1)/2}\cdot e^{-(R-\ell)/4} ) \in \Tilde{\bigO}(n^{-(1/\alpha - 1)} e^{-\ell(3/4-\alpha)}).$$

Then for $\alpha \leq 3/4$ this is upper bounded by using~\Cref{lem:layer-properties}, by which $\ell \leq \finallayer \leq \lceil\alpha^{-1}\cdot(\log(n) + \log(\log(n)) \rceil$ and we get $\Exp{Z_\ell} \in \Tilde{\bigO}(n^{-(1/\alpha - 1)} e^{-\finallayer(3/4-\alpha)}) \in \bigO(n^{-\frac{1}{4\alpha}})  \in \bigO(n^{-c(\alpha)})$ since $\alpha > 1/2$.

On the flips ide, for $\alpha > 3/4$ the term simplifies to $\Exp{Z_\ell} \in \Tilde{\bigO}(n^{-(1/\alpha - 1)})$ and we conclude the case with $\Exp{Z_\ell} \in \bigO(n^{-c(\alpha)})$ as $\alpha < 1$.

\smallskip
\textbf{Case 4}~[$b\cdot d$]: Since $R = 2\log(n) +C$ and using that $\ell$ is at most $\finallayer \leq \lceil\alpha^{-1}\cdot(\log(n) + \log(\log(n)) \rceil$ for typical HRG by~\Cref{lem:layer-properties}, we obtain in this case
$$\Exp{Z_\ell} \in \Tilde{\bigO}(n^{1 -\frac{1}{2\alpha}}e^{-\ell(\alpha -1/2)}\cdot e^{-(R-\ell)/2}) \in \Tilde{\bigO}(n^{-\frac{1}{2\alpha}} e^{\ell(1-\alpha)}) \in \Tilde{\bigO}(n^{-(1 - \frac{1}{2\alpha})}) \in \bigO(n^{-c(\alpha)}),$$
as $\alpha >1/2$. This conculdes the case and also the proof since we have shown in all cases that there exists a constant $c(\alpha) > 0$ such that $\Exp{Z_\ell} \in \bigO(n^{-c(\alpha)})$.
\end{proof}

We now show that \RCTID is likely to successfully colour a typical HRG fast. 

\begin{proposition}[\RCTID upper bound]\label{pro:rctid-upper-bound}
 Let $G\sim\hrg$ by a typical hyperbolic random graph and let $Z = |\{v \in V : \colourchosen{2}{v} = \emptyset\}|$ be the number of uncoloured vertices after round $t=2$ of \RCTID  with colour space $\Omega(\Delta)$ has ended. Then, there exists a constant $c > 0$ such that $Z=0$ with probability $1 - \bigO(n^{-c})$.

\end{proposition}

\begin{proof}
    To show that $Z = 0$ with probability $1 - \bigO(n^{-c})$, fix a layer $\mathcal{A}_\ell$ and let $Z_\ell$ be the number of uncoloured vertices after the second round in layer $\mathcal{A}_\ell$. Using \Cref{lem:RCTI-second-round}, we have $\Pro{Z_\ell \geq 1} \in \bigO(n^{-c(\alpha)})$ by using Markov's inequality~\Cref{lem:markov} with $a=1$. Recall that $\mathcal{A}_{\finallayer}$ is the layer with largest index $\finallayer$ that still contains at least one vertex and $\finallayer \in \bigO(\log(n))$. Using a union bound we obtain that

    \begin{align*}\label{eq:vanishing-probability-second-round}
    \Pro{Z = 0} = 1 - \Pro{Z \geq 1} \geq 1 - \sum_{\ell =0}^{\finallayer}\Pro{Z_\ell \geq 1} \in 1 - \bigO\left(\sum_{\ell =0}^{\finallayer}n^{-c(\alpha)}\right) \in 1 - \bigO(n^{-c}),
    \end{align*}

    for a constant $c>0$.

\end{proof}

\subsection{Theorem~\ref{the:rctid}, part 2: Limitations of Random Colour Trial with ID Priority}\label{sec:rctid-part2}

In this section we show that our upper bound for the colour space 
for \RCTID in \Cref{pro:rctid-upper-bound} cannot be much improved uppon, in the sense that using less colours by a polylogarithmic factor, it turns out to be highly unlikely that the algorithm terminates.

\begin{proposition}[\RCTID lower bound]\label{pro:no-termination-rctid}
    Let $G\sim\hrg$ be a typical hyperbolic random graph. Then \RCTID with colour space $|\palette| \in \bigO(\Delta/\log^{3}(n))$ colours never terminates on $G$ \wehp
\end{proposition}

\begin{proof}
    Recall that $V_\ell$ is the set of vertices located in layer $\mathcal{A}_\ell$ and with hindsight we set $\ell := (\log(n) - (1/2 + \alpha)\log(\log(n)))\alpha^{-1}$. We bound the size of $|V_\ell|$ via \Cref{eq:layer-measure} and by our choice of $\ell$ we obtain $\E{|V_\ell|} = \Theta(1)n \cdot e^{-\alpha\ell} \in \omega(\log(n))$ since $\alpha > 1/2$. Hence, by a Chernoff bound it follows that $|V_\ell| \in \omega(\log(n))$ \wehp and thus, also for our typical hyperbolic random graph $G$.
    
    Moreover, \Cref{lem:leaf} is applicable for our choice of $\ell$ on a typical hyperbolic random graph. This tells us that there exists a set $U_\ell \subseteq |V_\ell|$ of size $|U_\ell| \in \omega(\log(n))$ where a vertex $u \in U_\ell$ has have a constant fraction of its neighbours as leaves. We denote the set of neighbour leaves of $u$ by $L(u)$, and it holds for $u \in U_\ell$ that $|L(u)| = \Theta(1)\deg(u)$. Now fix a vertex $u \in U$ and assign an ID to $u$ such that the ID of $u$ has a lower priority than any of its leaf neighbour $v \in L(u)$.

    We now show that any colour $\colour \in \palette$ is tried by a leaf $v \in L(u)$. Since any leaf $v \in L(u)$ has a higher priority in ID than $u$ and $u$ is the only neighbour for $v \in L(u)$, this implies that $v \in L(u)$ assigns its candidate colour. Thus, if in the first round the entire colour space is consumed by leaves of $u$,  $u$ stays uncoloured and cannot be coloured in later rounds.

    To show this, we first get a lower bound on $|L(u)|$. Recall that $u \in U_\ell$ and that a constant fraction of $u$'s neighbours are leaves. We consider the  degree of $u \in V_\ell$ which is $\deg(u) \in \Omega\left(n^{\frac{1}{2\alpha}}\cdot \log^{-\frac{1/2 + \alpha}{2\alpha}}(n)\right)$ \wehp by our choice of $\ell$ and \Cref{lem:layer-properties}. Subsequently, it follows that $|L(u)| \in \Omega\left(n^{\frac{1}{2\alpha}}\cdot \log^{-\frac{1/2 + \alpha}{2\alpha}}(n)\right)$ as we consider a typical HRG.

    Next, notice that the colour space of the \RCTID is $|\palette| \in \bigO(\Delta\cdot\log^{-3}(n)) \in \bigO(n^{\frac{1}{2\alpha}}\cdot\log^{-2}(n))$ for a typical HRG by~\Cref{the:max-degree}
    
    Now, fix a colour $\colour \in \palette$ and let us write $Z_\colour$ that is the random variable that represents the number of leaves of $L(u)$, that try $\colour$ as their candidate colour in the first round. Since ${|L(u)|} \in \Omega\left(n^{\frac{1}{2\alpha}}\cdot \log^{-\frac{1/2 + \alpha}{2\alpha}}(n)\right)$ and a vertex picks its candidate colour uniform at random from $|\palette| \in \bigO(n^{\frac{1}{2\alpha}}\cdot\log^{-2}(n))$ , it follows by linearity of expectation that
    \begin{align*}
         \Exp{Z_\colour} &= \sum_{v \in L'(u)} |\palette|^{-1} \geq \Theta(1)\log^{\frac{3\alpha-1/2}{2\alpha}}(n) \in \omega(\log(n)), 
    \end{align*}
where we used $\alpha > 1/2$ in the second step. Since $Z_\colour$ is a binomial random variable, it follows by an application of a Chernoff bound that $Z_\colour \in \omega(\log(n))$ \wehp Using a union bound over all colours we conclude that every colour has been picked as a candidate by a leaf of $u$ at least once \wehp Since leaves are of degree $1$, these leaves do not discard their candidate colour as $u$ has a lower priority by our choice of IDs. Thus, $u$ has no colour to pick from in the second round and it is not possible to colour $G$ without a conflict using the colour palette $\palette$. 
\end{proof}

Using \Cref{pro:rctid-upper-bound} as an upper bound on the necessary colours for \RCTID and \Cref{pro:no-termination-rctid} for the lower bound for the colourspace, we obtain \Cref{the:rctid}.


\section{Analysis of Random Colour Trial with Degree Priority (Proof of Theorem~\ref{the:rctdegSimplified})}\label{sec:rctdeg}

In this section we analyse another variant of \RCT we refer to as \emph{random colour trial with degree priority} (\RCTDEG). Similar to \RCT, at the beginning of each round all uncoloured vertices try a candidate colour uniform at random from their respective colour space. A vertex assigns its candidate colour if and only if no vertex with higher degree has the same candidate colour; breaking ties of vertices with same degree by giving priority to the vertex with smaller ID. In a certain sense, \RCTDEG is a special case of \RCTID where the smallest ID's are given to the largest degree vertices. We first state the formal results. 

\begin{restatable}[Random colour trial with degree priority]{theorem}{rctdegTheorem}\label{the:rctdeg}\text{ } Let  $\zeta_1 = \frac{1}{2\alpha + 1/2}$ and let $\zeta_2 = \frac{1}{4\alpha -1}$. 
\begin{compactenum}\item On a typical hyperbolic random graph \RCTDEG terminates in 2 rounds \aas 
\begin{compactenum}
    \item if $\alpha \leq 3/4$ and the colour space is at least $\log^4(n)\cdot \chi^2$ with chromatic number $\chi$;
    \item if $\alpha > 3/4$ and the colour space is at least $\log^4(n)\cdot\sqrt{n}$. 
\end{compactenum}

\item On a typical hyperbolic random graph \RCTDEG terminates in $\bigO(1)$ rounds \aas 
\begin{compactenum}
    \item if $\alpha \leq 3/4$ and the colour space is at least $\log(n)\cdot n^{\zeta_1}$;
    \item if $\alpha > 3/4$ and the colour space is at least $\log(n)\cdot n^{\zeta_2}$. 
\end{compactenum}

\item  There exists a constant $\delta >0$ such that \RCTDEG on a typical HRG never terminates \wehp\begin{compactenum}
    \item if $\alpha \leq 3/4$ and the colour space is at most $\frac{\delta}{\log(n)} \cdot {n^{\zeta_1}}$;
    \item if $\alpha > 3/4$ and the colour space is at most $\frac{\delta}{\log(n)} \cdot {\chi^2}$ with chromatic number $\chi$.
\end{compactenum}
\end{compactenum}
\end{restatable} 

Before embarking on the proof of this theorem, we briefly show how this implies \Cref{the:rctdegSimplified}.

\begin{proof}[Proof of \Cref{the:rctdegSimplified}]

First recall that a typical hyperbolic random graph has maximum degree $\Delta \in n^{\frac{1}{2\alpha}\pm o(1)}$ (\Cref{the:max-degree}) and chromatic number $\chi \in \Theta(n^{1-\alpha})$~(\cite[Corollary 10]{bmrs-stacs-25}). We use this to show that \Cref{the:rctdeg} implies \Cref{the:rctdegSimplified}:

We start that our desired statements hold for $\alpha \leq 3/4$. To this end, let $0 < \delta_1 < 1 - 4\alpha(1-\alpha)$ and $\delta_2 = (4\alpha + 1)^{-1}>0$ and note that $\delta_2 > \delta_1$. By~\Cref{the:rctdeg} we have that \RCTDEG terminates after 2 rounds \aas if the colour space is at least $\log^4(n)\cdot\chi^2 \in \Omega(\log^4(n) \cdot n^{2(1-\alpha)})$. This is smaller than $\Delta^{1-\delta_1} \in n^{\frac{1-\delta_1}{2\alpha} -o(1)}$ by our choice of $\delta_1$ and thus, colour space $\Delta^{1-\delta_1}$ suffices for $\RCTDEG$ so it colour a typical HRG after two rounds \aas Similar, it holds \aas by~\Cref{the:rctdeg} that \RCTDEG colours a typical HRG within constant rounds if the colour space is at least $\log(n) \cdot n^{1/(2\alpha +1/2)}$ which is again smaller than $\Delta^{{1-\delta_2} +o(1)} \in n^{\frac{1-\delta_2}{2\alpha} + o(1)}$ by our choice of $\delta_2$ and using $o(1)$ large enough. Similar $\Delta^{{1-\delta_2} -o(1)} \in n^{\frac{1-\delta_2}{2\alpha} - o(1)}$ is smaller than $ \frac{\delta}{\log(n)} \cdot n^{1/(2\alpha +1/2)}$ and thus, by \Cref{the:rctdeg} $\Delta^{{1-\delta_2} -o(1)}$ fails \wehp to colour a typical HRG.

Next, we show that our desired statements also hold for $\alpha > 3/4$ given \Cref{the:rctdeg}. For this case we set $0 < \delta_1 < 1-\alpha$ or rather $0<\delta_2< 1 - \frac{2\alpha}{4\alpha - 1}$ and note that for such choices $\delta_2 > \delta_1$ is possible. Then, by our choice of $\delta_1$, it follows that there exists another constant $\delta'>0$ such that $\Delta^{1-\delta_1} \geq n^{1/2 + \delta'} > \log(n)\cdot\sqrt{n}$ which by \Cref{the:rctdeg} implies that $\Delta^{1-\delta_1}$ colours suffice for $\alpha > 3/4$ to colour a typical HRG in two rounds. Lastly, one can verify that for our choice of $\delta_2$, that also $\Delta^{1-\delta_2} > n^{\zeta_2 + o(1)}$, hence, by~\Cref{the:rctdeg} a typical HRG is coloured within constant rounds \aas by \RCTDEG, if the colour space is at least $\Delta^{1-\delta_2}$.
\end{proof}

In \Cref{sec:rctdeg-part1}, we present the proof that \RCTDEG only requires a colour space of size $\Delta^{1-\delta}$ in order to colour a hyperbolic random graph after two rounds \aas~(\Cref{pro:rctdeg-upperbound}). Afterwards we show in \Cref{sec:rctdeg-part2} that we can reduce the colour space even further if we aim to colour within $\bigO(1)$ rounds~(\Cref{pro:rctdeg-constant-rounds}). Then, in \Cref{sec:rctdeg-part3} we show lower bounds for the colour space that for certain regimes of the power law exponent are tight up to factor $\log^2(n)$~(\Cref{pro:lower-bound-rctdeg}). In \Cref{sec:all-together} we conclude the section by putting everything together to show \Cref{the:rctdeg}.

\subsection{Theorem~\ref{the:rctdeg}, part 1: Colouring in two Rounds}\label{sec:rctdeg-part1}
We call the set of vertices $V \cap \B_0(R/2)$ the \emph{core}. We first show that \RCTDEG with enough colours gives a complete colouring of the core within one round in \Cref{lem:round-1-rctdeg-core}. Then we show that we can colour the rest within two rounds in~\Cref{lem:second-round-rctdeg}.

We show for $f(n) \in \Omega(1)$ that \RCTDEG with colour space $f(n)\cdot\chi^2$ assigns a \emph{rainbow colouring}, i.e., a different colour for each vertex, to the core within the first round with an error probability that is vanishing in $f(n)$.

\begin{lemma}[\RCTDEG core rainbow colouring]\label{lem:round-1-rctdeg-core}
     Let $G\sim \hrg$ be a typical hyperbolic random graph and let $f(n) \in \Omega(1)$. Then, after the first round of \RCTDEG with colour space of size $|\palette| \in f(n)\cdot \chi^2$, the core is completely coloured with probability $1 -\bigO(1/f(n))$.
\end{lemma}

\begin{proof}
    
Let $c' = 7$ and let $U = V \cap \B_0(R/2 + c')$, i.e., the set of vertices that forms the core and some additional vertices with slightly larger radii. Since for a typical HRG all vertices $V\setminus U$ have a smaller degree than the vertices of the core in $\B_0(R/2)$ by \Cref{lem:larger-neighbour-radius} all vertices of the core assign their candidate colour in the first round, if all candidate colours for the vertices of $U$ are unique. To show this, let $A$ be the event that there exists no pair $u,v \in U$ that picked the same candidate colour in round 1 of the \RCTDEG. We obtain using that the colour space is $|\palette| \in f(n)\cdot \chi^2$ by the hypothesis of our lemma and $\chi \in \Theta(n^{1-\alpha})$ for a typical HRG by \cite[Corollary 10]{bmrs-stacs-25}, that event $A$ holds with probability 
    \begin{align*}
        \Pro{A} &= \prod_{j=1}^{|U|-1}(1- j\cdot |\palette|^{-1}) \geq \exp{\left(-\Theta(1)(f(n) \cdot \chi^2)^{-1}\cdot n^{2(1-\alpha)}\right)} \in 1- \bigO(1/f(n)),
    \end{align*}
where we used that $|U| \in \Theta(n^{(1-\alpha)})$ for a typical HRG by \Cref{lem:measure-inner-disk} in conjunction with a Chernoff bound. We conclude that all vertices of the core assigned their candidate colour after the first round with probability $1- \bigO(1/f(n))$ since under event $A$, all vertices that have a degree at least as large as the smallest degree vertex of the core have a unique candidate colour.
\end{proof}

In the next lemma, we bound for \RCTDEG the number of uncoloured vertices for each layer after the first round.

\begin{lemma}[Round~1 \RCTDEG]\label{lem:uncloured-layer-RCTDEG}
    Let $G\sim\hrg$ be a typical hyperbolic random graph and let $Y_\ell = |\{v \in V_\ell : \colourchosen{1}{v} =\emptyset\}|$ be the random variable which counts the number of uncoloured vertices in in layer $\mathcal{A}_\ell$ for $\ell \in [\lceil \alpha^{-1}(\log(n) - 2\log(\log(n))) \rceil]$ after the first round of \RCTDEG with colour palette $\palette$. Then
    $$
        Y_\ell\in\begin{cases}
{\bigO}(\log^3(n)\cdot|\palette|^{-1}ne^{-\ell(\alpha - 1/2)}) \text{ for $0 \leq \ell \leq \lfloor 2/(1-\alpha)\log(\log(n))\rfloor$},\\
{\bigO}(\log^2(n)\cdot|\palette|^{-1} ne^{-\ell(2\alpha - 1)}) \text{ for $\lceil 2/(1-\alpha)\log(\log(n))\rceil \leq \ell \leq \lceil R/2 \rceil$},\\
\bigO(\log^2(n)\cdot|\palette|^{-1} n^2e^{-2\alpha\ell}) \text{ for $ \lceil R/2 \rceil \leq \ell \leq \lceil \alpha^{-1}(\log(n) - 2\log(\log(n))) \rceil$},
\end{cases}  
    $$
    with probability $1 - \bigO(\log^{-2}(n))$. 
\end{lemma}

\begin{proof}
 We use \Cref{lem:layer-properties} to derive our desired bound: as in our statement, let $Y_\ell := \sum_{v \in V_\ell}\indicator{\colourchosen{1}{v} = \emptyset}$ to be the random variable with which we count the number of vertices not coloured after the first round in layer $\mathcal{A}_\ell$. Aiming for an upper bound on $Y_\ell$, for $u \in V_\ell$, let $Y_u$ be the indicator that is $1$ if $u$ is uncoloured after the first round. Fix a vertex $u \in V$ and reveal its candidate colour $\colour$. Then reveal the candidate colour of a neighbour $v \in N(u)$, that has larger degree than $u$ and the probability that $v$ tries the same candidate colour $\colour$ is $|\palette|^{-1}$. Thus, via a union bound over all neighbours of $u$ with larger degree, the probability that any of its neighbours with larger degree has candidate colour $\colour$ is $\Pro{Y_u = 1} \leq \deghigh{u}/|\palette|$. Then by linearity of expectation we get 
 \begin{align*}
 \E{Y_\ell} &= \E{\sum_{v \in V_\ell}Y_v} \leq |V_\ell| \cdot \deghigh{u}/|\palette| \in {\bigO}((ne^{-\alpha\ell} + \log(n))\cdot \deghigh{u}/|\palette|),    
 \end{align*}

 where we used in the last step that $G$ is a typical HRG and thus, the number of vertices in a layer is given by $|V_\ell| \in \Theta\left(ne^{-\alpha\ell}\right)$ via~\Cref{lem:layer-properties}. Next, we use \Cref{lem:neighbours-larger-degree} by which we get an upper bound on $\deghigh{u}$ for a typical HRG. This yields

\begin{align*}
    \E{Y_\ell}\in\begin{cases}
{\bigO}(\log(n)\cdot|\palette|^{-1}ne^{-\ell(\alpha - 1/2)}) \text{ for $0 \leq \ell \leq \lfloor 2/(1-\alpha)\log(\log(n))\rfloor$},\\
{\bigO}(|\palette|^{-1} ne^{-\ell(2\alpha - 1)}) \text{ for $\lceil 2/(1-\alpha)\log(\log(n))\rceil \leq \ell \leq \lceil R/2 \rceil$ }\\
\bigO(|\palette|^{-1} n^2e^{-2\alpha\ell}) \text{ for $ \lceil R/2 \rceil \leq \ell \leq \lceil \alpha^{-1}(\log(n) - 2\log(\log(n))) \rceil$}.
\end{cases}  
\end{align*}
The desired concentration follows by using Markov's inequality~\Cref{lem:markov}, setting $a=  \E{Y_\ell}\cdot\log^{2}(n)$ and we obtain $\Pro{Y_\ell \geq a} \leq \log^{-2}(n)$.
\end{proof}

Recall that the core is the clique that is formed by the set of vertices $V \cap \B_0(R/2)$. We now show that \RCTDEG with a relatively small colour space colours every vertex outside of the core after 2 rounds \aas To express our bound on for the colour space of \RCTDEG we set $\maxexp := \max(2(1-\alpha), 1/2)$.

\begin{lemma}[Second round \RCTDEG]\label{lem:second-round-rctdeg}
 Let $G\sim\hrg$ by a typical hyperbolic random graph and let $\maxexp := \max(2(1-\alpha), 1/2)$. Then after round 2 of \RCTDEG  with colour space $|\palette| \geq  \log^4(n)\cdot n^{\maxexp}$ has ended, all vertices outside the core are coloured with probability $1 - \bigO(1/\log(n))$.
\end{lemma}

\begin{proof}
Since we only consider vertices outside the core, we only consider vertices with radius larger than $R/2$. We consider the layers outside the core and get the number of uncoloured vertices per layer. We then wrap up the proof by summing over all expected uncoloured vertices per layer outside the core and show that this is vanishing.

To get started, fix a layer $\mathcal{A}_\ell$ with index $\ell \leq \lceil R/2 \rceil$ and let $U_\ell \subseteq V_\ell$ be the set of uncoloured vertices in layer $\mathcal{A}_\ell$ after round one has ended. Now let $A_1$ be the event that $|U_\ell|\in {\bigO}(\log^3(n)\cdot|\palette|^{-1}ne^{-\ell(\alpha - 1/2)})$ if $\ell \leq \lfloor 2/(1-\alpha)\log(\log(n))\rfloor$ and otherwise $U_\ell \in {\bigO}(\log^2(n)\cdot|\palette|^{-1} ne^{-\ell(2\alpha - 1)})$, which by \Cref{lem:uncloured-layer-RCTDEG} occurs with probability $\Pro{A_1} = 1 - \bigO(\log^{-2}(n))$. Moreover, let $A_2$ be the event that the minimal colour space of a vertex is $X := \min_{u\in U_\ell}\colourSpace{t=2}{u} \in \Omega(|\palette|)$. Since $\ell \leq \lceil R/2 \rceil$ the maximum degree of a vertex $u \in U_\ell$ is for a typical HRG by \Cref{lem:layer-properties}~$\deg(u) \in \bigO(\sqrt{n}) \in o(|\palette|)$ and thus, $A_2$ holds surely given that $G$ is a typical HRG. Let us write $A: = \{A_1 \cap A_2\}$ and we conclude $\Pro{A} =1 - \bigO(\log^{-2}(n))$. 

Next, let $Z_\ell:=\sum_{u \in U_\ell}\indicator{\colourchosen{2}{u} = \emptyset}$ be the number of uncoloured vertices in layer $\mathcal{A}_\ell$ after round 2 has ended. We upper bound $\Exp{\indicator{\colourchosen{2}{u} = \emptyset} |A}$. For this, consider a vertex $u \in U_\ell$ that has $\deghigh{u} = |\{v \in N(u): \deg(v) \geq \deg(u)\}|$ many neighbours with degree at least $\deg(u)$ while all uncoloured neighbours of $u$ have a colour space of size at least $X$. Then, reveal $u$'s candidate colour $\colour$ and the probability that any uncoloured neighbour of $u$ with degree at least $\deg(u)$ has the same candidate colour $\colour$ is $\Exp{\indicator{\colourchosen{2}{u} = \emptyset}} \leq (1 - 1/X)^{\deghigh{u}} \leq \frac{\deghigh{u}}{X}$. Recall that event $A$ implies that $X \in \Omega(|\palette|)$ and thus, $\Exp{\indicator{\colourchosen{2}{u} = \emptyset}} \leq \Theta(1) \deghigh{u}/|\palette|$. Then we get by conditioning on event $A$ that 

\begin{equation}\label{eq:expectation-indicator-conditioned-rctdeg}
    \Exp{\indicator{\colourchosen{2}{u} = \emptyset} |A} 
    \in \begin{cases}
       {\bigO}(|\palette|^{-1}\cdot e^{\ell/2}\log(n)) \text{ for $0 \leq \ell \leq \lfloor 2/(1-\alpha)\log(\log(n))\rfloor$},\\
\bigO(|\palette|^{-1}\cdot e^{\ell(1-\alpha)}) \text{ for $\lceil 2/(1-\alpha)\log(\log(n))\rceil \leq \ell \leq \lceil R/2 \rceil$,}
    \end{cases}     
\end{equation} 

by using \Cref{lem:neighbours-larger-degree} to upper bound $\deghigh{u}$.

We now proceed by a case distinction for $\ell$ to upper bound $Z_\ell$ for any layer. We show that for any $\ell$, the event $B_\ell := \{Z_\ell = 0\}$, that is, after the second round the number of uncoloured vertices in layer $\mathcal{A}_\ell$ is $0$, holds with probability $\Pro{B_\ell} \in 1 - \bigO(\log^{-2}(n))$. Our desired statement then follows afterwards via a union bound over all $\mathcal{O}(\log(n))$ layers.

\smallskip

\textbf{Case 1}~[$0 \leq \ell \leq \lfloor 2/(1-\alpha)\log(\log(n))\rfloor$]:

We now bound $\Exp{Z_\ell |A}$, the expected number of vertices after round 2, conditioning on event $A$. Recall that if event $A$ occurs, then the number of uncoloured vertices in layer $\mathcal{A}_\ell$ after round 1 is $|U_\ell|\in {\bigO}(\log^3(n)\cdot|\palette|^{-1}ne^{-\ell(\alpha - 1/2)})$. Thus, by conditioning on $A$ we obtain via linearity of expectation
\begin{align}\label{eq:expectation-all-conditioned-case1}
    \Exp{Z_\ell| A} &\in {\bigO}(\log^3(n)\cdot|\palette|^{-1}ne^{-\ell(\alpha - 1/2)}) \Exp{\indicator{\colourchosen{2}{u} = \emptyset}|A} \in {\bigO}(\log^4(n)\cdot ne^{\ell(1-\alpha)} \cdot |\palette|^{-2}),
\end{align}
where we used \Cref{eq:expectation-indicator-conditioned-rctdeg} in the second step. Next, we use that $\alpha < 1$ and thus,

\begin{align*}
      \Exp{Z_\ell| A} &\in  {\bigO}(\log^6(n)\cdot n \cdot |\palette|^{-2}) \in \bigO(\log^{-2}(n)),
\end{align*}

where we used that $\ell \leq \lfloor 2/(1-\alpha)\log(\log(n))\rfloor$ by our case and $|\palette| \geq \log^4(n)\cdot \sqrt{n}$ by the assumption of our lemma. Recall that $B_\ell$ is the event that $Z_\ell =0$. Using Markov's inequality we then get that $\Pro{B_\ell|A} \in 1 - \bigO(\log^{-2}(n))$. Recall that $A$ also occurred with the same probability bound and we get $\Pro{B_\ell \cap A} \geq 1 - (\Pro{A^C} + \Pro{B_\ell^C}) \in 1 - \bigO(\log^{-2}(n))$.

\smallskip

\textbf{Case }~2[$\lceil 2/(1-\alpha)\log(\log(n))\rceil \leq \ell \leq \lceil R/2 \rceil$]:

Aiming for an upper bound on $\Exp{Z_\ell}$, recall that the number of uncoloured vertices in layer $\mathcal{A}_\ell$ under event $A$ is bounded by $U_\ell \in {\bigO}(\log^2(n)\cdot|\palette|^{-1} ne^{-\ell(2\alpha - 1)})$. Using \Cref{eq:expectation-indicator-conditioned-rctdeg} and applying linearity of expectation we obtain

\begin{align}\label{eq:expectation-all-conditioned-case2}
\begin{split}
    \Exp{Z_\ell| A} &\in {\bigO}(\log^2(n)\cdot|\palette|^{-1} ne^{-\ell(2\alpha - 1)}) \Exp{\indicator{\colourchosen{2}{u} = \emptyset}|A} \in {\bigO}(\log^2(n)\cdot ne^{\ell(2-3\alpha)} \cdot |\palette|^{-2}).
\end{split} 
\end{align}

We consider different cases for $\alpha$ to obtain a general upper bound for \Cref{eq:expectation-all-conditioned-case2}. For $\alpha < 2/3$ this is monotonically increasing in $\ell$. Since $\ell \leq \lceil R/2 \rceil < \lceil \log(n) +C/2\rceil$ by our case and we have $|\palette| \geq \log^4(n)\cdot n^{2(1-\alpha)}$ by the hpyothesis of our lemma we obtain for $\Cref{eq:expectation-all-conditioned-case2}$
$$
 \Exp{Z_\ell| A} \in \bigO(\log^{-6}(n) \cdot n^{-(1-\alpha)}) \in o(\log^{-2}(n)).
$$

 For $2/3 \leq \alpha <1$, \Cref{eq:expectation-all-conditioned-case2} is monotonically decreasing in $\ell$ and we obtain for our colour space $|\palette| \geq \log^4(n)\cdot\sqrt{n}$

 $$
 \Exp{Z_\ell| A} \in \bigO(\log^{-6}(n)) \in o(\log^{-2}(n)),
 $$
using $\ell \geq 0$. Recalling that $\Pro{A} = 1 - \bigO(\log^{-2}(n))$, an application of Markov's inequality yields for the event $B_\ell = \{Z_\ell =0\}$ that $\Pro{B_\ell\cap A} \in 1 - \bigO(\log^{-2}(n))$ which concludes our second case.

In conclusion we have $\Pro{B_\ell\cap A} \in 1 - \bigO(\log^{-2}(n))$ for all layers and thus, a union bound over $\bigO(\log(n))$ layers shows that with probability $1 - \bigO(1/\log(n))$ no vertex with radius larger than $R/2$ is uncoloured after the second round of \RCTDEG what is what we sought out to show.

\end{proof}

We now get the first part of \Cref{the:rctdeg}: for $\maxexp = (2(1-\alpha),1/2)$, \RCTDEG with colour space $|\palette| \geq \log^4(n)n^{\maxexp}$ colours a hyperbolic random graph within two rounds \aas

\begin{proposition}[\RCTDEG 2 rounds \aas]\label{pro:rctdeg-upperbound}
     Let $G\sim\hrg$ by a typical hyperbolic random graph, let $\maxexp = (2(1-\alpha),1/2)$ and let $Z$ be the number of uncoloured vertices after round 2 of \RCTDEG  with colour space at  $|\palette| \geq \log^4(n)\cdot n^{\maxexp}$ has ended. Then, $Z=0$ \aas   
\end{proposition}
\begin{proof}
    Round 1: Using $\Cref{lem:round-1-rctdeg-core}$ with $f(n) = \log^4(n)$ and $\chi \in\Theta( n^{1-\alpha})$ for a typical HRG~(\cite[Corollary 10]{bmrs-stacs-25}) the core is coloured with probability $1 - \bigO(\log^{-4}(n))$ after the first round since $\maxexp > 2(1-\alpha)$.

    Round 2: By the colour space given for \RCTDEG we get by \Cref{lem:second-round-rctdeg} that after the second round all vertices outside the core are coloured with probability $1-\bigO(1/\log(n))$.

    We conclude that the entire graph is coloured after the second round with probability $1-\bigO(1/\log(n))$.
\end{proof}

\subsection{Theorem~\ref{the:rctdeg}, part 2: Constant rounds with a Smaller Colour Space}\label{sec:rctdeg-part2}

We show that if we aim for constants rounds, the upper bound on the colour space for \RCTDEG in \Cref{pro:rctdeg-upperbound} can be improved.

In the following lemma we show that given a specific radius $r$ and colour space $|\palette|$, \RCTDEG assigns a rainbow colouring to the set $V\cap \B_0(r)$, i.e., after the first round each vertex $V\cap \B_0(r)$ is coloured with a different colour.

\begin{lemma}[\RCTDEG generalised rainbow colouring]\label{lem:round-1-rctdeg-rainbow}
     Let $G\sim \hrg$ be a typical hyperbolic random graph, let $r=R(1 - (2\alpha +1/2)^{-1})$ and let $f(n) \in \Omega(1)$. Then, after the first round of \RCTDEG with colour space of size $|\palette| \geq f(n)\cdot n^{\frac{1}{2\alpha +1/2}}$, the set of vertices $V \cap \B_0(r)$ is coloured with probability $1 -\bigO(1/f(n))$.
\end{lemma}

\begin{proof}
    
Let $c' = 7$ and let $U = V \cap \B_0(r + c')$, i.e., the set of vertices we aim to show a rainbow colouring for and some additional vertices with slightly larger radii. Since for a typical HRG, all vertices $V\setminus U$ have a smaller degree than the vertices located in $\B_0(r)$ by \Cref{lem:larger-neighbour-radius}, each vertex $u \in V \cap \B_0(r)$ assigns their candidate colour in the first round, if all candidate colours for the vertices of $U$ are unique. To show this, notice that $\E{|U|} \leq (1+o(1))ne^{-\alpha(R-r-c')}$ by \Cref{lem:measure-inner-disk}. Then by our choice $r = R(1 - (2\alpha +1/2)^{-1})$ and $c'=7$, we obtain $\E{|U|}\leq (1+o(1))\sqrt{n^{\frac{1}{2\alpha + 1/2}}}$. A chernoff bound then yields $|U| \in \bigO\left(\sqrt{n^{\frac{1}{2\alpha + 1/2}}}\right)$ \wehp and thus, this holds for a typical HRG.

Now, let $A$ be the event that there exists no pair $u,v \in U$ that picked the same candidate colour in round 1 of the \RCTDEG. We obtain using $|\palette| \geq f(n)\cdot n^{\frac{1}{2\alpha +1/2}}$ for event $A$ that 
    \begin{align*}
        \Pro{A} &= \prod_{j=1}^{|U|-1}(1- j\cdot |\palette|^{-1})\geq \exp{\left(-\Theta(1) |\palette|^{-1} \cdot|U|^2\right)} \in 1- \bigO(1/f(n)),
    \end{align*}
where we also used that $|U| \in \bigO\left(\sqrt{n^{\frac{1}{2\alpha + 1/2}}}\right)$. We conclude that all vertices of $V \cap \B_0(r)$ assigned their candidate colour after the first round with probability $1- \bigO(1/f(n))$ since under event $A$, all vertices $U$ that might have a larger degree than any vertex $v \in V \cap \B_0(r)$ have a unique candidate colour.
\end{proof}

By the previous lemma, \RCTDEG colours all vertices up to a specific radius $r$ within one round \aas To colour the vertices with larger radius we show that, with the same colour space as used in~\Cref{{lem:round-1-rctdeg-rainbow}}, all vertices with radius larger than $r$ are coloured within constant rounds.

\begin{lemma}[\RCTDEG constant rounds for $\alpha\leq 3/4$]\label{lem:second-round-rctdeg-constant}
 Let $G\sim\hrg$ by a typical hyperbolic random graph, let $r=R(1 - (2\alpha +1/2)^{-1})$ and let $U = V \cap (\disk \setminus \B_0(r))$. Then after $t \in \bigO(1)$ rounds of \RCTDEG  with colour space $|\palette| \geq  \log(n)\cdot n^{\frac{1}{2\alpha + 1/2}}$, the set of vertices $U$ is coloured with probability $1 - \bigO(1/\log(n))$.
\end{lemma}

\begin{proof}
Since we only consider vertices with radius larger than $r$, we consider the layers with index $\ell \leq \lceil R -r\rceil$ and count the number of uncoloured vertices per layer. We then wrap up the proof by summing up all uncoloured vertices per layer we considered.

 Observe that the expected degree of a vertex $u \in V \cap (\B_0(R) \setminus \B_0(r))$ is $\E{\deg(u)} \leq \Theta(1)n e^{-r/2}$ by \Cref{lem:vertex-degree}. By our choice of $r$, a Chernoff and union bound it follows for all vertices $V \cap (\B_0(R) \setminus \B_0(r_0))$ that their degree is at most $\bigO(n^{\frac{1}{2\alpha + 1/2}}) \in o(|\palette|)$ \wehp which then also holds for a typical HRG. This implies that at any round $t$, any vertex $u \in V \cap (\B_0(R) \setminus \B_0(r))$ has a colour space of size $\colourSpace{t}{u} \in \Omega(|\palette|) \in \Omega(\log(n) \cdot n^{\frac{1}{2\alpha + 1/2}})$. Let $p_u := \Pro{\colourchosen{t}{u} = \emptyset}$ be the probability that $u$ is uncoloured after round $t$ of the \RCTDEG. We distinguish for two different regimes of layers:

    \smallskip
    
    \textbf{Case 1}[$\ell \in \omega(\log(\log(n)))$]: For $u \in V_\ell$ with $\ell \in \omega(\log(\log(n))$ we get by \Cref{lem:neighbours-larger-degree} 
    
    $$p_u \leq \left(\frac{\deghigh{u}}{\colourSpace{t}{u}}\right)^t \in \bigO\left(\frac{e^{\ell(1-\alpha)}}{\log(n) \cdot n^{\frac{1}{2\alpha + 1/2}}}\right)^t.$$

Next, let $B_\ell$ be the event that there exists a vertex $u \in V_\ell$ is uncoloured after round $t$. A union bound yields

    $$
    \Pro{B_\ell} \leq \sum_{v \in V_\ell} p_v \in \bigO\left(ne^{-\alpha\ell}\left(\frac{e^{\ell(1-\alpha)}}{\log(n) \cdot n^{\frac{1}{2\alpha + 1/2}}}\right)^t\right),
    $$

    where we used that $|V_\ell| \in \Theta(ne^{-\alpha\ell})$ for a typical HRG by \Cref{lem:layer-properties}. Notice that for vertices with radius larger than $r$, this is upper bouded for any $\ell \leq \lceil R- r\rceil$ by

    $$
    \Pro{B_\ell} \in \bigO\left(n\left(\frac{e^{(\lceil R- r\rceil)(1-\alpha)}}{\log(n) \cdot n^{\frac{1}{2\alpha + 1/2}}}\right)^t\right) \in o\left(n^{1-t\frac{2\alpha -1}{2\alpha + 1/2}}\right),
    $$
since $\alpha <1$ and we used that $\ell \leq \lceil R-r\rceil \leq \frac{2\log(n)}{2\alpha + 1/2}$. Since $\alpha > 1/2$ we have $\Pro{B_\ell} \in o(1/n)$ for $t \geq \frac{4\alpha + 1}{2\alpha -1}$. A union bound over all layers gives that for $\ell \in \omega(\log(\log(n)))$ no layer $\mathcal{A}_\ell$ contains an uncoloured vertex after constant rounds \aas

\smallskip

\textbf{Case 2}[$\ell \in \bigO(\log(\log(n)))$]: Observe that $\deghigh{u} \in \Tilde{\bigO}(e^{\ell/2})$ by \Cref{lem:neighbours-larger-degree} and we get for any $u \in V \cap (\B_0(R)\setminus \B_0(R-\ell))$ that the probability that $u$ is uncoloured after round $t$ is at most $p_u \leq  \left({\deghigh{u}}/\colourSpace{t}{u}\right)^t \in \Tilde{\bigO}\left(n^{-{t}/{(2\alpha +1/2)}}\right)$ since $\colourSpace{t}{u} \in \Omega\left(n^{\frac{1}{2\alpha + 1/2}}\right)$. Thus, setting the number of rounds $t \geq 4\alpha + 1$, the probability that there exists an uncoloured vertex $u \in V \cap (\B_0(R)\setminus \B_0(R-\ell))$ after $t$ rounds is $\Pro{B} \in o(1/n)$ via a union bound. This finishes the case and also the proof.
\end{proof}

The combination of \Cref{lem:round-1-rctdeg-rainbow} and \Cref{lem:second-round-rctdeg-constant} provides us with an improved colour space to that of \Cref{pro:rctdeg-upperbound}. The following lemma gives us a further improvement if $\alpha > 3/4$.

\begin{lemma}[\RCTDEG constant rounds for $\alpha >3/4$]\label{lem:rctdeg-large-alpha-upper}
Let $G\sim\hrg$ by a typical hyperbolic random graph. Then \RCTDEG with colour space $|\palette| \geq \log(n)\cdot n^{\frac{1}{4\alpha - 1}}$ colours $G$ in $\bigO(1)$ rounds with probability $1-\bigO(1/\log(n))$.
\end{lemma}

\begin{proof}
    We split the proof into two parts: first, we show that vertices up to a radius $r \leq 2\left(1 - \frac{1}{4\alpha - 1}\right)\log(n) =: r_0$ are coloured after the first round \aas Then, in a second step, we show that all vertices with radius larger than $r_0$ are coloured after constant rounds.
 
    For the set of vertices with radius smaller than $r_0$, fix vertex $v \in V \cap \B_0(r_0)$. Since in \RCTDEG a vertex $u$ assigns its candidate colour $\colour$ if no neighbour $w \in N^+(v) = \{w \in N(v) | \deg(w) \geq \deg(v)\}$ has the same candidate colour $\colour$, it suffices to show that no vertex of the set $V \cap \B_0(r_0)$ tries the same colour as any of its larger degree neighbours.
    
    Now, consider the fixed vertex $v \in V_\ell$ with $\ell \geq \lfloor R - r_0 \rfloor$ since $v \in V \cap \B_0(r_0)$. Then, for the set of larger degree vertices $N^+(v)$ the size is upper bounded by $|N^+(v)| = \deghigh{u} \leq \Theta(1)e^{\ell(1-\alpha)}$ using \Cref{lem:neighbours-larger-degree} for a typical HRG. Next, reveal first the candidate colour $\colour$ in the first round of the fixed vertex $v$. Thus, for $p_v := \Pro{\exists u \in N^+(v) \text{ with candidate colour } \colour.}$ we obtain 
    $$p_v = (1 - |\palette|^{-1})^{\deghigh{v}} \leq \Theta(1)e^{\ell(1-\alpha)}\cdot \log^{-1}(n) \cdot n^{-\frac{1}{4\alpha - 1}},$$

    by using $|\palette| \geq \log(n)\cdot n^{\frac{1}{4\alpha - 1}}$. Next, let $B_\ell$ be the event that there exists a vertex $u \in V_\ell$ that has the same candidate colour as any of its larger degree neighbours. A union bound gives us

    $$
    \Pro{B_\ell} \leq \sum_{v \in V_\ell} p_v \leq \Theta(1)\log^{-1}(n) \cdot n^{1 - \frac{1}{4\alpha - 1}} \cdot e^{-\ell(2\alpha -1)},
    $$

    where we used that $|V_\ell| \in \Theta(ne^{-\alpha\ell})$ for a typical HRG by \Cref{lem:layer-properties}. Next, let $B$ be the event that there exists a vertex $u \in V \cap \B_0(r_0)$, that has the same cadidate colour as any of its larger degree neighbourhood $N^+(u)$. Another union bound then gives us

    $$
    \Pro{B} \leq \sum_{\ell = \lfloor R-r_0\rfloor}^{\finallayer} \Pro{B_\ell} \leq \sum_{\ell = \lfloor R-r_0\rfloor}^{\finallayer}\Theta(1)\log^{-1}(n) \cdot n^{1 - \frac{1}{4\alpha - 1}} \cdot e^{-\ell(2\alpha -1)}.
    $$

    Since $\alpha > 1/2$, this is upper bounded by the integral

    $$
    \Pro{B} \leq \int_{\ell = \lfloor R-r_0\rfloor}^{\finallayer}\Theta(1)\log^{-1}(n) \cdot n^{1 - \frac{1}{4\alpha - 1}} \cdot e^{-\ell(2\alpha -1)} \dd \ell \in \bigO(1/\log(n)),
    $$
    where we also used that $\ell \geq \lfloor R-r_0\rfloor = \frac{2\log(n)}{4\alpha - 1}$. This shows that all vertices up to radius $r_0$ are coloured by \RCTDEG after the first round with probability $1-\bigO(1/\log(n))$ and finishes the first part of the proof.

    Next, we consider the set of vertices with radius larger than $r_0$. Observe that the expected degree of a vertex $u \in V \cap (\B_0(R) \setminus \B_0(r_0))$ is $\deg(u) \leq \Theta(1)n e^{-r_0/2}$ by \Cref{lem:vertex-degree}. By our choice of $r_0$, a Chernoff and union bound it follows for all vertices $V \cap (\B_0(R) \setminus \B_0(r_0))$ that their degree is at most $\bigO(n^{\frac{1}{4\alpha - 1}}) \in o(|\palette|)$ \wehp which then also holds for a typical HRG. This implies that at any round $t$, any vertex $u \in V \cap (\B_0(R) \setminus \B_0(r_0))$ has a colour space of size $\colourSpace{t}{u} \in \Omega(|\palette|) \in \Omega(\log(n) \cdot n^{\frac{1}{4\alpha - 1}})$. Let $p_u := \Pro{\colourchosen{t}{u} = \emptyset}$ be the probability that $u$ is uncoloured after round $t$ of the \RCTDEG. We distinguish for two different regimes of layers:
    
    \smallskip
    
    \textbf{Case 1}[$\ell \in \omega(\log(\log(n)))$]: For $u \in V_\ell$ with $\ell \in \omega(\log(\log(n)))$ we get by \Cref{lem:neighbours-larger-degree} 
    
    $$p_u \leq \left(\frac{\deghigh{u}}{\colourSpace{t}{u}}\right)^t \in \bigO\left(\frac{e^{\ell(1-\alpha)}}{\log(n) \cdot n^{\frac{1}{4\alpha - 1}}}\right)^t.$$

Next, let $B_\ell$ be the event that there exists a vertex $u \in V_\ell$ that is uncoloured after round $t$. A union bound yields

    $$
    \Pro{B_\ell} \leq \sum_{v \in V_\ell} p_v \in \bigO\left(ne^{-\alpha\ell}\left(\frac{e^{\ell(1-\alpha)}}{\log(n) \cdot n^{\frac{1}{4\alpha - 1}}}\right)^t\right),
    $$

    where we used that $|V_\ell| \in \Theta(ne^{-\alpha\ell})$ for a typical HRG by \Cref{lem:layer-properties}. Notice that for vertices with radius larger that $r_0$ this is upper bounded for any $\ell \leq \lceil R- r_0\rceil$ by

    $$
    \Pro{B_\ell} \in \bigO\left(n\left(\frac{e^{(\lceil R- r_0\rceil)(1-\alpha)}}{\log(n) \cdot n^{\frac{1}{4\alpha - 1}}}\right)^t\right) \in o\left(n^{1-t\frac{2\alpha -1}{4\alpha - 1}}\right),
    $$
since $\alpha <1$ and we used that $\ell \leq \lceil R-r_0\rceil \leq \frac{2\log(n)}{4\alpha - 1}$. Since $\alpha > 1/2$ we have $\Pro{B_\ell} \in o(1/n)$ for $t \geq \frac{8\alpha}{2\alpha -1}$. A union bound over all layers gives that for $\ell \in \omega(\log(\log(n)))$ no layer $\mathcal{A}_\ell$ contains an uncoloured vertex after constant rounds with probability $1-\bigO(1/\log(n))$

\smallskip

\textbf{Case 2}[$\ell \in \bigO(\log(\log(n)))$]: The only difference to the previous case is that $\deghigh{u} \in \Tilde{\bigO}(e^{\ell/2})$ by \Cref{lem:neighbours-larger-degree} and we get for any $u \in V \cap (\B_0(R)\setminus \B_0(R-\ell))$ that the probability that $u$ is uncoloured after round $t$ is at most $p_u \leq  \left({\deghigh{u}}/\colourSpace{t}{u}\right)^t \in \Tilde{\bigO}\left(n^{-{t}/{(4\alpha -1)}}\right)$. Thus, setting the number of rounds $t \geq 8\alpha$, the probability that there exists an uncoloured vertex $u \in V \cap (\B_0(R)\setminus \B_0(R-\ell))$ after $t$ rounds is $\Pro{B} \in o(1/n)$ via a union bound. This finishes the case and also the proof.
\end{proof}

We now put all lemmas of the section together to obtain the second part of \Cref{the:rctdeg}.

\begin{proposition}[\RCTDEG constant rounds \aas]\label{pro:rctdeg-constant-rounds}
     Let $G\sim\hrg$ by a typical hyperbolic random graph and let $\minexp' = \min((2\alpha + 1/2)^{-1},(4\alpha -1)^{-1})$. Then, \RCTDEG with colour space $|\palette| \geq \log(n)\cdot n^{\minexp'}$ terminates after $\bigO(1)$ rounds with probability $1 - \bigO(1/\log(n))$.
\end{proposition}
\begin{proof}
\textbf{Case 1}~$\left[|\palette| \geq \log(n)\cdot n^{\frac{1}{2\alpha + 1/2}}\right]$: After 1 round: Using \Cref{lem:round-1-rctdeg-rainbow} with $f(n) = \log(n)$, it follows that after the first round all vertices up to radius $r\leq R(1-(2\alpha + 1/2)^{1/2})$ are coloured with probability $1 - \bigO(1/\log(n))$.

After $\bigO(1)$ rounds: After constant rounds, all vertices with radius $r >  R(1-(2\alpha + 1/2)^{1/2})$ are coloured due to \Cref{lem:second-round-rctdeg-constant} with probability $1 - \bigO(1/\log(n))$. 

Thus, after constant rounds all vertices of a typical HRG are coloured by \RCTDEG with colour space $|\palette| \geq \log(n)\cdot n^{\frac{1}{2\alpha + 1/2}}$ with probability $1 - \bigO(1/\log(n))$. 

\smallskip

\textbf{Case 2}~$\left[|\palette| \geq \log(n)\cdot n^{\frac{1}{4\alpha -1}}\right]$: By \Cref{lem:rctdeg-large-alpha-upper} a typical HRG is coloured after $\bigO(1)$ rounds with probability  $1 - \bigO(1/\log(n))$. 

\end{proof}

\subsection{Theorem~\ref{the:rctdeg}, part 3: Limitations of Random Colour Trial with Degree Priority}\label{sec:colour-space-rctdeg}\label{sec:rctdeg-part3}

In this section we translate the results of \Cref{sec:colour-space} from \RCT to \RCTDEG. Note that, given we use the same colour space for \RCTDEG as for \RCT, the birthday paradox as stated in~\Cref{lem:birthday-paradox} holds also true for \RCTDEG as a candidate colour is choosen uniform at random.

\begin{corollary}[\RCTDEG birthday paradox]\label{cor:rctdeg-birthday-paradox}
    Let $G\sim \hrg$ be a typical hyperbolic random graph, let $\minexp = \min((2\alpha + 1/2)^{-1}, 2(1-\alpha))$ and let $\eps \in (0,1)$ be constant small enough. Then in the first round of \RCTDEG with colour space at most $|\palette| \leq \epsilon/\log(n) \cdot n^{\minexp}$, there exists a pair of neighbours $\{u,v\}\in E(G)$ such that $u$ and $v$ try the same candidate colour and $|L(u)|, |L(v)| \in \Omega(n^{\minexp})$.
\end{corollary}

From this we get in similar fashion to~\Cref{pro:dist-coul-never} that at least one of the two vertices $u,v \in E(G)$ as given in \Cref{cor:rctdeg-birthday-paradox} will never be coloured  \wehp 

\begin{proposition}[\RCTDEG lower bound]\label{pro:lower-bound-rctdeg}
    Let $G\sim \hrg$ be a typical hyperbolic random graph, let $\minexp = \min((2\alpha + 1/2)^{-1}, 2(1-\alpha))$ and let $\eps \in (0,1)$ be constant small enough. Then \RCTDEG with colour space at most $|\palette| \leq \eps/\log(n) \cdot n^{\minexp}$ will never terminate on $G$ \wehp 
\end{proposition}

\begin{proof}

Consider the pair of neighbours $u,v \in E(G)$ given by \Cref{cor:rctdeg-birthday-paradox} \wehp, that share the same candidate colour and which have both $|L(u)|, |L(v)| \in \Omega(n^{\minexp})$ many leaves and fix $u$. Then, fix a colour $\colour \in \palette$ and let $Z_{\colour}$ be the random variable that counts the number of leaves of $u$, that try candidate colour $\colour$ in the first round. Moreover, let $Z_w$ be the indicator random variable that is $1$ if $w \in L(u)$ tries the fixed candidate colour in the first round. Since in the first round a vertex tries a candidate colour u.a.r. from $\palette$, we get  by linearity of expectation and $u$ having $|L(u)| \in \Theta(n^{\minexp})$ leaves that 
       \begin{align}\label{eq:leaves-consume-space2}
           \Exp{Z_{\colour}} = \Exp{\sum_{w \in L(u)} Z_w}\geq \Theta(1)|L(u)|\cdot|\palette|^{-1} \geq \Theta(1)\eps^{-1}\log(n),
       \end{align}
      
       since $|\palette| \leq \eps/\log(n) \cdot n^{\minexp}$. A   Chernoff bound for $Z_{\colour}$ gives that there exists at least one leaf $w \in L(u)$ that tries $\colour$ with probability $\Pro{Z_\psi \geq 1} \geq 1 - n^{-\Theta(1)/\eps}$. Thus, given that $\eps > 0$ is small enough, a union bound over all colours yields that any colour is tried at least once by a leaf \wehp Since leaves only have an edge to $u$, this implies that none but the colour pick of $u$ is discarded after the first round finishes \wehp 

       Via union bound this holds also true for $v$. Without loss of generality let $\deg(u) \geq \deg(v)$. Then, since $u$ and $v$ have the same candidate colour and \RCTDEG gives priority in the colour choice to the vertex with larger degree, $v$ does not assign its candidate colour (in case of $\deg(u) = \deg(v)$ consider without less of generality that $u$ has the smaller ID among the two). Since after the first round all colours except the candidate colour of $v$ are assigned to its leaves \wehp, after round 1 vertex $v$ has an empty colour space $\colourSpace{2}{v} = 0$ while being uncoloured after the first round. Thus, \RCTDEG never terminates \wehp and the proof is complete.
\end{proof}

\subsection{Putting everything together: Proof of Theorem~\ref{the:rctdeg}}\label{sec:all-together}
By the propositions that we collected throughout the section we now obtain our theorem for \RCTDEG on hyperbolic random graphs.

\rctdegTheorem*

\begin{proof}
    For the first part, we obtain by~\Cref{pro:rctdeg-upperbound} that \RCTDEG terminates \aas within two rounds on a typical hyperbolic random graph with $\alpha \leq 3/4$ and a colour space $|\palette| \geq \log^4 (n) \cdot n^{\maxexp} = \log^4(n)\cdot n^{2(1-\alpha)}$. Since the chromatic number for a typical HRG is $\chi \in \Theta(n^{1-\alpha})$~(\cite[Corollary 10]{bmrs-stacs-25}), the result for $\alpha \leq 3/4$ follows.

    For $\alpha > 3/4$ we have by~\Cref{pro:rctdeg-upperbound} that \RCTDEG terminates \aas within two rounds on a typical hyperbolic random graph if the colour space is at least $|\palette| \geq \log^4 (n) \cdot n^{\maxexp} = \log^4\cdot \sqrt{n}$ which finishes the first part.

    For the second part, we use~\Cref{pro:rctdeg-constant-rounds} by which \RCTDEG finishes \aas after constant rounds on a typical hyperbolic random graph with $\alpha \leq 3/4$ if the colour space is at least $|\palette| \geq \log(n) \cdot n^{\minexp'} = n^{\frac{1}{2\alpha + 1/2}}$.

    If $\alpha >3/4$ for a typical hyperbolic random graph, another application of~\Cref{pro:rctdeg-constant-rounds} yields that \RCTDEG finishes \aas after constant rounds on given that the colour space is at least $|\palette| \geq \log(n) \cdot n^{\minexp'} = n^{\frac{1}{4\alpha - 1}}$.

    For the final part of the theorem, we apply~\Cref{pro:lower-bound-rctdeg} such that for $\alpha \leq 3/4$, there exists a constant $\delta >0$ so that\RCTDEG never terminates on a typical hyperbolic random graph \wehp if the colour space is at most $|\palette| \leq \delta \log^{-1}(n) \cdot n^{\minexp} = \frac{\delta\cdot n^{\frac{1}{2\alpha + 1/2}}}{\log(n)}$.

    Finally, if $\alpha > 3/4$, \RCTDEG never terminates \wehp if the colour space is at most $|\palette| \leq \delta \log^{-1}(n)\cdot n^{\minexp} = \frac{\delta \cdot n^{2(1-\alpha)}}{\log(n)}$~by~\Cref{pro:lower-bound-rctdeg}~where $\delta > 0$ is a constant small enough. Using that the chromatic number for a typical HRG is $\chi \in \Theta(n^{1-\alpha})$~by~(\cite[Corollary 10]{bmrs-stacs-25}) concludes the proof.
\end{proof}

\section{Deterministic LOCAL Colouring (Proof of Theorem~\ref{the:deterministic-colouring})}\label{sec:deterministic}
In the previous sections we looked into randomised distributed algorithms. Here we shed some light on deterministic colouring and consider a deterministic algorithm that is tailored towards a hyperbolic random graph.

To this end, we partition our disk $\disk$ into an \emph{outer} and an \emph{inner} part. More concretely, let $\eps > 0$ be a fixed constant and let $O(\eps) := \{v \in V : r(v) \geq (2-\eps)\log(n)\}$ be the set of $\emph{outer vertices}$ and $I(\eps) := V\setminus O(\eps)$ the set of $\emph{inner vertices}$. We first show that $G[I(\eps)]$ has a constant diameter in \Cref{lem:const-diameter}. We combine this with the fact that the vertices in the outer disk have a smallish degree (shown in \Cref{lem:outer-set}), to show that in a constant number of rounds we can colour a hyperbolic random graph with $(1+o(1))\chi$ colours in \Cref{the:deterministic-colouring}.

The following shows that vertices of the inner disk have a bounded diameter.

\begin{lemma}[Constant Diameter]\label{lem:const-diameter}
    Let $G \sim \mathcal{G}(n, \alpha, C)$ be a threshold hyperbolic random graph and let $G_I = G [V \cap \mathcal{B}_0 ((2 - \varepsilon)log(n))]$. Then for
any $0 < \varepsilon < 1$ there exists $n \geq n_0$ large enough such that the diameter for $G_I$ is constant \wehp
\end{lemma} 
\begin{proof}
    We prove the statement in two steps. First, we show that any vertex $u \in V_\ell$ where $\ell \geq \lfloor\eps\log(n)\rfloor$ with $\eps>0$ being an arbitrary small constant, there exist \wehp a vertex $v \in V_{\ell(1 + \delta)}$ (whereby
$\delta > 0$ is a constant) such that $u$ and $v$ are adjacent. We then proceed by using this argument iteratively to argue that $u$ reaches another vertex having distance $\ell(1 + \delta)^c$ to the boundary within $c$ hops
\wehp Finally, setting $c \in \bigO(1)$ large enough we reach a vertex with radius at most $R/2$. Since all vertices
in $\mathcal{B}_0(R/2)$ form a clique by the triangle inequality, we therefore obtain that all nodes of $G_I$ form a
connected component with diameter $\bigO(1)$. 

Indeed, let $\ell \geq \lfloor\eps log(n)\rfloor$ and fix a vertex $u \in V_\ell$. Moreover, we set with hindsight $\delta:= (1-\alpha)/(2\alpha-1)$ and let $B$ be the event that $u$ has no neighbour $v \in V_{\ell'}$ where $\ell' = \ell(1 + \delta)$. We upper bound the probability of event $B$. By~\Cref{lem:max-angle} it holds that $u$ has an edge to $v$ if the angle distance is at least $\theta_R(R-\ell, R-\ell') \geq \Theta(1)n^{-1}e^{(\ell + \ell')/2}$. Moreover, by \Cref{eq:layer-measure}, the expected number of vertices in layer $\mathcal{A}_{\ell'}$ is $\E{|V_{\ell'}|} = \Theta(1) n\cdot e^{-\alpha\ell}$. Let $V'_{\ell'} \subseteq V_{\ell'}$ be the set of vertices in layer $\mathcal{A}_{\ell'}$ that have an edge with $u \in V_\ell$. By linearity of expectation, it follows 

\begin{align*}\label{eq:increasing-prob}
    \E{|V'_{\ell'}|} \geq \theta_R(R-\ell, R-\ell') \cdot \E{|V_\ell|} \geq \Theta(1) e^{-\alpha\ell + (\ell' + \ell)/2}.
\end{align*}

Notice that event $B$ is equivalent to the event that $|V'_{\ell'}| = 0$. Since $|V'_{\ell'}|$ follows a Poisson distribution we now get

\begin{align*}
   \Prob{B} = e^{-\E{|V'_{\ell'}|} }&\leq \exp(-\Theta(1) \cdot e^{-\alpha\ell(1+\delta)} \cdot e^{(\ell+\ell(1+\delta))/2})\\
&= \exp(-\Theta(1) \cdot e^{(\ell(1+\delta))(1/2-\alpha)+\ell/2})\\
&= \exp(-n^{\Omega(1)}), 
\end{align*}

where the last line follows as $\ell \geq \eps \log(n)$ and our choice of $\delta = (1-\alpha)/(2\alpha-1)$. Applying a union bound we get that all vertices in layer $\mathcal{A}_\ell$ reach another vertex in layer $\mathcal{A}_{\ell'}$ via one hop \wehp We conclude that with $c$ hops a vertex with in layer level $\ell \leq \lfloor\eps\log(n)\rfloor$ reaches a vertex in a layer level $\ell(c) =  \lfloor(1 + \delta)^c \eps \log(n)\rfloor$ \wehp Thus, setting $c = \left\lceil\frac{\log(1/\eps)}{\log(1+\delta)}\right\rceil$ we reach a vertex in $\mathcal{B}_O(R/2)$ within  $c \in \bigO(1)$ hops \wehp The claim now follows by another application of a union bound.
\end{proof} 

Recall that the set of outer vertices is the set $O(\eps) = \{v \in V : r(v) \geq (2-\eps)\log(n)\}$. The following lemma shows that for an $\eps$ small enough, the maximum degree $\Delta(O) := \max_{v \in O(\eps)} \deg(v)$ of the outer set of vertices $O(\eps)$ is not too large.

\begin{lemma}\label{lem:outer-set}
    Let $G\sim \hrg$ be a threshold hyperbolic random graph and let $\eps > 0$ be a constant. Then for $O:=O(\eps)$ it holds $\Delta(O) \leq \Theta(1)\cdot n^{\eps/2}$ \wehp
\end{lemma}

\begin{proof}
    Recall that the degree of a vertex $u \in V$ is given by $\E{\deg(u)} = \Theta(1) e^{(R-r(u))/2}$ using \Cref{lem:vertex-degree}. Then, using that the minimum radius of a vertex $u \in O(\eps)$ is $r(u) \geq (2-\eps)\log(n)$ and that the expected degree of a vertex is monotonically increasing in $r$, we obtain for a fixed vertex $u \in O(\eps)$ that

    $$
    \E{\deg(u)} \leq \Theta(1)e^{(R-r(u))/2} \leq \Theta(1)e^{(2\log(n)-(2-\eps)\log(n))/2} \leq\Theta(1) n^{\eps/2},
    $$

    where we used that $R = 2\log(n) + C$. A Chernoff bound then shows that $\deg(u) \leq\Theta(1) n^{\eps/2}$ \wehp A union bound over all vertices in $O(\eps)$ finishes the proof.
\end{proof}

We now colour the inner vertices $I$ and the outer vertices $O$ with two different colour palettes of combined size $(1+o(1))\chi$ in constant rounds.

\deterministicLOCAL*

\begin{proof}
We set $\palette(I) := \{1,2,.., \chi\}$ and $\palette(O) :=\{\chi + 1, \chi+2,.., \chi + n^{(1-\alpha)\frac{3}{4}}\}$ such that $\palette(I) \cap \palette(O) = \emptyset$ and $|\palette(I)| + |\palette(O)| \leq (1+o(1))\chi$ \wehp using that the chromatic number of a threshold hyperbolic random graph is \wehp given by $\chi \in \Theta\left(n^{1-\alpha}\right)$~\cite[Corollary 10]{bmrs-stacs-25}.

$\palette(I)$ is the set of inner colours we use to colour $I(\eps)$ and $\palette(O)$ is the set of outer colours we use for $O(\eps)$. To do so, we set $\eps := (1-\alpha)/2$ and obtain via \Cref{lem:outer-set} that the maximum degree in of the vertices in $O$ is bounded by $\Delta(O) \leq \Theta(1)n^{(1-\alpha)/4}$. Now, one round of Linial's algorithm can colour a graph with maximum degree $\Delta$ and unique IDs from an ID space of size $m$ with $\bigO(\Delta^2\log m)$ colours~(see \cite{Linial-92, M21}). As the maximum degree of $O:=O(\eps)$ by our choice of $\eps$ is $\Delta(O)\leq \Theta(1)n^{(1-\alpha)/4}$ \wehp, we can colour $O$ with 

$$\bigO(\Delta(O)^2\log n) \in \Tilde{\bigO}\left(n^{(1-\alpha)/2}\right) \in o(|\palette(O)|)$$ 

colours in a single round \wehp Thus, we use $\palette(O)$ to colour $O$ within a round \wehp

It is left to colour $I:=I(\eps)$ for our choice of $\eps$. Since the diameter of $G[I]$ is constant \wehp by \Cref{lem:const-diameter}, we can brute force $I$ with chromatic number colours in constant rounds. Thus, we can and we will use $\palette(I)$ to colour $I$ in constant time. Since $\palette(I) \cap \palette(O) = \emptyset$ and $|\palette(I)| + |\palette(O)| \leq (1+o(1))\chi$ \wehp by construction, $I$ and $O$ combined are coloured with different colour palettes and with in total $(1+o(1))\chi$ colours in constant rounds. Since $V = I \cup O$ this concludes the proof. \end{proof}

\printbibliography 

@String{FOCS = {Proceedings of the IEEE Symposium on Foundations of Computer Science (FOCS)}}

@String{PODC = {Proceedings of the ACM Symposium on Principles of Distributed Computing (PODC)}}

@String{SODA = {Proceedings of the SIAM-ACM Symposium on Discrete Algorithms (SODA)}}

@String{STOC = {Proceedings of the ACM Symposium on Theory of Computing (STOC)}}

@String{DISC = {Proceedings of the International Symposium on Distributed Computing (DISC)}}

@String{STACS = {{Proceedings of the International Symposium on Theoretical Aspects of Computer Science (STACS)}}}

@book{mu-pc-05,
  author =       {Mitzenmacher, Michael and Upfal, Eli},
  title =        {{Probability and Computing: Randomized Algorithms and
                  Probabilistic Analysis}},
  year =         2005,
  publisher =    {Cambridge University Press},
  DOI =          {10.1017/CBO9780511813603},
  place =        {Cambridge},
}

@inproceedings{HNT22,
  author    = {Magnús M. Halldórsson and
  Alexandre Nolin and 
  Tigran Tonoyan},
  title     = {Overcoming congestion in distributed coloring},
  booktitle = {PODC'22},
  year = {2022},
  doi          = {10.1145/3519270.3538438}
}

@inproceedings{M21,
  author       = {Yannic Maus},
  title        = {Distributed Graph Coloring Made Easy},
  booktitle    = {SPAA'21},
  year         = {2021},
  doi          = {10.1145/3409964.3461804},
}

@Article{alon86,
  author  = {Noga Alon and L\'{a}sl\'{o} Babai and Alon Itai},
  journal = {Journal of Algorithms},
  title   = {{A Fast and Simple Randomized Parallel Algorithm for the Maximal Independent Set Problem}},
  year    = {1986}
}

@Article{luby86,
  author  = {M. Luby},
  journal = {SIAM Journal on Computing},
  title   = {{A Simple Parallel Algorithm for the Maximal Independent Set Problem}},
  year    = {1986}
}

@book{p-rgg-03,
  author =       {Penrose, Mathew},
  title =        {Random Geometric Graphs},
  year =         2003,
  publisher =    {Oxford University Press},
  isbn =         9780198506263,
}

@article{vhhk-s-19,
  author =       {Voitalov, Ivan and van der Hoorn, Pim and van der
                  Hofstad, Remco and Krioukov, Dmitri},
  title =        {Scale-free Networks Well Done},
  journal =      {Physical Review Research},

  year =         2019,
  doi =          {10.1103/PhysRevResearch.1.033034}
}

@article{kpk-h-10,
  author =       {Krioukov, Dmitri and Papadopoulos, Fragkiskos and Kitsak,
                  Maksim and Vahdat, Amin and Bogu\~n\'a, Mari\'an},
  title =        {Hyperbolic Geometry of Complex Networks},
  journal =      {Physical Review E},
  year =         2010,
  doi =          {10.1103/PhysRevE.82.036106},
}

@article{bf-evacaga-22,
 title     = {On the External Validity of Average-case Analyses of Graph Algorithms},
  author    = {Thomas Bläsius and Philipp Fischbeck},
  year      = {2024},
  journal   = {ACM Transactions on Algorithms (TALG)},
  doi       = {10.1145/3633778}
}

@article{bfm-giant-15,
    author = {Bode, Michel and Fountoulakis, N. and Müller, Tobias},
year = {2015},
title = {On the largest component of a hyperbolic model of complex networks},
journal = {Electronic Journal of Combinatorics},
doi = {10.1214/17-AAP1314}
}

@article{fm-giant-18,
 URL = {https://www.jstor.org/stable/26542317},
 author = {Nikolaos Fountoulakis and Tobias Müller},
 journal = {The Annals of Applied Probability},
 title = {Law of large numbers for the largest component in a hyperbolic model of complex networks},
 year = {2018}
}

@inproceedings{gpp-rhg-12,
  author =       {Luca Gugelmann and Konstantinos Panagiotou and Ueli
                  Peter},
  title =        {{Random Hyperbolic Graphs: Degree Sequence and
                  Clustering}},
  booktitle =    {ICALP'12},
  year =         2012,
  doi =          {10.1007/978-3-642-31585-5_51},
}

@article{ms-k-19,
  author =       {Müller, Tobias and Staps, Merlijn},
  title =        {{The Diameter of KPKVB Random Graphs}},
  journal =      {Advances in Applied Probability},
  year =         2019,
  DOI =          {10.1017/apr.2019.23},
}

@article{km-slcrhg-19,
  author =       {Kiwi, Marcos and Mitsche, Dieter},
  title =        {On the Second Largest Component of Random Hyperbolic Graphs},
  journal =      {SIAM Journal on Discrete Mathematics},
  year =         2019,
  doi =          {10.1137/18M121201X},
}

@article{fk-dhrg-18,
  author =       {Tobias Friedrich and Anton Krohmer},
  title =        {{On the Diameter of Hyperbolic Random Graphs}},
  journal =      {SIAM Journal on Discrete Mathematics},
  year =         2018,
  doi =          {10.1137/17M1123961},
}

@article{bfk-chrg-18,
author = {Bl\"{a}sius, Thomas and Friedrich, Tobias and Krohmer, Anton},
title = {Cliques in Hyperbolic Random Graphs},
year = {2018},
journal = {Algorithmica},
doi = {10.1007/s00453-017-0323-3}
}

@inproceedings{HKMT21,
  author    = {Magn{\'{u}}s M. Halld{\'{o}}rsson and
               Fabian Kuhn and
               Yannic Maus and
               Tigran Tonoyan},
  title     = {Efficient randomized distributed coloring in {CONGEST}},
  booktitle = {STOC'21},
  year      = {2021},
}

@phdthesis{katz-diss-23,
  doi = {10.25932/PUBLISHUP-58296},
  url = {https://publishup.uni-potsdam.de/58296},
  author = {Katzmann,  Maximilian},
  title = {About the analysis of algorithms on networks with underlying hyperbolic geometry},
  school = {Universit\"{a}t Potsdam},
  year = {2023},
}

@article{BE10,
  author       = {Leonid Barenboim and
                  Michael Elkin},
  title        = {Sublogarithmic distributed {MIS} algorithm for sparse graphs using
                  Nash-Williams decomposition},
  journal      = {Distributed Comput.},
  year         = {2010},
  doi          = {10.1007/s00446-009-0088-2}
}

@InProceedings{bks-hudg-23,
  author =	{Bl\"{a}sius, Thomas and Friedrich, Tobias and Katzmann, Maximilian and Stephan, Daniel},
  title =	{{Strongly Hyperbolic Unit Disk Graphs}},
  booktitle =	{STACS'23},
year =	 2023,
  doi =		{10.4230/LIPIcs.STACS.2023.13}
}

@inproceedings{bmrs-stacs-25,
  author       = {Samuel Baguley and
                  Yannic Maus and
                  Janosch Ruff and
                  George Skretas},
  title        = {Hyperbolic Random Graphs: Clique Number and Degeneracy with Implications
                  for Colouring},
  booktitle    = {STACS'25},
  year         = {2025},
  doi          = {10.4230/LIPICS.STACS.2025.13}
}

@inproceedings{bfkrz-esa-2023,
  author       = {Thomas Bl{\"{a}}sius and
                  Tobias Friedrich and
                  Maximilian Katzmann and
                  Janosch Ruff and
                  Ziena Zeif},
 
  title        = {On the Giant Component of Geometric Inhomogeneous Random Graphs},
  booktitle    = {ESA'23},
  year         = {2023},
  doi          = {10.4230/LIPICS.ESA.2023.20}
}

@inproceedings{pkbv-info-2010,
  author       = {Fragkiskos Papadopoulos and
                  Dmitri V. Krioukov and
                  Mari{\'{a}}n Bogu{\~{n}}{\'{a}} and
                  Amin Vahdat},
  title        = {Greedy Forwarding in Dynamic Scale-Free Networks Embedded in Hyperbolic
                  Metric Spaces},
  booktitle    = {INFOCOM'10},
  year         = {2010},
  doi          = {10.1109/INFCOM.2010.5462131}
}

@article{Kiwi2024,
  title = {Cover and hitting times of hyperbolic random graphs},
  DOI = {10.1002/rsa.21249},
  journal = {Random Structures \& Algorithms},
  author = {Kiwi,  Marcos and Schepers,  Markus and Sylvester,  John},
  year = {2024}
}

@article{Newman2003,
  title = {Why social networks are different from other types of networks},
  DOI = {10.1103/physreve.68.036122},
  journal = {Phys. Rev. E},
  author = {Newman,  M. E. J. and Park,  Juyong},
  year = {2003}
}

@article{PhysRevE.74.056114,
  title = {Clustering in complex networks. {I.} General formalism},
  author = {Serrano, M. \'Angeles and Bogu\~n\'a, Mari\'an},
  journal = {Phys. Rev. E},
  year = {2006},
  doi = {10.1103/PhysRevE.74.056114}
}

@article{Watts1998,
  title = {Collective dynamics of ‘small-world’ networks},
  DOI = {10.1038/30918},
  journal = {Nature},
  author = {Watts,  Duncan J. and Strogatz,  Steven H.},
  year = {1998}
}

@article{faloutsos1999power,
  title={On power-law relationships of the internet topology},
  author={Faloutsos, Michalis and Faloutsos, Petros and Faloutsos, Christos},
  journal={ACM SIGCOMM computer communication review},
  year={1999}
}

@inproceedings{DBLP:conf/analco/KiwiM15,
  author       = {Marcos Kiwi and
                  Dieter Mitsche},
  title        = {A Bound for the Diameter of Random Hyperbolic Graphs},
  booktitle    = {ANALCO'15},
  year         = {2015},
  doi          = {10.1137/1.9781611973761.3}
}

@inproceedings{bfk-tw-2016,
  title     = {Hyperbolic Random Graphs: Separators and Treewidth},
  author    = {Thomas Bläsius and Tobias Friedrich and Anton Krohmer},
  year      = {2016},
  booktitle = {ESA},
  doi       = {10.4230/LIPIcs.ESA.2016.15}
}

@article{katzmann-approxvc-2023,
  author       = {Thomas Bl{\"{a}}sius and
                  Tobias Friedrich and
                  Maximilian Katzmann},
  title        = {Efficiently Approximating Vertex Cover on Scale-Free Networks with
                  Underlying Hyperbolic Geometry},
  journal      = {Algorithmica},
  year         = {2023}
}

@article{katzmann-exactvc-2023,
  author       = {Thomas Bl{\"{a}}sius and
                  Philipp Fischbeck and
                  Tobias Friedrich and
                  Maximilian Katzmann},
  title        = {Solving Vertex Cover in Polynomial Time on Hyperbolic Random Graphs},
  journal      = {Theory Comput. Syst.},
  year         = {2023}
}

@article{bffkmm-icalp-2022,
  title     = {Efficient Shortest Paths in Scale-Free Networks with Underlying Hyperbolic Geometry},
  author    = {Thomas Bläsius and Cedric Freiberger and Tobias Friedrich and Maximilian Katzmann and Felix Montenegro-Retana and Marianne Thieffry},
  year      = {2022},
  journal   = {ACM Transactions on Algorithms (TALG)},
  doi       = {10.1145/3516483}
}

@article{Linial-92,
author = {Linial, Nathan},
title = {Locality in Distributed Graph Algorithms},
journal = {SIAM Journal on Computing},
year = {1992},
doi = {10.1137/0221015}
}

@InProceedings{BEG17,
  author  = {Leonid Barenboim and Michael Elkin and U. Goldenberg},
  title   = {Locally-Iterative Distributed (Delta + 1)-Coloring below Szegedy-Vishwanathan Barrier, and Applications to Self-Stabilization and to Restricted-Bandwidth Models},
   booktitle = {PODC'18},
  year      = {2018},
}

@inproceedings{FHK16,
author    = {Pierre Fraigniaud and
               Marc Heinrich and
               Adrian Kosowski},
  title     = {Local Conflict Coloring},
  booktitle = {FOCS'16},
  year      = {2016}
  }

@InProceedings{szegedy93,
	author =       {M. Szegedy and S. Vishwanathan},
	title =        {Locality Based Graph Coloring},
	booktitle =    {STOC'93},
	year =         1993
}

@InProceedings{Kuhn2006On,
	author =       {F. Kuhn and R. Wattenhofer},
	title =        {On the Complexity of Distributed Graph Coloring},
	booktitle =    {PODC'06},
	year =         2006
}

@inproceedings{HN21,
  author       = {Magn{\'{u}}s M. Halld{\'{o}}rsson and
                  Alexandre Nolin},
  title        = {Superfast Coloring in {CONGEST} via Efficient Color Sampling},
  booktitle    = {SIROCCO'21},
  year         = {2021},
  doi          = {10.1007/978-3-030-79527-6\_5}
}

@article{Naor91,
  author    = {Naor, M.},
  title     = {A Lower Bound on Probabilistic Algorithms for Distributive Ring Coloring},
  journal   = {{SIAM} J. Discrete Math.},
  year      = {1991}
}

@inproceedings{GG24,
  author       = {Mohsen Ghaffari and
                  Christoph Grunau},
  title        = {Near-Optimal Deterministic Network Decomposition and Ruling Set, and
                  Improved {MIS}},
  booktitle    = {FOCS'24},
  year         = {2024},
  doi          = {10.1109/FOCS61266.2024.00007}
}

@InProceedings{CLP18,
 author    = {Yi{-}Jun Chang and
               Wenzheng Li and
               Seth Pettie},
  title     = {An optimal distributed ({\(\Delta\)}+1)-coloring algorithm?},
  booktitle = {STOC'18},
  year      = {2018},
}

@inproceedings{RG20,
	author    = {V{\'{a}}clav Rozho\v{n} and
	Mohsen Ghaffari},
	title     = {Polylogarithmic-time deterministic network decomposition and distributed
	derandomization},
	booktitle = {STOC'20},
	year      = {2020}
}

@InProceedings{barenboim15,
	author =       {L. Barenboim},
	title =        {Deterministic ({$\Delta$} + 1)-Coloring in Sublinear (in {$\Delta$}) Time in Static, Dynamic and Faulty Networks},
	booktitle = {PODC'15},
	year =      2015
}

@InProceedings{barenboim10,
	author = 	 {L. Barenboim and M. Elkin},
	title = 	 {Deterministic Distributed Vertex Coloring in Polylogarithmic Time},
	booktitle = {PPDC'10},
	year = 	 2010}

@Article{PS95,
author="Panconesi, Alessandro
and Srinivasan, Aravind",
title="The local nature of {$\Delta$}-coloring and its algorithmic applications",
journal="Combinatorica",
year="1995",
doi="10.1007/BF01200759",
}

@inproceedings{GHKM18,
  author    = {Mohsen Ghaffari and
               Juho Hirvonen and
               Fabian Kuhn and
               Yannic Maus},
  title     = {Improved Distributed Delta-Coloring},
  booktitle = {PODC'18},
  year      = {2018},
  url       = {https://dl.acm.org/citation.cfm?id=3212764},
}

@inproceedings{FHM23,
author = {Manuela Fischer and Magnús M. Halldórsson and Yannic Maus},
title = {Fast Distributed {Brooks'} Theorem},
booktitle = {SODA'23},
doi = {10.1137/1.9781611977554.ch98},
year ={2023}
}

@inproceedings{HM24,
  author       = {Magn{\'{u}}s M. Halld{\'{o}}rsson and
                  Yannic Maus},
  title        = {Distributed Delta-Coloring Under Bandwidth Limitations},
  booktitle    = {DISC'24},
  year         = {2024},
  doi          = {10.4230/LIPICS.DISC.2024.31}
}

@inproceedings{MT20,
	author    = {Yannic Maus and
	Tigran Tonoyan},
	title     = {Local Conflict Coloring Revisited: Linial for Lists},
	booktitle = {DISC'20},
	year      = {2020},
	doi       = {10.4230/LIPIcs.DISC.2020.16}
}

@book{Prob-Method-16,
author = {Alon, Noga and Spencer, Joel H.},
title = {The Probabilistic Method},
 publisher    = {Wiley},
year = {2016}
}

@phdthesis{Krohmer2016,
  author      = {Anton Krohmer},
  title       = {Structures \& algorithms in hyperbolic random graphs},
url          = {https://publishup.uni-potsdam.de/frontdoor/index/index/docId/39597},
  type        = {doctoralthesis},
  school      = {Universit{\"a}t Potsdam},
  year        = {2016},
}

@article{hrg-spectral,
  title = {Spectral gap of random hyperbolic graphs and related parameters},
  DOI = {10.1214/17-aap1323},
  journal = {The Annals of Applied Probability},
  author = {Kiwi,  Marcos and Mitsche,  Dieter},
  year = {2018}
}

@article{DBLP:journals/im/AbdullahFB17,
  author       = {Mohammed Amin Abdullah and
                  Nikolaos Fountoulakis and
                  Michel Bode},
  title        = {Typical distances in a geometric model for complex networks},
  journal      = {Internet Math.},
  year         = {2017},
  doi          = {10.24166/IM.13.2017}

}

@article{BEPS,
author = {Barenboim, Leonid and Elkin, Michael and Pettie, Seth and Schneider, Johannes},
title = {The Locality of Distributed Symmetry Breaking},
year = {2016},
journal = {Journal of the ACM (JACM)},
doi = {10.1145/2903137}
}

@article{TommyAndy,
  title     = {The impact of heterogeneity and geometry on the proof complexity of random satisfiability},
  author    = {Thomas Bläsius and Tobias Friedrich and Andreas Göbel and Jordi Levy and Ralf Rothenberger},
  year      = {2023},
  journal   = {Random Structures \& Algorithms},
  doi       = {10.1002/rsa.21168}
}

@incollection{Doer-21,
  author       = {Benjamin Doerr},
  title        = {Probabilistic Tools for the Analysis of Randomized Optimization Heuristics},
  booktitle    = {Theory of Evolutionary Computation - Recent Developments in Discrete
                  Optimization},
  year         = {2020},
  doi          = {10.1007/978-3-030-29414-4\_1} 
}

@article{random-graphs-distributed-colouring,
title = {Distributed algorithms for random graphs},
journal = {Theoretical Computer Science},
year = {2015},
doi = {https://doi.org/10.1016/j.tcs.2015.08.037},
author = {K. Krzywdziński and K. Rybarczyk}
}

@inproceedings{planar-graphs1,
author = {Chechik, Shiri and Mukhtar, Doron},
title = {Optimal distributed coloring algorithms for planar graphs in the LOCAL model},
year = {2019},
booktitle = {SODA'19}
}

@inproceedings{planar-graphs2,
author = {Aboulker, Pierre and Bonamy, Marthe and Bousquet, Nicolas and Esperet, Louis},
title = {Distributed Coloring in Sparse Graphs with Fewer Colors},
year = {2018},
booktitle = {PODC'18},
doi = {10.1145/3212734.3212740}
}

@INPROCEEDINGS{planar-graphs3,
  author={Postle, Luke},
  booktitle={FOCS'19}, 
  title={Linear-Time and Efficient Distributed Algorithms for List Coloring Graphs on Surfaces}, 
  year={2019},
  doi={10.1109/FOCS.2019.00060}}

@article{4-color-part-1,
  author = {Appel, K. and Haken, W.},
  journal = {Illinois J. Math.},
  title = {Every planar map is four colorable. Part I: Discharging},
doi = {10.1215/ijm/1256049011},
  year = 1977
}

@article{4-color-part-2,
author = {K. Appel and W. Haken and J. Koch},
title = {{Every planar map is four colorable. Part II: Reducibility}},
journal = {Illinois Journal of Mathematics},
year = {1977},
doi = {10.1215/ijm/1256049012}
}

@book{paleg-00,
author = {Peleg, David},
title = {Distributed computing: a locality-sensitive approach},
year = {2000},
publisher = {SIAM}
}

@article{Brooks_1941, title={On colouring the nodes of a network}, DOI={10.1017/S030500410002168X}, journal={Mathematical Proceedings of the Cambridge Philosophical Society}, author={Brooks, R. L.}, year={1941}}

@article{Bogu2010,
  title = {Sustaining the Internet with hyperbolic mapping},
  DOI = {10.1038/ncomms1063},
  journal = {Nature Communications},
  author = {Boguñá,  Marián and Papadopoulos,  Fragkiskos and Krioukov,  Dmitri},
  year = {2010}}

@article{Serrano2008,
  title = {Self-Similarity of Complex Networks and Hidden Metric Spaces},
  DOI = {10.1103/physrevlett.100.078701},
  journal = {Physical Review Letters},
  author = {Serrano,  M. Ángeles and Krioukov,  Dmitri and Boguñá,  Marián},
  year = {2008}
}

@article{HSS18,
  author       = {David G. Harris and
                  Johannes Schneider and
                  Hsin{-}Hao Su},
  title        = {Distributed ({\(\Delta\)} +1)-Coloring in Sublogarithmic Rounds},
  journal      = {Journal of the ACM (JACM)},
  year         = {2018},
  doi          = {10.1145/3178120}
}

@Book{barenboimelkin_book,
  author = 	 {Leonid Barenboim and Michael Elkin},
  title = 	 {Distributed Graph Coloring: Fundamentals and Recent Developments},
  publisher = 	 {Morgan \& Claypool Publishers},
  year = 	 2013}

@inproceedings{SW10,
  author    = {Johannes Schneider and
               Roger Wattenhofer},
  title     = {A new technique for distributed symmetry breaking},
  booktitle = {PODC'10},
  year      = {2010},
  doi       = {10.1145/1835698.1835760}
}

@article{johansson99,
  author    = {{\"{O}}jvind Johansson},
  title     = {Simple Distributed {$\Delta+1$}-Coloring of Graphs},
  journal   = {Inf. Process. Lett.},
  year      = {1999}
}

@inproceedings{HKNT22,
  author       = {Magn{\'{u}}s M. Halld{\'{o}}rsson and
                  Fabian Kuhn and
                  Alexandre Nolin and
                  Tigran Tonoyan},
  title        = {Near-optimal distributed degree+1 coloring},
  booktitle    = {STOC'22},
  year         = {2022},
  doi          = {10.1145/3519935.3520023}
}

@inproceedings{ACK19,
  author       = {Sepehr Assadi and
                  Yu Chen and
                  Sanjeev Khanna},
  title        = {Sublinear Algorithms for ({\(\Delta\)} + 1) Vertex Coloring},
  booktitle    = {SODA'19},
  year         = {2019},
  doi          = {10.1137/1.9781611975482.48}
}

@inproceedings{CDP21,
  author       = {Artur Czumaj and
                  Peter Davies and
                  Merav Parter},
  title        = {Improved Deterministic ({\(\Delta\)}+1) Coloring in Low-Space {MPC}},
  booktitle    = {PODC'21},
  year         = {2021},
  doi          = {10.1145/3465084.3467937}
}

@inproceedings{BBN25,
  author       = {Yann Bourreau and
                  Sebastian Brandt and
                  Alexandre Nolin},
  title        = {Faster Distributed {\(\Delta\)}-Coloring via Ruling Subgraphs},
  booktitle    = {STOC'25},
  year         = {2025},
  doi          = {10.1145/3717823.3718320}
}

@inproceedings{BFHKLRSU16,
  author       = {Sebastian Brandt and
                  Orr Fischer and
                  Juho Hirvonen and
                  Barbara Keller and
                  Tuomo Lempi{\"{a}}inen and
                  Joel Rybicki and
                  Jukka Suomela and
                  Jara Uitto},
  title        = {A lower bound for the distributed Lov{\'{a}}sz local lemma},
  booktitle    = {STOC'16},
  year         = {2016},
  doi          = {10.1145/2897518.2897570}
}

@article{PS15,
  author       = {Seth Pettie and
                  Hsin{-}Hao Su},
  title        = {Distributed coloring algorithms for triangle-free graphs},
  journal      = {Inf. Comput.},
  year         = {2015},
  doi          = {10.1016/J.IC.2014.12.018}
}

@inproceedings{SLOCAL17,
	author    = {Mohsen Ghaffari and
	Fabian Kuhn and
	Yannic Maus},
	title     = {On the complexity of local distributed graph problems},
	booktitle = {STOC'17},
	year      = {2017}
}

@InProceedings{newHypergraphMatching,
	author    = {Mohsen Ghaffari and David G. Harris and Fabian Kuhn},
	title     = {{On Derandomizing Local Distributed Algorithms}},
	booktitle = {FOCS'18},
	year      = {2018}
}

@article{BBHORS21,
  author       = {Alkida Balliu and
                  Sebastian Brandt and
                  Juho Hirvonen and
                  Dennis Olivetti and
                  Mika{\"{e}}l Rabie and
                  Jukka Suomela},
  title        = {Lower Bounds for Maximal Matchings and Maximal Independent Sets},
  journal      = {J. {ACM}},
  volume       = {68},
  number       = {5},
  pages        = {39:1--39:30},
  year         = {2021},
  url          = {https://doi.org/10.1145/3461458},
  doi          = {10.1145/3461458},
  bibsource    = {dblp computer science bibliography, https://dblp.org}
}

\newpage

\appendix

\section{Stochastic Domination for Randomised Colouring}\label{sec:stochastic-dominance}

In this section, we briefly explain the issue with applying stochastic dominance as done in~\cite[Lemma 5.2]{HKMT21}. There, at its core, the authors aim to apply \cite[Lemma 1.8.7]{Doer-21} stated by Doerr, which says that a sequence of arbitrary binary random variables $X_1,\ldots, X_n$ is stochastically dominated by independent random variables $X^*_1,\ldots, X^*_n$ if $$\Pro{X_i = 1 | X_1 = x_1,\ldots, X_{i-1} = x_{i-1}} \leq \Pro{X^*_i = 1}$$

for all $i \in [1..n]$ and all $x_1, \ldots, x_{i-1} \in \{0, 1\}$ with $\Pro{X_1 = x_1, \ldots, X_{i-1} = x_{i-1}} > 0$. The upshot of using stochastic domination lies in the applicability of Chernoff bounds for the random variable $X = X_1 + \ldots +X_n$. 

In~\cite[Lemma 5.2]{HKMT21}, the authors define $Z_u$ to be the indicator random variable that is $1$ if vertex $u$ is uncoloured after iteration $i$ and then basically sate that $\Pro{Z_u = 1} \leq \deg(u)/|\mathsf{ColourSpace}(u)|$, and that this holds \emph{irrespective} of other $Z_{u'}$. While this is true if we condition on the candidate colour of $u'$, observe that in order to apply \cite[Lemma 1.8.7]{Doer-21}, the upper bound on $\Pro{Z_u |Z_{u'} = z_{u'} }$ needs to hold true for \emph{all} $z_{u'} \in \{0,1\}$ with $\Pro{Z_{u'} = z_{u'}} > 0$.

Now, consider as an example a graph that consists of two vertices $V = \{\{u\},\{u'\}\} $ with edge $e = \{u,u'\}$ and at iteration $i$, an uncoloured vertex $v \in V$ picks its candidate colour u.a.r. from all available candidate colours. Let $Z_u$ be the indicator random variable that is $1$ if $Z_u$ is uncoloured after iteration $i$. Fixing a vertex $u$ we have indeed $\Exp{Z_u} \leq \deg(u)/|\mathsf{ColourSpace}(u)|$. However, when we condition on the event $Z_{u'} =1$ which has probability $\Pro{Z_{u'} =1} > 0$, then $u$ must have surely tried the same candidate colour as $u'$ giving $\Pro{Z_{u} | Z_{u'} = 1} = 1 \not\leq \deg(u)/|\mathsf{ColourSpace}(u)|$ and thus, the stochastic domination does not apply. 

\section{Lower Bound for Random Colour Trial on Cliques}
\label{app:RCTLowerBoundWorstCase}
\begin{restatable}[Random colour trial on cliques]{proposition}{ProClique}\label{pro:clique-theorem}
    Let $G$ be a clique with $n$ vertices. Then $\RCT$ with colour space $\Delta + 1$ requires \whp at least $\Omega(\log n)$ rounds to colour $G$.
\end{restatable}

\begin{proof}
    We use the following:

    \begin{claim} For round $t\geq 1$ let $U_t$ be set of uncoloured vertices of $G$ after round $t$. Then consider some round $t$ such that 
        $|U_{t-1}| \in n^{\Omega(1)}$ holds. Then, there exists a constant $\delta > 0$ such that with probability $1 - n^{-\omega(1)}$ the number of uncoloured vertices after round $t$ is at least 
        \begin{equation}\label{eq:claim}
            |U_t| \geq \delta \cdot |U_{t-1}|.
        \end{equation}
        
    \end{claim}

    Before proving the claim, we outline how our desired statement follows by using~\Cref{eq:claim}: let $c > 0$ be a constant that we will choose adequately small enough later, and let $T = c \cdot \log n$. We show that $|U_T| \in n^{\Omega(1)}$ \whp given $c$ is small enough from which our desired statement follows.
    
    Note that we obtain for $t=1$ that $|U_1| \geq \delta \cdot |U_0| \in n^{\Omega(1)}$ with probability $1-n^{-\omega(1)}$ by~\Cref{eq:claim} since $|U_0| = n \in n^{\Omega(1)}$. Subsequently, it holds by a union bound over $T$ rounds that $\delta  \in \Omega(1)$ with probability $1 - n^{-\omega(1)}$ by \Cref{eq:claim}, if for any constant $\eps>0$ we have $|U_0| \cdot \eps^T \in n^{\Omega(1)}$. Indeed, we obtain by setting $c=(2\log(1/\eps))^{-1}$ that  
    
     \begin{align*}
        |U_0|\cdot\eps^T = n \cdot \eps^{T} \geq \eps^{c\cdot \log n}\cdot n = \eps^{c\cdot \log_{\eps}(n)\log(\eps)} \cdot n= n^{1-c\cdot\log(1/\eps)} = \sqrt{n},
    \end{align*}

    and thus, $\delta \in \Omega(1)$ with probability $1 - n^{-\omega(1)}$. We conclude that $|U_T| \geq n \cdot \delta^T \in n^{\Omega(1)}$ with probability $1 - n^{-\omega(1)}$ using $c=(2\log(1/\delta))^{-1}$ for $T = c\log(n)$. This implies that after $\Omega(\log n)$ rounds, \RCT did not finish on a clique \whp and what remains to be proven is \Cref{eq:claim}.
    
    For that, consider any round $t$ and assume $|U_t| \in n^{\Omega(1)}$ as given as the hypothesis of the claim. We count the number of colours that are tried exactly twice by vertices in round $t+1$ by the random variable $Z$. Note that $|U_{t+1}| \geq 2\cdot Z$ since every colour that is used twice implies two uncoloured vertices after round $t+1$. For colour $\colour \in \palette$, let $Z_\colour$ be the indicator random variable that is $1$ if colour $\colour$ is tried by exactly two vertices in round $t+1$ such that $Z = \sum_{\colour \in \palette_{t+1}} Z_\colour$ where $\palette_{t+1}$ is the colour space available at round $t+1$. Since we use $\Delta + 1$ as our colour space and $G$ is a clique it holds for all $t$ that $|\palette_{t+1}| = |U_t| =: n' \in n^{\Omega(1)}$ and thus, $Z \geq n' \cdot \delta/2$ (for a constant $\delta > 0$) with probability $1 - n^{-\omega(1)}$ suffices to prove our desired claim. Now, for round $t+1$, consider a balls into bins experiment with $N \sim \text{Po}(n')$ balls and $n'$ bins. We count the number of bins with exactly two balls (representing the number of colours used exactly twice) by random variable $Z'$ and indicator random variables $Z'_i = 1$ for bin $i$ containing two balls which yields $Z' = \sum_{i=1}^{N} Z'_i$. By Poisson distribution it follows $Z_i' \in \Omega(1)$ and since $n' \in n^{\Omega(1)}$ we get by using linearity of expectation that there exists a $\delta' >0$ such that $\Exp{Z'} = \delta' \cdot n' \in n^{\Omega(1)}$. By a Chernoff bound it follows that there exists a constant $0 < \delta < \delta'$ such that $Z' \geq n' \cdot \delta/2$ with probability $1 - n^{-\omega(1)}$ since $n' \in n^{\Omega(1)}$. We then conclude by "de-Poissonising" the balls into bins experiment~\cite[Corollary 5.9]{mu-pc-05} that also $Z \geq n' \cdot \delta/2 $ with  probability $1 - n^{-\omega(1)}$ what proves the desired claim.
\end{proof}

\section{Complementary proofs of Section~\ref{sec:prelims}}\label{sec:degree}

\degree*

\begin{proof}
Since a vertex $u$ with radius $r$ has $\dist(u,v) \leq R$ if $r(v) \leq R - r$ we have for the area of the ball $B_u(R)$ intersecting the disk $\disk$
\begin{align*}
    \mu(\mathcal{B}_u(R) \cap \disk) = \mu(\mathcal{B}_0(R-r)) + \int_{R-r}^R \int_{-\theta_R(x,r)}^{\theta_R(x,r)}\rho(x) \dd\theta dx.
\end{align*}
By \cite[Corollary 5]{bks-hudg-23} we have  $\sqrt{e^{(R-r-x)}} \leq \theta_R(r,x) \leq \pi \sqrt{e^{(R-r-x)}}$ for $x+r\geq R$ since $r\geq 1$ by the hypothesis of our statement. Thus, for every $x \in [R-r, R)$, there exists a $\tau$ with $1 \leq \tau \leq \pi$ such that  
\begin{align*}
     \mu(\mathcal{B}_u(R) \cap \disk) &=\mu(\mathcal{B}_0(R-r)) + 2 \int_{R-r}^R \tau e^{\frac{R-r-x}{2}}\frac{\alpha\sinh(\alpha x)}{2\pi(\cosh(\alpha R) - 1)} \dd x\\
    &= \mu(\mathcal{B}_0(R-r)) + \frac{  \alpha e^{(R-r)/2}}{\pi(\cosh(\alpha R) - 1)}\int_{R-r}^{R} \tau e^{-x/2}\sinh(\alpha x) \dd x.
\end{align*}    

We proceed by using $\pi$ as an upper bound on $\tau$. The lower bound follows analogously by using $\tau \geq 1$. We get for the integral above using $\tau \leq \pi$

\begin{align*}
    \int_{R-r}^{R}\pi e^{-x/2}\sinh(\alpha x) \dd x &= \left[\frac{2\pi}{4\alpha^2 - 1}e^{-x/2}(2\alpha \cosh(\alpha x ) + \sinh(\alpha x))\right]_{R-r}^R\\
    &= \frac{2\pi}{4\alpha^2 -1}\left(e^{-R/2}(2\alpha \cosh(\alpha R) + \sinh(\alpha R))\right)\\
    &-\frac{2\pi}{4\alpha^2 -1}\left(e^{-(R-r)/2}(2\alpha \cosh(\alpha (R-r)) + \sinh(\alpha (R-r))\right).
\end{align*}

We then obtain by similar calculations to \cite[Lemma 3.2]{gpp-rhg-12}
\begin{align*}
\frac{  \alpha e^{(R-r)/2}}{\pi(\cosh(\alpha R) - 1)}\int_{R-r}^{R} \tau e^{-x/2}\sinh(\alpha x) \dd x \leq \frac{(1+o(1))\alpha e^{-r/2}}{\alpha - 1/2},  
\end{align*}

and then using \Cref{lem:measure-inner-disk}

\begin{align*}
    \mu(\mathcal{B}_u(R) \cap \disk) &=\mu(\mathcal{B}_0(R-r)) + 2 \int_{R-r}^R \tau e^{\frac{R-r-x}{2}}\frac{\alpha\sinh(\alpha x)}{2\pi(\cosh(\alpha R) - 1)} \dd x\\
    &\leq (1+o(1))e^{-\alpha r} + \frac{(1+o(1))\alpha e^{-r/2}}{\alpha - 1/2}\\
    &\leq \frac{(1+o(1))\alpha e^{-r/2}}{\alpha - 1/2},
\end{align*}

since $\alpha > 1/2$. The desired bound follows since $\deg(u)$ is a Poisson random variable so the expected vertices in the area is $n \cdot  \mu(\mathcal{B}_u(R) \cap \disk)$. The lower bound follows by doing the same calculations using $\tau \geq 1$.

\end{proof}

In the original work by Gugelmann, Konstantinos and Peter~\cite{gpp-rhg-12} the statement for the maximum degree was given without providing a definition for \whp This is not meant as a criticism to the authors as their proof is correct in the sense that it shows an \aas guarantee according to our definitions. In any case, to avoid any possible confusion, we provide a proof for our stated probabilistic guarantees.

\maxDegree*

\begin{proof}
    We start by showing the lower bound. Let $r_{high} := (2-1/\alpha)\log(n) + 2\log(\log(n))/\alpha$. Using \Cref{lem:measure-inner-disk} we obtain for the expected number of vertices in $\B_0(r_{high})$ that 
    $$n\cdot\mu(\B_0(r_{high})) = n(1-o(1))e^{-\alpha\left(R-r_{high}\right)} = \Theta(1)\log^2(n).$$

Using a Chernoff bound we have $|V \cap \B_0(r_{high})| \in \Theta(\log^2(n))$ \wehp Next we lower bound the vertex degree for a vertex $u \in V \cap \B_0(r_{high})$. Using \Cref{lem:vertex-degree} it holds
$$
\E{\deg(u)} \geq \Theta(1)e^{\left(R-r_{high}\right)/2} = \Theta(1)n^{\frac{1}{2\alpha}}\cdot\log^{-1/\alpha}(n).
$$

Then, using another Chernoff bound yields $\deg(u) \geq \Theta(1)n^{\frac{1}{2\alpha}}\cdot\log^{-1/\alpha}(n) \> n^{\frac{1}{2\alpha}}\log^{-2}(n)$ \wehp what finishes the lower bound as $\alpha > 1/2$.

For the upper bound we set $r_{low}:= (2-1/\alpha)\log(n) - \log(\log(n))/\alpha$ and use \Cref{lem:measure-inner-disk} to show that the expected number of vertices in $\B_0(r_{low})$ is
$$
n\cdot\mu(\B_0(r_{low})) = n(1-o(1))e^{-\alpha\left(R-r_{low}\right)} =\Theta(1)(1/\log(n)) \in o(1).
$$

Using Markov's inequality there is no vertex in the area $\B_0(r_{low})$ \aas Let us write $A$ for the event that $\{V \cap \B_0(r_{low}) = \emptyset\}$. Then, since $\mu(B_r(R) \cap \disk)$ is monotonically decreasing in $r$ (see \cite[Lemma 3.3]{gpp-rhg-12}), we obtain by using \Cref{lem:vertex-degree} that the expected degree for a vertex $u \in V \cap \disk$ conditioned on event $A$ is  

$$
\E{\deg(u)|A} \leq n\cdot\mu(\B_{r_{low}}(R)\cap \disk)) = \Theta(1)e^{(R-r_{low})/2} = \Theta(1)(n\cdot\log(n))^{\frac{1}{2\alpha}}.
$$

Using a Chernoff bound and $\alpha > 1/2$ it follows $\deg(u) \leq \Theta(1)(n\cdot\log(n))^{\frac{1}{2\alpha}} < n^{\frac{1}{2\alpha}}\cdot\log(n)$ \wehp conditioned on the event $A$. The proof for the upper bound is complete by a union bound over all vertices over vertices with radii larger than $r_{\text{low}}$, where the probabilistic guarantee stems from the fact that event $A$ occurs \aas

\end{proof}

\section{Proof of Leaves Lemma (Lemma~\ref{lem:leaf})}\label{sec:leavesproof}

\leavesLemma*

\begin{proof}
The idea goes as follows: first, we consider the set of vertices $U = V \cap \B_0(R-\startlayer-1)$, i.e., the set of vertices with smaller radius than the set $V_{\startlayer}$. For $u \in U$, we then write $N_0(u)$ for the set of vertices $V_0 \cap N(u)$ and let $\xi(u)$ be the largest angle that is spanned among a pair of vertices in $N_0(u)$ (see \Cref{fig:leaves}a). This is useful in the following sense: let $\Xi(u)$ be the sector that is defined by bisector $\varphi(u)$ and angle $\xi(u)$. Then whenever we consider a vertex $v \in V \cap (\mathcal{A}_0 \setminus \Xi(u))$, we know that $v$ has no edge to $u$. Thus, an upper bound on the angles $\sum_{u \in U} \xi(u) =: \xi$ gives us the size of the layer $\mathcal{A}_0$ on which the probability of events for vertices $V \cap \mathcal{A}_0$ are independent of $U$. Then let $\Xi := \cup_{u \in U}\Xi(u)$ and in a second step we operate on the sub area $\mathcal{A}_0\setminus\Xi$. We show the desired properties by revealing enough leaves in $\mathcal{A}_0\setminus\Xi$ that have only a single vertex $u \in V_{\startlayer}$ as a neighbour.

To bound the angle $\xi$ recall that $\Xi := \cup_{u\in U}\Xi(u)$ where for $u \in U$ the area $\Xi(u)$ is the sector with bisector $\varphi(u)$ and angle $\xi(u)$. Using \Cref{lem:max-angle} we obtain $\xi(u) \leq 2\theta_R(r(u), R-1)$ since $\theta_R(\cdot,\cdot)$ is monotonically decreasing in both arguments. Then an upper bound for the angle of $\Xi$ is given by summing over all upper bounds of the angles of all $\Xi(u)$. This yields

$$\xi = \sum_{u \in U}{\xi(u)} \leq \sum_{u \in U}{2\theta_R(r(u), R-1)}.$$

 Moreover let $A_\finallayer$ be the layer of largest index $\finallayer$ that is non-empty. Since the set $U$ are the set of vertices located layers larger than $\startlayer$ and a vertex in layer $\mathcal{A}_\ell $ has radius at least $R - \ell - 1$, we obtain  via $\theta_R(\cdot, \cdot)$ being monotonically decreasing in its two arguments

 $$
 \xi \leq \sum_{\ell = \startlayer + 1}^{\finallayer}\sum_{v \in V_\ell}2\theta_R(R-\ell-1,R - 1) \in \sum_{\ell = \startlayer + 1}^{\finallayer}\bigO(\max(ne^{-\alpha \ell}, \log^2(n))\cdot\theta_R(R-\ell-1,R - 1),
 $$

 where we used \Cref{lem:layer-properties} by which $|V_\ell| \in \bigO(\max(ne^{-\alpha \ell}, \log^2(n))$ \wehp Next let us write $\ell_0 := \lceil \alpha^{-1}(\log(n) - 2\log(\log(n)))\rceil$. Then we use  \Cref{lem:layer-properties} by which $\finallayer \leq \lceil\alpha^{-1}\cdot(\log(n) + \log(\log(n)) \rceil$ \aas and obtain

$$
\xi  \in \sum_{\ell = \startlayer + 1}^{\ell_0}\mathcal{{\bigO}}\left(e^{-\ell(\alpha-1/2)}\right) + \sum_{\ell_0 + 1}^{\finallayer} \bigO\left(\log^{2 + \frac{1}{2\alpha}}(n)\cdot n^{-\frac{2\alpha - 1}{2\alpha}}\right),
$$

where we used \Cref{lem:max-angle} to obtain the upper bound $\theta_R(R-\ell-1, R-1) \leq \Theta(1)e^{\ell/2 - \log(n)}$ which is monotonically increasing in $\ell$ for $\theta_R(R- \ell -1, R-1)$. Using that there are $\bigO(\log(n))$ layers and $\alpha > 1/2$ it follows that

$$
\xi  \in \mathcal{{\bigO}}\left(\log(n)\cdot e^{-\startlayer(\alpha-1/2)}\right) + o(1) \in o(1),
$$

where we used in the last step that $\startlayer\geq c \log(\log(n))$ by the hypothesis of our lemma. Thus, for $c \geq 2(\alpha - 1/2)^{-1}$ the term is vanishing and we conclude $\xi \in o(1)$ \aas 

Now, we partition the rest disk $\mathcal{D}_R\setminus\Xi$ into $k = \left\lceil ne^{-\alpha\startlayer} \right\rceil$ disjoint sectors such that each has angle $\phi = (2\pi-\xi)/k = \Theta(1)n^{-1}e^{\alpha\startlayer}$. Let us write $\Phi_i$ for $i$-th sector with angle $\phi$ and we denote the intersection of a sector and the layer $\mathcal{A}_{\startlayer}$ by $\mathcal{A}_{\startlayer} \cap \Phi_i =: \mathcal{C}_i$ (see \Cref{fig:leaves}b).
Fixing any area $\mathcal{C}_i$, we obtain for its measure 
\begin{equation}\label{eq:non-vanishing-probability}
    \mu(\mathcal{C}_i) = \mu(\mathcal{A}_{\startlayer} \cap \Phi_i) = \mu(\B_0(R-\ell)\setminus\B_0(R-\ell-1)\cdot \phi = \Theta(1)e^{-\alpha\startlayer}\cdot n^{-1}e^{\alpha\startlayer} \in \Theta(1/n),
\end{equation}
where we used \Cref{lem:measure-inner-disk} and $\phi = \Theta(1)n^{-1}e^{\alpha\startlayer}$. We conclude that $\E{|V \cap \mathcal{C}_i|} \in \Theta(1)$. Since the distribution of the vertices follows Poisson point distribution we get for any $i \in k$ a non-vanishing probability for the event that there is exactly one vertex in $\mathcal{C}_i$. 

In the next step, we count the number of vertices in $\mathcal{A}_{\startlayer}$ that have constant fraction of its neighbour as leaves. To this end, let $A_i$ be the event that a vertex $u \in V \cap \mathcal{C}_i$ has $\Theta(\deg(v))$ leaves in layer $\mathcal{A}_0$. Recall that $\startlayer \geq c\log(\log(n)) \in \omega(1)$ by which we get with \Cref{lem:leaves} in conjunction with \Cref{eq:non-vanishing-probability} using law of total expectation $\E{A_i} \geq \E{A_i | |\mathcal{C}_i| = 1}\Prob{|\mathcal{C}_i| = 1} \in \Omega(1)$.

Now recall that we partitioned $\mathcal{D}_R\setminus\Xi$ into $k$ sectors and define $A = \sum_{i}^{\lfloor k/2\rfloor} \indicator{A_i}$. We note that $A$ is crude lower bound on the desired random variable we are looking for in our lemma as $A_i$ is the event that $\mathcal{C}_i$ contains a vertex where a constant fraction of its neighbours are leaves while only located in $\mathcal{A}_0\setminus \Xi$. By linearity of expectation we obtain 
\begin{align}\label{eq:lin-exp}
    \E{A} = \Theta(k) \E{A_i} \in \Omega(k).
\end{align}
We proceed by showing independence between $A_i$ and $A_j$ for $|i-j|>1$. First observe that $\Phi_i \cap \Phi_j = \emptyset$. Moreover recall that $\phi = \Theta(1)n^{-1}e^{\alpha\startlayer}$, such that the leaves located in $\B_0(R)\setminus\B_0(R-1)$ have a neighbour in $\mathcal{D}_R\setminus\Xi$ with angle distance at most $\theta_R(R-\startlayer-1, R-1) \leq \Theta(1)n^{-1}e^{\eta/2} \in o(\phi)$ using \Cref{lem:max-angle}. Finally recall that we only operate on the sub disk $\mathcal{D}_R\setminus\Xi$ where vertices in layer $\mathcal{A}_0$ do not have a neighbour with radius at most $R - \startlayer - 1$. It follows that the area of leaves are independent and thus, also the events $A_i$ and $A_j$ for $|i-j|>1$.

To finish the proof recall that $\startlayer \leq \alpha^{-1}(\log(n) - \log(\log(n)) -c)$ and thus, $k \geq c\cdot\log(n)$ which implies $\E{A} \geq \Theta(1)c \cdot \log(n)$ by \Cref{eq:lin-exp}. A Chernoff bound yields $\Prob{A \in o(k)} \leq n^{-c'}$ for any constant $c'>$ given that $c$ was chosen large enough. Another Chernoff bound gives us in conjunction with \Cref{lem:measure-inner-disk} that $\Prob{|V_{\startlayer}| \in \omega(k)} \leq n^{-c'}$ for any constant $c'>0$. A union bound over these two events and another union bound over $\log n$ layers concludes the proof where the probabilistic guarantee \aas stems from the fact that we conditioned on the event that $\finallayer \leq \lceil \alpha^{-1}(\log(n) + \log(\log(n)))\rceil$ to bound $\xi$ by $o(1)$ which occurs \aas 
\end{proof}

\section{Concentration bounds}\label{item:prob-guarantees}


To obtain probabilistic guarantees, we use the following concentration bounds (see, e.g. \cite{Prob-Method-16, mu-pc-05}).

\begin{lemma}[Markov's inequality]\label{lem:markov}
    Let $X$ be a random variable with expectation $\Exp{X}$. Then for $a > 0$, $\Pro{X \geq a} \leq a^{-1}\Exp{X}$.
\end{lemma}

\smallskip

\begin{theorem}[Chernoff bound]\label{lem:Chernoff}
For $i \in [k]$, let $X_i \in \{0,1\}$ be independent random variables and $X = \sum_i X_i$. Then
\begin{align*}
    \Pro{X\geq \frac{3}{2}\Exp{X}} \leq \exp{\left(-\frac{\Exp{X}}{12}\right)} \text{ and }
    \Pro{X\leq \frac{1}{2}\Exp{X}} \leq \exp{\left(- \frac{\Exp{X}}{8}\right)}.
\end{align*}
\end{theorem}

The following is a Chernoff\footnote{For convenience, we refer to both as a Chernoff bound whenever using either of the two.} type deviation bound (see e.g. \cite[Lemma 6]{Kiwi2024}) which is necessary for the distribution of vertices as it follows a Poisson point distribution.

\begin{lemma}[Poisson Chernoff bound]\label{lem:Poisson-Chernoff}
    Let $X$ have a Poisson distribution with mean $\E{X}$. Then for $a \geq e^{3/2}$,
    \begin{align*}
        \Prob{X\leq \frac{1}{2}\mathbb{E}[X]} \leq \exp{(- \mathbb{E}[X]/8)} \text{ and }
    \Prob{X\geq a\mathbb{E}[X]} \leq \exp{(a\cdot \mathbb{E}[X] /2)}.
    \end{align*}
\end{lemma}

Finally, for analysing the \RCTID algorithm, we use the following variant of an Azuma-Hoeffding bound (see e.g. \cite[Corollary A.5]{BEPS}).

\begin{lemma}[Azuma-Hoeffding]\label{lem:BEPS}
    Let $Z = Z_1 + \ldots + Z_n$ be the sum of n random variables and $X_0,\ldots, X_n$ be a sequence, where $Z_i$ is uniquely determined by $X_0,\ldots, X_i$, $\Exp{Z_i| X_0,\ldots, X_{i-1}}$, $\Exp{Z} = \sum_i \Exp{Z_i}$, and $a_i \leq Z_i \leq a'_1$. Then
    $$
    \Pro{Z \geq \Exp{Z} + a}\leq \exp{\left(-\frac{a^2}{2\sum_i(a'_i - a_i)^2}\right)}.
    $$
\end{lemma}

\end{document}